\documentclass[12pt,reqno]{amsart}
\usepackage{graphicx,amsmath,amsfonts,amssymb}
\usepackage{psfrag}
\usepackage{epsf}
\usepackage{latexsym}
\usepackage{rotating}

\newtheorem{theorem}{Theorem}

\newtheorem{lemma}{Lemma}

\newtheorem{remark}{Remark}

\newcommand{\ve}{\vspace{1ex}}

\begin{document}

\title[Periodic travelling waves and compactons
in granular chains]{Periodic travelling waves and compactons in granular chains}
\author{Guillaume James}
\address{Laboratoire Jean Kuntzmann, Universit\'e de Grenoble and CNRS,
BP 53 \\ 
38041 Grenoble Cedex 9, France.}
\email{Guillaume.James@imag.fr}
\date{\today}

\maketitle

\begin{center}
Laboratoire Jean Kuntzmann,\\
Universit\'e de Grenoble and CNRS,\\
BP 53, 38041 Grenoble Cedex 9, France.
\end{center}

\begin{abstract}
We study the propagation of an unusual type of periodic travelling waves in
chains of identical beads interacting via Hertz's contact forces.
Each bead periodically undergoes a compression phase followed by a free flight,
due to special properties of Hertzian interactions
(fully nonlinear under compression and vanishing in the absence of contact).
We prove the existence of such waves close to binary oscillations, and
numerically continue these solutions when their
wavelength is increased. In the long wave limit,
we observe their convergence towards shock profiles
consisting of small compression regions close to solitary waves,
alternating with large domains of free flight
where bead velocities are small. 
We give formal arguments to justify this
asymptotic behaviour, using a matching
technique and previous results concerning
solitary wave solutions. 
The numerical finding of such waves
implies the existence of compactons, i.e. compactly supported
compression waves propagating at a constant velocity, depending on the 
amplitude and width of the wave. The beads are stationary and
separated by equal gaps outside the wave, and each bead reached
by the wave is shifted by a finite distance during a finite time interval.
Below a critical wavenumber,
we observe fast instabilities of the periodic travelling waves
leading to a disordered regime.
\end{abstract}


\section{\label{intro}Introduction and main results}

Understanding wave propagation in granular media is
a fundamental issue in many contexts, e.g. 
to design shock absorbers \cite{sen,frater}, 
derive multiple impact laws \cite{hinch,acary,ma,liu1,liu2},
detect buried objects by acoustic sensing \cite{senm},
or understand possible dynamical mechanisms of earthquake triggering \cite{johnson}.
One of the important factors that influence wave propagation 
is the nature of elastic interactions between grains. 
According to Hertz's theory, the repulsive force $f$
between two identical initially tangent spherical beads 
compressed with a small relative displacement $\delta$
is $f(\delta )=k\, \delta^{\alpha}$ at leading order in $\delta$,
where $k$ depends on the ball radius and material properties and
$\alpha = 3/2$ (see figure \ref{boules}). 
This result remains valid for much more general geometries
(smooth non-conforming surfaces), and $\alpha$ can be even larger in the
presence of irregular contacts \cite{johnsonbook,fu}. 
The Hertz contact force has several properties that make the analytical study
of wave propagation more difficult than in classical systems of interacting particles :
the dependency of $f(\delta )$ on $\delta \approx 0$
is fully nonlinear for $\alpha >1$, $f^{\prime\prime}(0)$ is not defined
for $\alpha <2$, and no force is present when beads are not in contact. 
This makes the use of perturbative methods rather delicate,
since the latter often rely on nonlinear modulation of linear waves 
which are not present in such systems, and
usually require higher regularity of nonlinear terms.

\begin{figure}[h]
\begin{center}
\includegraphics[scale=0.1]{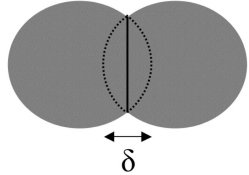}
\end{center}
\caption{\label{boules} 
Schematic representation of 
two identical and initially tangent spherical beads 
that are compressed and slightly flatten, 
the distance between their centers decreasing by $\delta \approx 0$.
}
\end{figure}

\ve

The simplest model in which these difficulties show up consists of a line
of identical spherical beads, in contact with their neighbours at a single point
when the chain is at rest. 
For an infinite chain of beads, the dynamical equations 
read in dimensionless form
\begin{equation}
\label{nc}
\frac{d^2 x_{n}}{dt^2} = 
V^\prime(x_{n+1}-x_n)-V^\prime(x_{n}-x_{n-1}),
\ \ \
n\in \mathbb{Z},
\end{equation}
where $x_{n}(t)\in \mathbb{R}$ is
the displacement of the $n$th bead from a reference position and
the interaction potential $V$ corresponds to Hertz contact forces.
It takes the form
\begin{equation}
\label{vhertz}
V(x)=\frac{1}{1+\alpha}\, |x|^{1+\alpha}\, H(-x) ,
\end{equation}
where $H$ denotes the Heaviside function vanishing on $\mathbb{R}^-$ and
equal to unity on $\mathbb{R}^+$ and $\alpha >1$ is a fixed constant.
System (\ref{vhertz}) is Hamiltonian with total energy
\begin{equation}
\label{ham}
{\mathcal H}=
\sum_{n\in \mathbb{Z}}{\frac{1}{2}\, (\frac{dx_{n}}{dt})^2+V(x_{n+1}-x_n)}
.
\end{equation}
For $\alpha = 3/2$, Nesterenko analyzed the propagation of compression
pulses in this system using a formal continuum limit and found approximate
soliton solutions with compact support \cite{neste1,neste2}
(see also \cite{ap,porter,sen} for more recent results and references).
As shown by MacKay \cite{mackay}
(see also \cite{ji}), exact solitary wave solutions of (\ref{nc}) exist since an existence 
theorem of Friesecke and Wattis \cite{friesecke} can be
applied to the chain of beads with Hertz contact forces
(see also \cite{stef} for an alternate proof). Moreover
these solitary waves have in fact a doubly-exponential decay \cite{chat,english,stef}
that was approximated by a compact support in Nesterenko's analysis.

\ve

Much more analytical results are available
on nonlinear waves in the granular chain 
when an external load $f_0$ is applied at both ends and
all beads undergo a small compression $\delta_0$ when
the system is at rest.
Around this new equilibrium state, the dynamical equations
correspond to a Fermi-Pasta-Ulam (FPU) lattice \cite{cam,gal,pankovbook},
that sustains in particular stable solitary waves with exponential decay
(see e.g. \cite{friesecke,pego,iooss,pego4,hoffman3}) 
and periodic travelling waves 
also known as ``nonlinear phonons"
\cite{filip,iooss,dreyer,pankovbook,herrmann}.
The existence and qualitative properties of
periodic travelling waves in granular chains are important elements in
the understanding of energy propagation and dispersive shocks 
in these systems, as shown in reference \cite{dreyer2} in the similar
context of FPU lattices. However,
the persistence of periodic travelling waves 
when $f_0 \rightarrow 0$ is not obvious because the sound velocity
(i.e. the maximal velocity of linear waves) vanishes
as $f_0^{1/6}$, and the uncompressed chain of beads is commonly
denoted as a ``sonic vacuum"  \cite{neste2}.

\ve

In this paper we show that periodic travelling waves with unusual properties exist
in granular chains under Hertz contacts without precompression. These waves
consist of packets of compressed beads alternating with 
packets of uninteracting ones, so that any given bead 
periodically switches between a free flight regime
and a contact regime where it interacts nonlinearly with
two or one neighbours.
Waves of this type have been numerically computed in reference \cite{sv}, for
periodic granular chains consisting of three or four beads.
Here we provide an existence theorem valid at wavenumbers close
to $\pi$ and proceed by numerical continuation when the wavenumber
goes from $\pi$ to the long wave limit $q\rightarrow 0$.
The limit $q=\pi$ corresponds to binary oscillations, i.e.
travelling waves with spatial period $2$ where nearest neighbours oscillate out of phase.
In the long wave limit, the periodic waves display small compression regions close to solitary waves,
alternating with large domains of free flight where bead velocities are small. 

\ve

Our existence result for periodic travelling waves close to binary oscillations
is described in the following theorem. 
In what follows, 
$C^k_{\rm{per}}(0,2\pi )$ denotes the classical
Banach space of $2\pi$-periodic and $C^k$ functions 
$u\, : \mathbb{R}\rightarrow \mathbb{R}$, endowed with 
the usual supremum norm taking into account
all derivatives of $u$ up to order $k$.

\begin{theorem}
\label{existence}
System (\ref{nc}) with interaction potential (\ref{vhertz})
admits two-parameter families of periodic travelling wave solutions
\begin{equation}
\label{twthm}
x_n^{\pm} (t; a,\mu )=a\,  u_\mu [(\pi + \mu )\, n \pm a^{\frac{\alpha -1}{2}} t ],
\end{equation}
parametrized by $a>0$ and $\mu \in {\mathcal V}$, where
${\mathcal V}$ denotes an open interval containing $0$. 
The function $u_\mu$ is odd, $2\pi$-periodic and satisfies the advance-delay differential equation
\begin{equation}
\label{adthm}
u_\mu^{\prime\prime} (\xi ) =
V^\prime (u_\mu (\xi +\pi + \mu)-u_\mu(\xi ))-V^\prime (u_\mu(\xi )-u_\mu(\xi -\pi - \mu)).
\end{equation}
The map $\mu \mapsto u_\mu$ belongs to $C^1 ({\mathcal V} , C^2_{\rm{per}}(0,2\pi ))$
and satisfies the symmetry property $u_{\mu}(\xi )=-u_{-\mu}(\xi +\pi)$.
The function $u_0$ is determined by the initial value problem
\begin{equation}
\label{int}
u_0^{\prime\prime} + W^\prime (u_0 )=0,
\end{equation}
\begin{equation}
\label{ci}
u_0 (0)=0, \ \ \
u_0^{\prime}(0)=p_0, \ \ \
p_0= 
(1+\alpha)^\frac{\alpha}{1-\alpha}
\Big[ \frac{\sqrt{\pi}\,  \Gamma{(\frac{1}{1+\alpha}+\frac{1}{2})}  }{ \Gamma{(\frac{1}{1+\alpha})}  } \Big]^\frac{1+\alpha}{1-\alpha} ,
\end{equation}
where 
$W(x)= \frac{2^\alpha}{1+\alpha}|x|^{1+\alpha}$
is a symmetrized Hertz potential and
$\Gamma{(x)}=\int_{0}^{+\infty}{e^{-t}\, t^{x-1}\, dt}$ denotes Euler's Gamma function.
Moreover, for $\mu <0$, $u_\mu$ is a linear function of $\xi$ on
an interval $[-\xi_1 (\mu ) , \xi_1 (\mu )]$ with
$\xi_1 (\mu )= |\mu |/2 + o(\mu )$. It takes the form
\begin{equation}
\label{linthm}
u_\mu (\xi ) = p_\mu \, \xi \mbox{ for all } \xi \in [-\xi_1 (\mu ) , \xi_1 (\mu )],
\end{equation}
with $p_\mu = p_0 +O(|\mu |)$.
When $\xi = (\pi + \mu )\, n \pm a^{\frac{\alpha -1}{2}} t 
\in (-\xi_1(\mu) + 2 k \pi,\xi_1 (\mu )+ 2 k \pi)$ ($k\in \mathbb{Z}$), 
the $n$th particle performs a free flight characterized by 
$$
x_{n+1}^{\pm} (t; a,\mu )>x_n^{\pm} (t; a,\mu ),\ \ \ x_{n}^{\pm} (t; a,\mu )>x_{n-1}^{\pm} (t; a,\mu ), \ \ \ 
\dot{x}_n^{\pm} (t; a,\mu )=\pm a^{\frac{\alpha +1}{2}}\, p_\mu. 
$$
\end{theorem}

To prove theorem \ref{existence}, we locally solve equation (\ref{adthm}) 
using the implicit function theorem, where we have to pay a particular attention
to regularity issues due to the limited smoothness of the interaction potential $V$.
This approach is reminiscent of 
the previous work \cite{jamesc}, where we proved the existence of 
periodic travelling waves in the Newton's cradle, a mechanical system in which the granular chain
is modified by attaching each bead to a local linear pendulum. 
However, in that case the
periodic travelling waves were obtained by nonlinear modulation 
of linear oscillations in the local potentials, a limit which does not exist in the present situation.

\ve

A complementary approach to the analysis of exact periodic travelling waves
consists of obtaining approximate solutions described by continuum models. 
Approximate periodic travelling waves have 
been previously studied through a formal continuum limit in the granular chain
(see \cite{neste2}, sections 1.3-5, \cite{falcon}, section 4.6 and references therein),
but in a case where all beads were interacting under compression.
Closer to our case, Whitham's modulation equation may be used to 
approximate the solutions of theorem \ref{existence} and study
their stability properties, following the lines of reference \cite{dreyer} 
where the example of a granular chain
with linear contact interactions was analyzed. 

\ve

A limitation of theorem \ref{existence} stems from the fact that
it provides an existence result for wavenumbers $q=\pi + \mu$
close to $\pi$. To analyze the existence of periodic travelling
waves on a full range of wavenumbers, 
we numerically continue (for $\alpha=3/2$) the solutions of
theorem \ref{existence} by decreasing the wavenumber $q$
down to values close to $0$. 
We compute the nonlinear dispersion relation
associated to periodic travelling waves, and analyze their asymptotic
form in the long wave limit $q \approx 0$. In this limit 
and when the wave velocity
is normalized to unity, one finds large ensembles of beads performing a free flight
at small velocity, separated by smaller regions consisting of $4$-$5$ 
balls under compression with relative displacements close to a solitary wave.
Using a matching technique, we perform a formal asymptotic study of the limit $q\rightarrow 0$
which provides a limiting wave profile in good agreement with our numerical computations.

\ve

These results are completed by a stability analysis. We
compute the Floquet spectrum modulo shifts of the periodic travelling waves,
and find fast linear instabilities below a critical wavenumber $q_c \approx 0.9$.
This threshold corresponds to the case when the average number of interacting beads
in the compression regions becomes larger or equal to three, 
or equivalently to the maximal number of adjacent interacting beads becoming $\geq 4$.
In this regime, 
the initial numerical errors made on the travelling wave profiles are amplified
during dynamical simulations, leading to the disappearance of the travelling
waves after some transcient time. A disordered regime is established shortly after the instability, 
but interestingly some partial order can be still observed in the form of intermittent large-scale 
organized structures. 
In contrast to the above situation, long-time dynamical simulations of periodic travelling
waves with wavenumbers $q>q_c$ yield wave profiles that
remain practically unchanged when propagating along the lattice.
The existence of very slow instabilities in this regime is a more delicate
question which will be addressed in a future work.

\ve

Besides periodic waves, our numerical results also demonstrate the existence of
compactons (solitary waves with compact supports)
in chains of beads separated by equal gaps outside the compression wave. 
For $q\leq q_k \approx 1.8$, the periodic solutions of (\ref{adthm}) obtained numerically
behave linearly on intervals of length larger than the delay $q=\pi + \mu$ involved in (\ref{adthm}), and these solutions
can be linearly extended in order to get new solutions of (\ref{adthm}). 
Using this property and the Galilean invariance of
(\ref{nc}), we obtain compacton solutions for which 
the beads are stationary and
separated by equal gaps outside the wave, and each bead reached
by the wave is shifted by a finite distance during a finite time interval.

Our situation is quite different to what occurs in the absence of gaps between beads, 
because in that case compactons exist in continuum models of granular chains \cite{neste1,neste2,ap},
but a transition from compactons to noncompact (super-exponentially localized) solutions occurs when
passing from the continuum model to the discrete lattice \cite{ap}. In our case the existence of
compactons is linked with the occurence of free flight, made possible
by the unilateral character of Hertzian interactions which vanish when beads are not in contact.

\ve

The outline of the paper is as follows. In section \ref{localcont} we prove
theorem \ref{existence} and study some qualitative properties of the
travelling wave profiles. The numerical results are presented in section \ref{num}, which contains in addition a formal
asymptotic study of the long wave limit and the discussion of the existence of compactons. 
Section \ref{disc} gives a summary of the main results and mentions implications for other works. 
Lastly, some useful properties of solitary wave solutions
are recalled in an appendix.

\section{\label{localcont}Local continuation of periodic travelling waves}

Periodic travelling wave
solutions of (\ref{nc}) take the form
\begin{equation}
\label{tw}
x_n (t)= u(\xi ), \ \ \ \xi =q\, n -\omega\, t ,
\end{equation}
where $u$ is $2\pi$-periodic, $q \in [0 , 2\pi )$ denotes the wavenumber,
$\omega \in \mathbb{R}\setminus \{0\}$ the wave frequency and 
$\xi \in \mathbb{R}$ the spatial coordinate in a frame moving with the wave.
Thanks to a scale invariance of (\ref{nc}),
the full set of periodic travelling waves can be deduced from
its restriction to $\omega =1$. Indeed, due to the special form
of the Hertz potential (\ref{vhertz}), any solution $x_n$ of (\ref{nc})
generates two families of solutions 
$x_n^{(a)} (t)= a \, x_n (\pm\, a^{\frac{\alpha -1}{2}}\, t)$ parametrized by $a >0$.
Consequently, any solution of the form (\ref{tw}) with $\omega =1$
corresponds to two families of periodic travelling waves propagating in opposite directions
\begin{equation}
\label{twfam}
x_n^{(a)} (t)=a\, u(q\, n \mp a^{\frac{\alpha -1}{2}} t ),
\end{equation}
with frequency $\omega = \pm\, a^{\frac{\alpha -1}{2}}$.

Fixing $\omega =1$, equation (\ref{nc}) yields the advance-delay differential equation
\begin{equation}
\label{ad}
u^{\prime\prime} (\xi ) =
V^\prime (u(\xi +q)-u(\xi ))-V^\prime (u(\xi )-u(\xi -q)), \ \ \
\xi \in \mathbb{R},
\end{equation}
with periodic boundary conditions
\begin{equation}
\label{pbc}
u(\xi + 2\pi )=u(\xi ).
\end{equation}
Equation (\ref{ad}) possesses the symmetry $u(\xi ) \rightarrow -u({-\xi})$
originating from the reflectional and time-reversal symmetries of (\ref{nc}).
In the sequel we restrict our attention to
solutions of (\ref{ad}) invariant under this symmetry (i.e. odd in $\xi$).
This assumption simplifies our continuation procedure, because
it eliminates a degeneracy of periodic travelling waves linked with 
the invariance of (\ref{ad}) under translations and phase-shift.

\subsection{\label{bin}Binary oscillations}

The following lemma exhibits a particular solution of (\ref{ad})-(\ref{pbc})
obtained for $q=\pi$ and corresponding to binary oscillations in system (\ref{nc}).

\begin{lemma}
\label{binary}
For $q=\pi$, there exists an odd solution $u_0$ of (\ref{ad})-(\ref{pbc}) determined by
the initial value problem (\ref{int})-(\ref{ci}), and satisfying
\begin{equation}
\label{sym}
u_0(\xi + \pi )=-u_0(\xi ).
\end{equation}
In addition, there exists a one-parameter family of
solutions of (\ref{nc}) taking the form
\begin{equation}
\label{binfam}
x_n (t)=a\, (-1)^{n+1} \, u_0( a^{\frac{\alpha -1}{2}} t  ), \ \ \ a\in (0,+\infty ).
\end{equation}
\end{lemma}

\begin{proof}
We look for solutions of (\ref{ad}) satisfying (\ref{sym}), so that
the advance-delay differential equation reduces
to the ordinary differential equation (\ref{int})
with $W(x)=\frac{1}{2} V(-2|x|)= \frac{2^\alpha}{1+\alpha}|x|^{1+\alpha}$.
Now we must find a $2\pi$-periodic solution $u_0$ of (\ref{int})
satisfying (\ref{sym}).
Equation (\ref{int}) is integrable since 
$I=\frac{1}{2}(u^\prime )^2 + W(u)$ is constant along any solution $u$,
and its phase-space is filled by periodic orbits. 
Any solution $u$ of (\ref{int}) with initial condition
$u (0)=0$, $u^\prime (0)=p >0$ is odd in $\xi$ 
(due to the evenness of $W$) and its period $T(p)$ is given by 
$$
T(p)=4 \int_0^{u_{\rm{max}}}{(p^2 - 2 W(v))^{-1/2}\, dv}
$$
where $u_{\rm{max}}=\frac{1}{2}[p^2 (1+\alpha ) ]^{\frac{1}{1+\alpha}}$.
Our aim is to find $p_0   >0$ such that $T(p_0 )=2\pi$.
Using the change of variable $v=u_{\rm{max}}\, t^{ \frac{1}{1+\alpha} }$, one finds
$T(p)=2 (1+\alpha )^{\frac{-\alpha}{1+\alpha}}\, 
\beta \, p^{\frac{1-\alpha}{1+\alpha}}$, where 
$$
\beta = 
\int_0^1{t^{ \frac{1}{1+\alpha}-1}\, (1 - t)^{-1/2}\, dt} = 
B(\frac{1}{1+\alpha},\frac{1}{2}),
$$
and $B(z,w)=\frac{\Gamma{(z)}\, \Gamma{(w)}}{\Gamma{(z+w)}}$ 
denotes Euler's Beta function 
(see \cite{abra}, formula 6.2.1 p. 258). Since $\Gamma{(1/2)}=\sqrt{\pi}$,
one obtains finally 
\begin{equation}
\label{period}
T(p)=c_\alpha \, p^{\frac{1-\alpha}{1+\alpha}} , \ \ \
c_\alpha =2\sqrt{\pi} (1+\alpha )^{\frac{-\alpha}{1+\alpha}}\, 
\frac{\Gamma{(\frac{1}{1+\alpha})}}{\Gamma{(\frac{1}{1+\alpha}+\frac{1}{2})}}.
\end{equation}
Consequently, for $p$ varying in $(0,+\infty )$ the period
$T(p)$ is strictly decreasing from $+\infty$ to $0$,
and there exists a unique $p_0   >0$ such that $T(p_0 )=2\pi$.
For the corresponding solution $u_0$ of (\ref{int}), $u_0^\prime (0)=p_0$ 
is given by (\ref{ci}). 
Lastly, there remains to check property (\ref{sym}) in order to ensure that
$u_0$ satisfies (\ref{ad}). Using the fact that
$u_0(\pi)=0$ (since $u_0$ is odd and $2\pi$-periodic), we have
$u_0^\prime (\pi )=-p_0$ by conservation of $I$.
Consequently, $u_0(\xi + \pi )$ and $u_0(-\xi )$ are solutions of
the same Cauchy problem at $\xi =0$, which implies that
(\ref{sym}) is satisfied. Solutions (\ref{binfam}) are obtained using (\ref{twfam}).
\end{proof}

\begin{remark}
The solution $u_0$ of lemma \ref{binary} is non-unique because
$-u_0$ is also a solution, but this second choice
corresponds to shifting solution (\ref{binfam}) by one lattice site
(or equivalently performing a half-period time-shift). In addition,
the solutions of (\ref{int}) with period
$2\pi /(2k+1)$ ($k\in \mathbb{N}$) also satisfy (\ref{sym}), and thus they are
solutions of (\ref{ad})-(\ref{pbc}) for $q=\pi$.
However, these additional solutions do not yield any new
travelling wave because they can be recovered from
$u_0$ by tuning $a$ in (\ref{binfam}). Moreover,
the solutions of (\ref{int})-(\ref{pbc}) with period $\pi /k$ 
do not satisfy (\ref{sym}),
hence they do not satisfy (\ref{ad})-(\ref{pbc}) for $q=\pi$. 
\end{remark}

\subsection{\label{localconttw}Periodic travelling waves close to binary oscillations}

In order to locally continue the solution
$u_0$ of (\ref{ad})-(\ref{pbc}) for $q \approx \pi$ we need to
define a suitable functional setting. 
In the sequel we note $\mathbb{X}_k = C^k_{\rm{per}}(0,2\pi )$
and consider the closed subspace of $\mathbb{X}_k$
$$
X_k = \{\, 
u\in C^k_{\rm{per}}(0,2\pi ), \
u(-\xi )=-u(\xi )
\, \} .
$$
We denote by $\tau_q \in {\mathcal L}(\mathbb{X}_k)$
the shift operator $(\tau_q u )(\xi)=u(\xi + q)$
(note that $\tau_\pi$ maps $X_k$ into itself).
Problem (\ref{ad})-(\ref{pbc}) restricted to odd solutions can be rewritten
\begin{equation}
\label{pbnl}
u^{\prime\prime}+N(u,q)=0 \mbox{ in } X_0, \ \ \ u\in X_2,
\end{equation}
where the map $N(.,q)$ defined by
$$
N(u,q) =V^\prime ( (I-\tau_{-q}) u)-V^\prime ( (\tau_q-I) u )
$$
maps $X_0$ into itself.
The smoothness of $N$ restricted to $X_2 \times \mathbb{R}$
is proved in the following lemma.

\begin{lemma}
\label{c1}
The map $N$ belongs to $C^1 (X_2 \times \mathbb{R} , X_0)$ and
\begin{equation}
\label{diffu0}
D_u N(u_0 , \pi)=\frac{1}{2}\, W^{\prime\prime}(u_0)\, (I-\tau_\pi ).
\end{equation}
\end{lemma}

\begin{proof}
Since $V^\prime \in C^1 (\mathbb{R})$, the map
$u \mapsto V^\prime (u)$ belongs to $C^1 (\mathbb{X}_0)$
(this classical result follows from the uniform continuity of $V^{\prime\prime}$
on compact intervals). Consequently, to prove that $N$ is $C^1$ it suffices
to show that the map $G\, : \, \mathbb{X}_2 \times \mathbb{R} \rightarrow \mathbb{X}_0$ defined by
$G(u,q)=\tau_q u$ is $C^1$. Since the continuous bilinear mapping 
$\Pi \, : \, \mathbb{X}_2 \times \mathcal{L}(\mathbb{X}_2,\mathbb{X}_0) 
\rightarrow \mathbb{X}_0$, $(u,A)\mapsto A\, u$ is $C^\infty$,
it suffices to check that the map $q \mapsto \tau_q$ belongs to 
$C^1 (\mathbb{R} , \mathcal{L}(\mathbb{X}_2,\mathbb{X}_0))$. Taylor's formula yields
$$
\| \frac{\tau_{q+h}u - \tau_{q}u}{h} - \tau_q u^\prime  \|_{L^\infty}
\leq 
\frac{|h|}{2}\,  \| u^{\prime\prime} \| _{L^\infty}
$$
for all $u\in \mathbb{X}_2$ and $q,h\in \mathbb{R}$ with $h\neq 0$.
Consequently, $\lim\limits_{h\rightarrow 0}{\frac{1}{h}\, (\tau_{q+h} - \tau_{q})}=\tau_q \frac{d}{d\xi}$
in $\mathcal{L}(\mathbb{X}_2,\mathbb{X}_0)$, i.e. $\frac{d\tau_q}{dq}=\tau_q \frac{d}{d\xi}$
in $\mathcal{L}(\mathbb{X}_2,\mathbb{X}_0)$. Moreover, by the mean value inequality
$$
\| \tau_{q+h}\, u^\prime - \tau_{q}\, u^\prime  \|_{L^\infty}
\leq 
|h|\,  \| u^{\prime\prime} \| _{L^\infty},
$$
hence ${\| \tau_{q+h}\, \frac{d}{d\xi} - \tau_{q}\, \frac{d}{d\xi}  \|}_{\mathcal{L}(\mathbb{X}_2,\mathbb{X}_0)}
\leq |h|$ and $\frac{d\tau_q}{dq}$ is continuous in the operator norm. 
This completes the proof that $N$ is $C^1$.
Formula (\ref{diffu0}) follows from elementary computations,
using the chain rule, 
the fact that $\tau_{\pi}=\tau_{-\pi}$ on $\mathbb{X}_k$ and 
$\tau_{\pi}u_0=-u_0$, and the equality
$W^{\prime\prime}(x)=2(V^{\prime\prime}(2x)+V^{\prime\prime}(-2x))$.
\end{proof}

In order to apply the implicit function theorem to equation (\ref{pbnl})
in a neighbourhood of the solution $(u,q)=(u_0,\pi )$, 
we now prove the invertibility of the operator
$L=\frac{d^2}{d\xi^2}+D_u N(u_0 , \pi)$.

\begin{lemma}
\label{invert}
The linear operator $L \in {\mathcal L}(X_2, X_0) $ is invertible.
\end{lemma}

\begin{proof}
For a given $f \in X_0$, we look for $y \in X_2$ satisfying 
\begin{equation}
\label{aff}
L\, y=f.
\end{equation}
We use the splitting $X_k = X_k^+ \oplus X_k^-$, where
$$
X_k^\pm = \{\, 
u\in X_k, \
u(\xi +\pi )=\pm u(\xi )
\, \} 
$$
and denote by $P^\pm = \frac{1}{2}(I\pm \tau_\pi )$ the corresponding projectors on $X_k^\pm$.
Since $u_0 \in X_2^-$ and $W^{\prime\prime}$ is even, it follows that
$W^{\prime\prime} (u_0) \in X_0^+$ and 
$D_u N(u_0 , \pi) = W^{\prime\prime}(u_0)\,P^- \in {\mathcal L}(X_0, X_0^-)$.
Setting $f_\pm = P^\pm f$, $y_\pm = P^\pm y$, problem (\ref{aff}) can be rewritten
\begin{eqnarray}
\label{eq1}
y_+^{\prime\prime}&=&f_+ ,\\
\label{eq2}
L^-\, y_- &=&f_- ,
\end{eqnarray}
where $L^-\, y_- = y_-^{\prime\prime} + W^{\prime\prime}(u_0)\, y_-$.
We first observe that 
the operator $\frac{d^2}{d\xi^2} \in {\mathcal L}(X_2^{\pm}, X_0^{\pm})$
is invertible, with inverse $A \in {\mathcal L}(X_0^{\pm}, X_2^{\pm})$ given by
\begin{equation}
\label{invds}
(A\, f )(\xi) = \frac{\xi}{\pi}\int_0^\pi{(s-\pi )\, f(s)\, ds}+\int_0^\xi{(\xi -s )\, f(s)\, ds}
\end{equation}
(injectivity comes from the fact that functions in $X_k^{\pm}$ are odd, and 
it is lengthy but straightforward to check that (\ref{invds}) satisfies the required
periodicity and symmetry properties when $f\in X_0^{\pm}$). Consequently,
equation (\ref{eq1}) admits a unique solution $y_+ \in X_2^+$.
We now consider the case of equation (\ref{eq2}).
The operator $T \in {\mathcal L}(X_2^{-}, X_0^{-})$ defined by
$T\, y_- = W^{\prime\prime}(u_0)\, y_-$ is compact 
(since $X_2^{-}$ is compactly embedded in $X_0^{-}$),
hence $L^- \in {\mathcal L}(X_2^-, X_0^-)$ is a compact perturbation
of an invertible operator. As a consequence, $L^- $ is a Fredholm operator with index $0$
and its invertibility will follow from the fact that $\mbox{Ker}\, L^- = \{0\}$.
The proof of the injectivity of $L^-$ is classical (see e.g. \cite{MacSep} p. 693) but
we include it for completeness.
In the sequel we denote by $v_p$ the solution of (\ref{int}) with initial
condition $v_p (0)=0$, $v_p^\prime (0)=p$, so that $u_0 = v_{p_0}$,
$v_p$ is odd and periodic with period $T(p)$ given by (\ref{period}).
In addition we rescale $v_p$ in a $2\pi$-periodic function
$w (\tau , p) = v_p(\frac{T(p)}{2\pi}\, \tau )$.
The solutions of the homogeneous equation 
$y^{\prime\prime} + W^{\prime\prime}(u_0)\, y=0$
are spanned by two linearly independent solutions 
$u_0^\prime$ and $z = \frac{\partial v_p}{\partial p}_{|p=p_0}$,
where $u_0^\prime$ is even, $2\pi$-periodic and $z$ is odd. Since we have
$v_p( \xi ) = w (\frac{2\pi}{T(p)}\, \xi, p)  $ and $T(p_0 )=2\pi$, it follows that 
\begin{equation}
\label{split}
z (\xi)= - \frac{1}{2\pi}\, T^\prime(p_0)\, \xi\, u_0^\prime(\xi) 
+ \frac{\partial w}{\partial p} (\xi ,p_0)
\end{equation}
with $T^\prime(p_0)=\frac{2\pi(1-\alpha )}{p_0(1+\alpha )}$.
Since the nondegeneracy
condition $T^\prime (p_0) \neq 0$ is satisfied,
$z$ is the sum of an unbounded function of $\xi$ 
and a $2\pi$-periodic one. As a conclusion,
$\mbox{Ker}\, L^- = \{0\}$ since $z$ is non-periodic and
$u_0^\prime$ is even, hence $L$ is invertible. 
\end{proof}

Lemma \ref{invert} allows one to solve equation (\ref{pbnl})
locally in the neighbourhood of $(u,q)=(u_0, \pi )$ using the implicit function
theorem. This yields the following result. 

\begin{lemma}
\label{ift}
There exists an open neighbourhood $\Omega$ of $u_0$ in $X_2$,
an open interval ${\mathcal I}$ containing $q=\pi$ and 
a map $\Phi \in C^1 ({\mathcal I} , X_2)$ (with $\Phi (\pi )=u_0$)
such that for all $q \in {\mathcal I}$, equation (\ref{pbnl})
admits a unique solution $u$ in $\Omega$ given by $u = \Phi (q)$.
\end{lemma}

The function $u_\mu$ describing the family of periodic travelling waves of theorem \ref{existence}
is defined as 
\begin{equation}
\label{defumu}
u_\mu = \Phi (\pi + \mu). 
\end{equation}
Now there remains to analyze in more details the qualitative properties of $u_\mu$.

\subsection{\label{quali}Qualitative properties}

We start by examining a symmetry property of (\ref{ad}).
One can check that $u$ is an odd solution of (\ref{ad})-(\ref{pbc}) with $q=\pi - \mu$
if and only if $-\tau_\pi\,  u$ is an odd solution for $q=\pi + \mu$. Consequently,
by lemma \ref{ift} (which ensures the local uniqueness of the solution)
one has $\Phi (\pi + \mu)=-\tau_\pi \Phi (\pi - \mu)$, i.e. 
$u_{\mu}(\xi )=-u_{-\mu}(\xi +\pi)$.

\vspace{1ex}

The next lemma describes the monotonicity properties of $u_{\mu}$ in different intervals.

\begin{lemma}
\label{mono}
For all $\mu$ small enough, there exists $\theta (\mu ) \in (0,\pi )$
such that $u_\mu^\prime (\xi )>0$ for $\xi \in [0, \theta(\mu ) )$,
$u_\mu^\prime (\xi )<0$ for $\xi \in (\theta(\mu ) ,\pi ]$.
Moreover $\lim\limits_{\mu\rightarrow 0}{\theta(\mu )}=\pi/2$. 
\end{lemma}

\begin{proof}
From the construction of $u_0$ in lemma \ref{binary}, the results
holds true for $\mu=0$ with $\theta(0 )=\pi/2$, and we have
$u_0^{\prime\prime}=-W^\prime (u_0 )<0$ on $(0,\pi )$.
Since $u_\mu$ is $C^2$-close to $u_0$ for $\mu \approx 0$,
for all $\epsilon \in (0,\pi /2 )$ one has 
$u_\mu^{\prime\prime}<0$ on $[\frac{\pi}{2}-\epsilon , \frac{\pi}{2}+\epsilon]$ for $\mu$ small enough,
$u_\mu^{\prime}>0$ on $[0 , \frac{\pi}{2}-\epsilon]$ and
$u_\mu^{\prime}<0$ on $[\frac{\pi}{2}+\epsilon ,\pi]$.
Then the sign properties of $u_\mu^\prime$ stated in lemma \ref{mono}
follow from the monotonicity of $u_\mu^{\prime}$ on $[\frac{\pi}{2}-\epsilon , \frac{\pi}{2}+\epsilon]$
and the intermediate value theorem, and the asymptotic behaviour of $\theta(\mu )$
is obtained by letting $\epsilon \rightarrow 0$.
\end{proof}

Since $u_\mu$ is odd and $2\pi$-periodic, it satisfies
$u_\mu (0)=u_\mu (\pi )=0$ and $u_\mu (\pi+h)=-u_\mu (\pi -h)$.
By lemma \ref{mono} it follows that
$u_\mu$ is positive on $[0,\pi ]$, maximal at $\xi=\theta(\mu )$ and negative on $[\pi , 2\pi]$.
Figure \ref{graphu} illustrates the profile of $u_\mu$
for $\mu \approx -0.63$. 

Using lemma \ref{mono}, we now prove the existence of
an interval where the function $u_\mu$ is linear.

\begin{figure}[h]
\begin{center}
\includegraphics[scale=0.4]{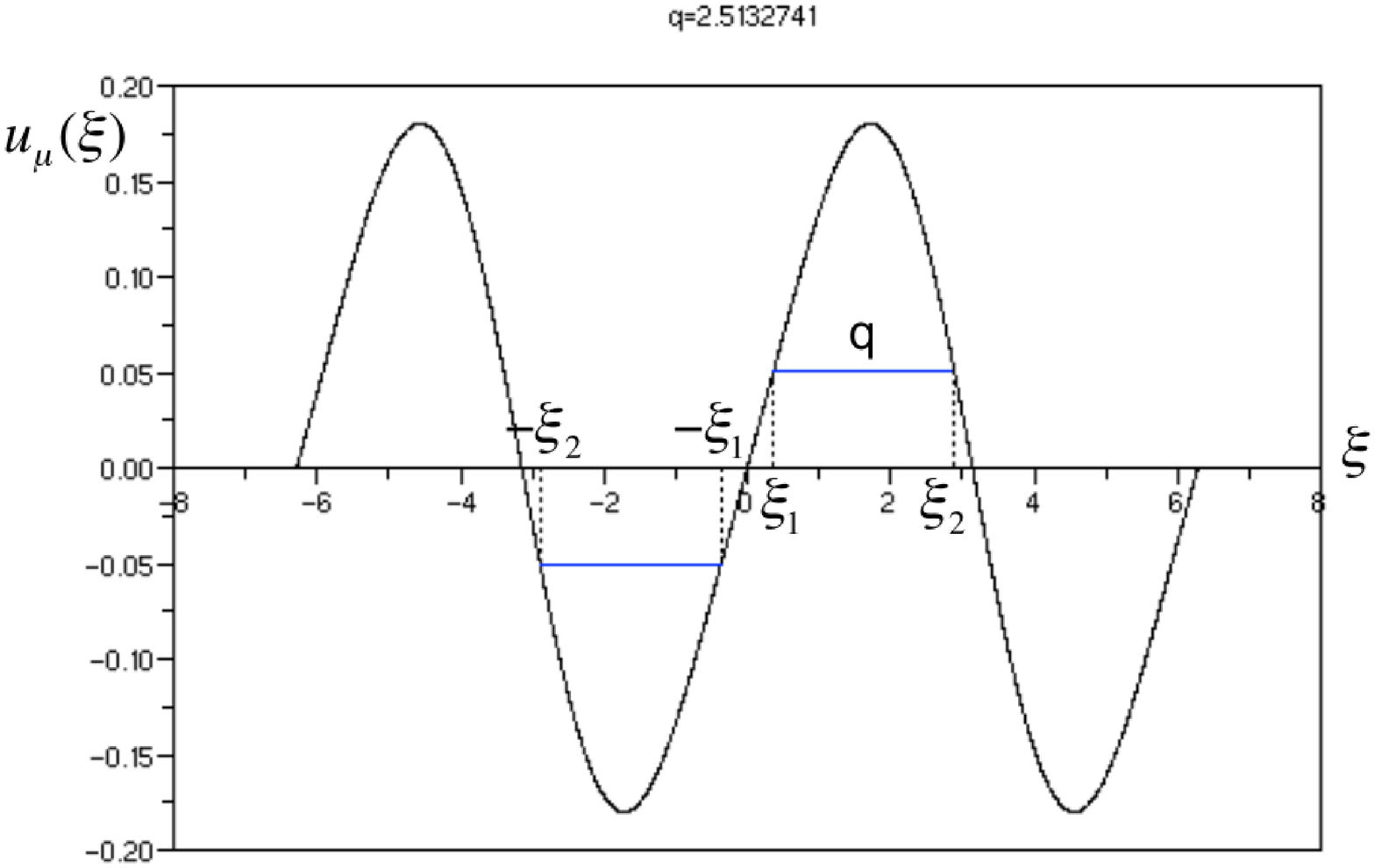}
\end{center}
\caption{\label{graphu} 
Graph of $u_\mu$ for $q=\pi+\mu \approx 2.51$, numerically
computed from equation (\ref{ad}) using the method described in section \ref{num}.
The profile of $u_\mu$ is linear on the interval $[-\xi_1 , \xi_1]$ (see lemma \ref{linumu}),
and the values of $u_\mu (\xi)$ at $\xi_1$ and $\xi_2=\xi_1+q$ coincide.}
\end{figure}

\begin{lemma}
\label{linumu}
For all $\mu <0$ small enough, there exist 
$\xi_1 (\mu )\in (0, \theta (\mu ))$ 
such that
\begin{eqnarray}
\label{ineq1}
u_\mu (\xi + \pi + \mu )-u_\mu (\xi ) &>0& \mbox{ for all }\xi \in (-\pi -\mu -\xi_1 , \xi_1),\\
\label{ineq2}
u_\mu (\xi + \pi + \mu )-u_\mu (\xi ) &<0& \mbox{ for all }\xi \in (\xi_1 , \pi-\mu - \xi_1),
\end{eqnarray}
\begin{equation}
\label{devxi1}
\xi_1 (\mu )= -\frac{\mu}{2} + o(\mu ).
\end{equation}
In addition, one has 
\begin{equation}
\label{linb}
u_\mu (\xi ) = p_\mu \, \xi \mbox{ for all } \xi \in [-\xi_1  , \xi_1 ],
\end{equation}
with $p_\mu = p_0 +O(|\mu |)$ and $p_0$ defined by (\ref{ci}).
\end{lemma}

\begin{proof}
We first show that equation 
\begin{equation}
\label{chord}
u_\mu (\xi_1 + \pi + \mu ) = u_\mu (\xi_1)
\end{equation}
admits a solution
$\xi_1 (\mu )\in (0, \theta (\mu ))$, 
a property clearly illustrated by figure \ref{graphu}.
From lemma \ref{mono} it follows that 
for all $\xi_1 \in [0 , \theta  ]$, the equation
\begin{equation}
\label{eqxi}
u_\mu (\xi )=u_\mu (\xi_1 )
\end{equation}
admits a unique solution $\xi=\xi_2$ in
$[\theta , \pi ]$, the latter 
depending smoothly on $\xi_1 \in [0,\theta )$ by the implicit function theorem. Moreover,
$d(\xi_1)=\xi_{2} - \xi_{1}$ is a strictly decreasing positive function of $\xi_1$,
with $d(0)=\pi$ and $d(\theta)=0$ (see figure \ref{graphu}).
Consequently, for $\mu \in (-\pi ,0)$, the equation 
\begin{equation}
\label{chordeq}
d(\xi_1 )=\pi+\mu 
\end{equation}
admits a unique solution $\xi_1 (\mu )$ in $(0, \theta )$, which
depends smoothly on $\mu$ by the implicit function theorem.
Moreover, expansion (\ref{devxi1}) is obtained by solving (\ref{chordeq})
for $(\xi_1 ,\mu ) \approx (0,0)$ with the implicit function theorem, and using the identity
\begin{equation}
\label{eqder}
d^\prime (0)=\frac{p_\mu}{u_\mu^\prime(\pi)}-1=-2+O(|\mu |),
\end{equation}
where we note $p_\mu = u_\mu^\prime(0)$.
The first equality in (\ref{eqder}) is obtained by
considering the case $\xi_1\approx 0$ of equation (\ref{eqxi}),
and the second one holds true because $u_0^\prime(\pi)=-u_0^\prime(0)$. 
By definition of $d$, $\xi_1(\mu )$ defines a solution of (\ref{chord})
and $\xi_1 (\mu ) + \pi +\mu \in (\theta , \pi )$.

\vspace{1ex}

In the sequel we prove inequalities (\ref{ineq1}) and (\ref{ineq2}).
The proof is based on the known monotonicity properties of $u_\mu$, 
and can be followed more easily with the help of figure \ref{graphu}.

Let us first prove (\ref{ineq1}) by treating the cases $\xi \in (-q -\xi_1  , -\xi_1 ]$
and $\xi \in (-\xi_1 , \xi_1 )$ separately (we recall that $q=\pi +\mu \in (0, \pi )$).

Choosing $\xi \in (-q -\xi_1  , -\xi_1 ]$ yields $-\xi_1 < \xi + q < 2\pi -q -\xi_1 $, therefore we have
$u_\mu (\xi+q ) > \inf_{[-\xi_1  , 2\pi -q -\xi_1 ]}u_\mu=u_\mu (-\xi_1 )$.
Since $u_\mu (-\xi_1 )\geq u_\mu (\xi)$ for all $\xi \in (-q -\xi_1  , -\xi_1 ]$, we have
$u_\mu (\xi+q )>u_\mu (\xi ) $.

In the second case $\xi \in (-\xi_1 , \xi_1 )$, we have $\xi_1 < \xi+q < \xi_1 +q$ for $\mu$ small enough,
hence $u_\mu (\xi+q ) > \inf_{[\xi_1  , \xi_1 +q ]}u_\mu=u_\mu (\xi_1 )$. It follows that
$u_\mu (\xi+q )-u_\mu (\xi ) >u_\mu (\xi_1 ) - u_\mu (\xi )>0$ for all $\xi \in (-\xi_1 , \xi_1 )$.

\vspace{1ex}

Consequently, we have proved inequality (\ref{ineq1}). 
Now let us prove (\ref{ineq2}),
treating the cases $\xi \in (\xi_1  , \xi_1 +q ]$
and $\xi \in (\xi_1 +q  , 2\pi - q -\xi_1 )$ separately. 

For $\xi \in (\xi_1  , \xi_1 +q ]$ we have $\xi_1 + q < \xi+ q < 2\pi$ (since $\xi_1 + q < \pi$ and $q<\pi$),
hence $u(\xi +q )<\sup_{[\xi_1 +q  , 2\pi ]}u_\mu = u_\mu (\xi_1 +q)$.
Consequently we have $u_\mu (\xi+q )-u_\mu (\xi )<u_\mu (\xi_1 +q)-u_\mu (\xi )\leq 0$ for $\xi \in (\xi_1  , \xi_1 +q ]$, hence
$u_\mu (\xi+q )-u_\mu (\xi )<0$.

In the case $\xi \in (\xi_1 +q  , 2\pi - q -\xi_1 )$, we have $2\pi - \xi_1 - q <\xi +q <2\pi - \xi_1$ for $\mu$ small enough,
hence $u_\mu (\xi +q)<\sup_{[2\pi - \xi_1 - q  , 2\pi - \xi_1 ]}u_\mu=u_\mu (2\pi - \xi_1 -q )$. 
Since $u_\mu (2\pi - \xi_1 -q )<u_\mu (\xi )$ for $\xi \in (\xi_1 +q  , 2\pi - q -\xi_1 )$, we have $u_\mu (\xi +q)<u_\mu (\xi )$, which
completes the proof of inequality (\ref{ineq2}). 

\vspace{1ex}

Now we can deduce (\ref{linb}) from
the advance-delay equation (\ref{ad}) and inequality (\ref{ineq1}).
The latter implies $u_\mu (\xi )-u_\mu (\xi -q )>0$ for all $\xi \in (-\xi_1 , \xi_1 + q)$.
Consequently, for all $\xi \in (-\xi_1  , \xi_1 )$ we have both 
$u_\mu (\xi +q)-u_\mu (\xi )>0$ and $u_\mu (\xi )-u_\mu (\xi -q)>0$.
Equation (\ref{ad}) yields in that case $u_\mu^{\prime\prime}(\xi )=0$, from which
(\ref{linb}) follows easily.
\end{proof}

\vspace{1ex}

As a conclusion, from lemmas \ref{ift}, \ref{linumu} and expression (\ref{twfam}),
we have obtained the 
two-parameter families of periodic travelling wave solutions $x_n^{\pm} (t; a,\mu )$
of theorem \ref{existence}, close to the binary oscillations 
described in lemma \ref{binary}. 

\vspace{1ex}

Now let us discuss in more detail some implications of lemma \ref{linumu}.
In what follows we shall write $x_n^{\pm} (t; a,\mu )=x_n(t)$ for notational simplicity.
Since $x_n(t)=a\,  u_\mu (\xi )$
with $\xi = (\pi + \mu )\, n \pm a^{\frac{\alpha -1}{2}} t $,
it follows that as soon as
$\xi \in (-\xi_1(\mu) + 2 k \pi,\xi_1 (\mu )+ 2 k \pi)$ ($k\in \mathbb{Z}$) we have
\begin{equation}
\label{freef1}
x_{n+1}(t)>x_n(t),\ \ \ x_{n}(t)>x_{n-1}(t), \ \ \ \dot{x}_n(t)=\pm a^{\frac{\alpha +1}{2}}\, p_\mu. 
\end{equation}
In other words, each bead periodically switches between a free flight regime
corresponding to (\ref{freef1}) and a contact regime where it interacts with
two or one neighbours. From inequalities (\ref{ineq1}) and (\ref{ineq2}), interaction with
two neighbours takes place for $\xi \in [\xi_1 + \pi + \mu + 2 k \pi, \pi - \mu - \xi_1 + 2 k \pi]$.
This completes the proof of theorem \ref{existence}.

\vspace{1ex}

\begin{remark}
By Galilean invariance of (\ref{nc}), we deduce another family of solutions
\begin{equation}
\label{statbead}
\tilde{x}_n (t)= x_n (t) + v\, t, \ \ \ v=\mp a^{\frac{\alpha +1}{2}}\, p_\mu.
\end{equation}
In that case, each bead periodically switches between a pinning regime where it remains
stationary, and a contact regime where it interacts with
two or one neighbours. Each bead is shifted by $\mp 2\pi a p_\mu$
after one cycle of period $T=2\pi a^{\frac{1-\alpha}{2}}$, therefore bead displacements are unbounded in time.
Such wave profiles are reminiscent of stick-slip oscillations occuring in the Burridge-Knopoff model \cite{schm}, 
albeit this model corresponds to completely different
mathematical and physical settings involving differential inclusions and nonlinear friction. 
\end{remark}

\section{\label{num}Numerical results}

The analysis of section \ref{localcont} has proved the existence
of periodic travelling wave solutions of (\ref{nc}) with wavenumbers
$q\approx \pi$, close to the binary oscillations (\ref{binfam})
corresponding to $q=\pi$. In this section 
we numerically continue this solution branch by decreasing $q$
up to values close to $0$ and analyze the qualitative properties of the waves. 
Our computations are performed for $\alpha =3/2$.
In section \ref{compact}, we deduce the existence of compactons from
the numerical results obtained for periodic waves.

\subsection{\label{numc}Numerical continuation}

We discretize problem (\ref{ad})-(\ref{pbc}) using a second-order finite
difference scheme with step $h= \pi\, .10^{-3}$, and consider the discrete
set of wavenumbers $q=m\, \frac{\pi}{M}$ obtained with all integers $m$ in the interval $[1,M]$,
where we fix $M=50$. With this choice, the corresponding
delays $q$ occuring in (\ref{ad}) are multiples of the step $h$, hence
the advance-delay terms of (\ref{ad}) can be computed by
a discrete shift of the numerical solution. In addition
it is important to keep $h \ll q$ due to a boundary layer effect that occurs
at small wavenumbers (see below). The resulting nonlinear equation
is solved iteratively using a Broyden method \cite{dennis} and path-following,
starting the continuation at $q=\pi$ and $u=u_0$.
Note that higher-order finite difference schemes are useless for $\alpha=3/2$,
because $V^\prime \in C^{1, 1/2}(\mathbb{R})$ has only a limited smoothness,
which guarantees that $u$ is $C^3$ but not $C^4$.

\ve

Figure \ref{profiles} displays the solution $u$ computed numerically
for several values of $q$. As shown by figure \ref{normes},
the supremum norm of $u$ diverges as $q\rightarrow 0$ with
\begin{equation}
\label{normscaling}
\| u \|_{\infty} \sim k_\alpha\, q^{\frac{2}{1-\alpha}}
\end{equation}
and $k_{3/2}\approx 1.36$.

The solution is found exactly linear 
(more precisely, both terms at the right side of (\ref{ad}) vanish)
for $\xi \in [-\ell(q) + 2 k \pi,\ell(q)+ 2 k \pi]$ ($k\in \mathbb{Z}$), with 
$\lim\limits_{q\rightarrow 0}{\ell(q)}=\pi$ and $\lim\limits_{q\rightarrow \pi}{\ell(q)}=0$.
For $q \approx 0$, the linear part of $u$ 
behaves like $u(\xi ) \approx \frac{k_{\alpha}}{\pi}\, q^{\frac{2}{1-\alpha}}\, \xi$
for $\xi \in [-\ell(q) ,\ell(q)]$ (see figure \ref{pentes}).
More generally, the monotonicity properties of $u$ proved in lemmas
\ref{mono} and \ref{linumu} for $q=\pi + \mu \approx \pi$ 
are still valid for the numerical solution with $q \in (0,\pi ]$
(with the correspondence $\ell(q)=\xi_1 (q-\pi )$ between the notations of lemma  
\ref{linumu} and the present ones). 

\ve

Since $x_n (t)= u(q\, n - t)$,
the linear behaviour of $u$ corresponds to a free flight of some beads,
separated by equal gaps of size $u^\prime(0)\, q$.
Free flight occurs when $q\, n - t \in (-\ell(q) + 2 k \pi,\ell(q)+ 2 k \pi)$ ($k\in \mathbb{Z}$),
i.e. we have
\begin{equation}
\label{freef}
q\, n - t \in (-\ell(q) + 2 k \pi,\ell(q)+ 2 k \pi) \, \Leftrightarrow \,
x_{n+1}(t)>x_n(t), \ \ \ 
x_{n}(t)>x_{n-1}(t), \ \ \
\ddot{x}_n(t)=0. 
\end{equation}
For $\xi \in (0,2\pi )$, the two linear parts of $u(\xi )$ are connected by an inner solution
that becomes steeper as $q\rightarrow 0$, and extends over an interval of length $2 (\pi - \ell(q))$.
This part of $u(\xi )$ corresponds to the nonlinear interaction of some packets of beads.

\vspace{1ex}

It is interesting to study
the size of the packets of interacting beads and beads in free flight
as a function of the wavenumber $q$.
Following (\ref{freef}), and recalling that $q=m\, \frac{\pi}{M}$, we define
the $k$th packet of beads in free flight
as the set of beads with index $n$ in the interval 
\begin{equation}
\label{bff}
\mathbb{I}_{2k}(t)=(q^{-1}{(-\ell (q) +t)}+\frac{2Mk}{m},q^{-1}{(\ell (q) +t)}+\frac{2Mk}{m}),
\end{equation}
and the $k$th packet of interacting beads as the set of
beads with index $n$ in the interval 
\begin{equation}
\label{intb}
\mathbb{I}_{2k+1}(t)=[q^{-1}{(\ell (q) +t)}+\frac{2Mk}{m},q^{-1}{(-\ell (q) +t)}+\frac{2M(k+1)}{m}].
\end{equation}
Note that $\mathbb{I}_k (t)$ moves with the wave, so that 
any given bead periodically switches between the free flight and contact regimes. 

Let us define $a_k(t)=q^{-1}t+\frac{2M}{m}\, k+q^{-1}\ell (q)$ and
\begin{equation}
\label{defnq}
N(q)=2 q^{-1}(\pi - \ell(q) ),
\end{equation} 
so that $\mathbb{I}_{2k+1}(t)=[a_k(t),a_k(t)+N(q)]$.
We denote by $\mathcal{N}(a,L)$ the number of integers in the interval $[a,a+L]$,
equal to $\lfloor L \rfloor$ or $\lfloor L \rfloor+1$ depending on the values of $a$ and $L$
($\lfloor L \rfloor$ denotes the integer part of $L$). The function $\mathcal{N}$ satisfies
$$
\mathcal{N}(a+1,L)=\mathcal{N}(a,L), \ \ \
\int_0^1{\mathcal{N}(a,L)\, da}=L.
$$
The number of beads with index $n$ in $\mathbb{I}_{2k+1}(t)$ is then equal to
$\mathcal{N}(a_k(t),N(q))$, and switches between the two values
$ \lfloor N(q) \rfloor$ and $ \lfloor N(q) \rfloor+1$ during the motion.
This number being $q$-periodic in $t$ and $m$-periodic in $k$,
the average number of adjacent interacting beads reads
\begin{equation}
\label{avcomp}
\frac{1}{qm}\sum_{k=1}^m\int_0^q{\mathcal{N}(a_k(t),N(q))\, dt}=\int_0^1{\mathcal{N}(a,N(q))\, da}=N(q)
\end{equation}
(we have used the changes of variables $a=a_k(t)$ in above computation). 
Figure \ref{number} depicts the evolution of 
$N(q)$ when $q$ is varying in $(0,\pi )$. One observes that $N(\pi)=2$
(as expected for binary oscillations, since beads are paired for almost all times), and the value of $N$ at our minimal $q$ is
$N(\pi /50) \approx 4.8$, which is
close to the spatial extent of Nesterenko's solitary wave \cite{neste2}.

In addition, using similar arguments as above
(just considering $\mathbb{I}_{2k}(t)$ instead of $\mathbb{I}_{2k+1}(t)$), 
the average number of adjacent
beads in free flight is $F(q)=\frac{2\ell(q)}{q}$, i.e. it is equal to the
length of the interval of linear increase of $u$ divided by the wavenumber $q$.
The number of beads with index $n$ in the open interval $\mathbb{I}_{2k}(t)$ switches between 
the two values $ \lceil F(q) \rceil -1$ and $\lceil F(q) \rceil $ during the motion, where
$\lceil F \rceil $ denotes the smallest integer not less than $F$
(note that $ \lceil F \rceil -1$ and $\lfloor F \rfloor$ coincide when $F$ is not an integer).
Figure \ref{lsq} provides the graph of $F(q)$. One can notice that
$F(q)> 1$ when $q< q_k \approx 1.8$.
In that case, two packets of interacting beads are always separated by
some beads in free flight.

\begin{figure}[!h]
\psfrag{x}[0.9]{ $\xi$}
\psfrag{u}[1][Bl]{ $u(\xi )$}
\begin{center}
\includegraphics[scale=0.4]{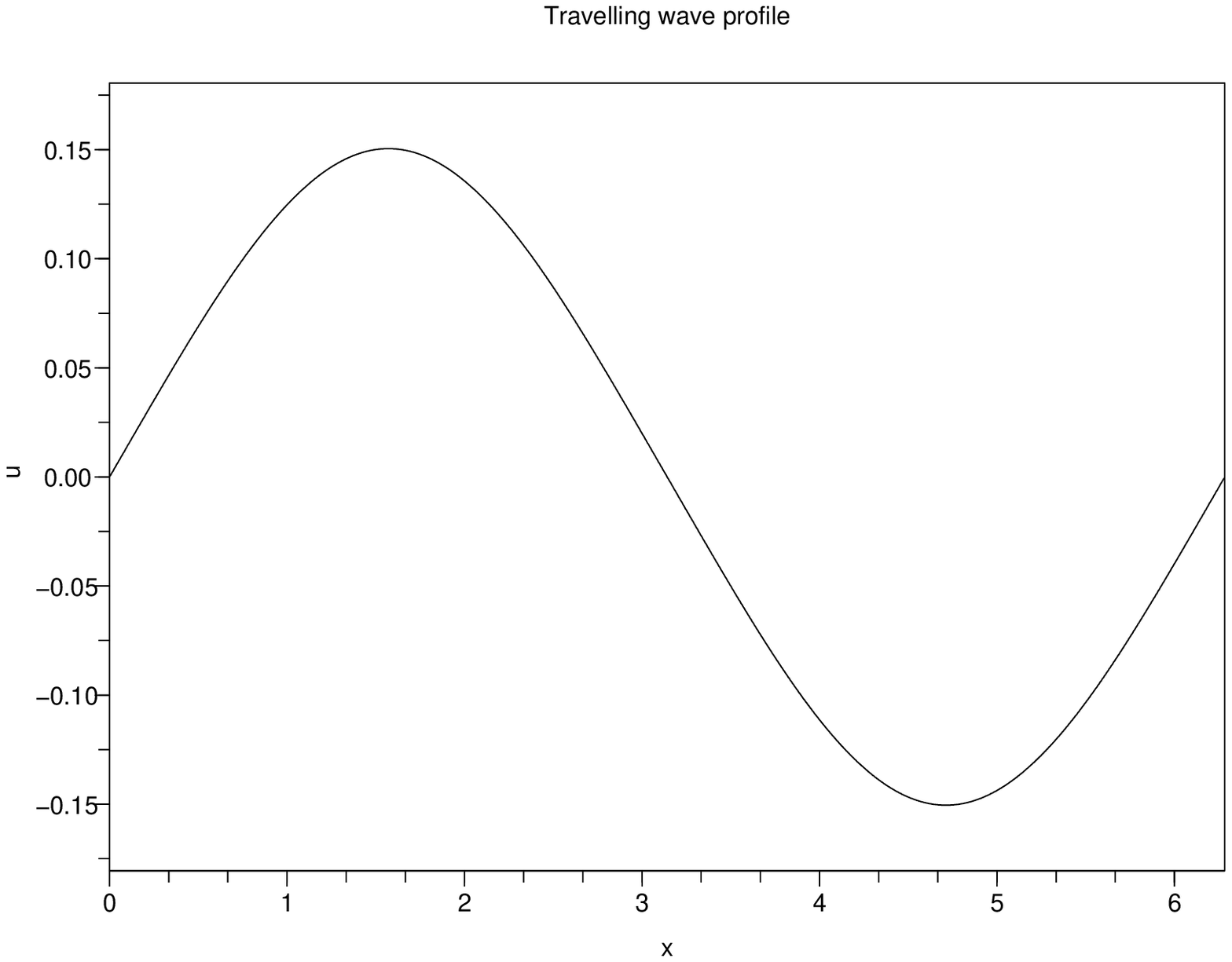}
\includegraphics[scale=0.4]{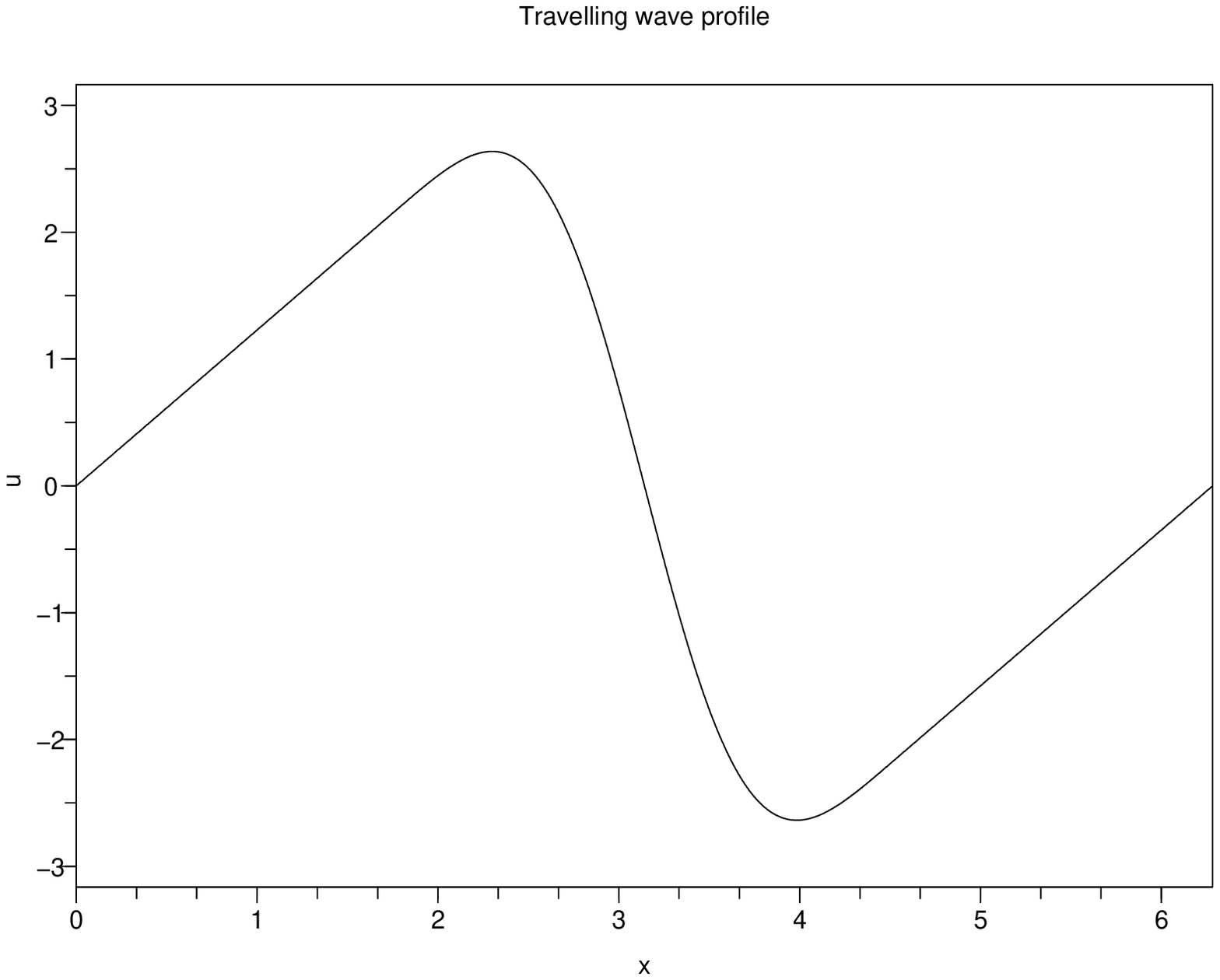}
\includegraphics[scale=0.4]{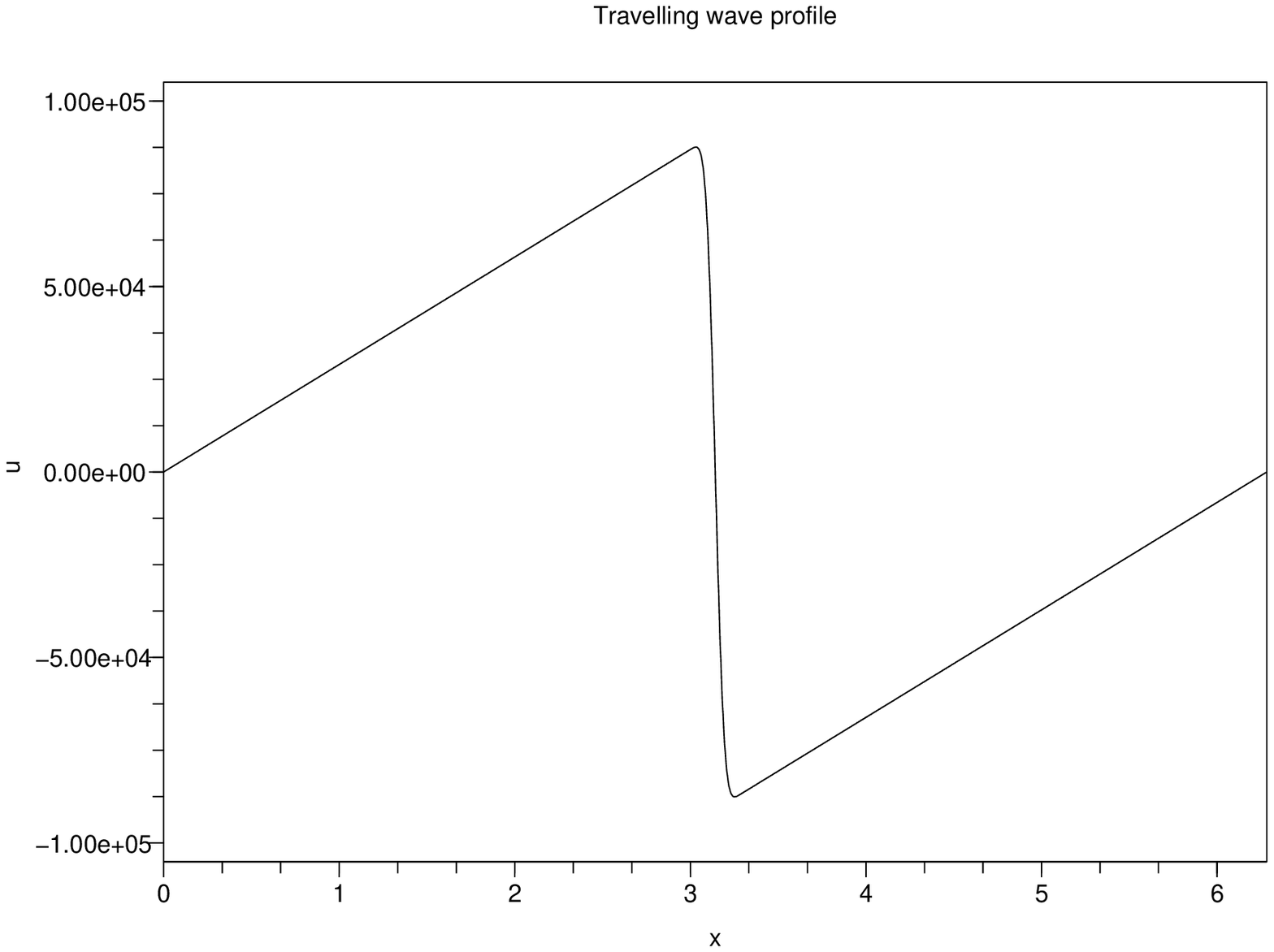}
\end{center}
\caption{\label{profiles} 
Solution of (\ref{ad})-(\ref{pbc}) computed numerically for
$q=\pi$ (top), $q=3\pi / 10$ (middle) and $q=\pi / 50$ (bottom).
}
\end{figure}

\begin{figure}[!h]
\psfrag{q}[0.9]{ $q$}
\psfrag{norm}[1][Bl]{ $\| u \|_{\infty}$}
\begin{center}
\includegraphics[scale=0.4]{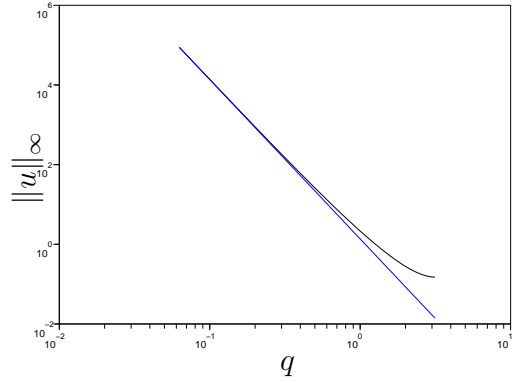}
\end{center}
\caption{\label{normes} 
Graph of the
supremum norm of $u$ (in logarithmic scale)
when the wavenumber $q$ varies.
The line corresponds to the approximation
$\| u \|_{\infty}\approx k_{\alpha}\, q^{\frac{2}{1-\alpha}}$ for $\alpha =3/2$
and $k_{3/2}\approx 1.36$.
}
\end{figure}

\begin{figure}[!h]
\psfrag{q}[0.9]{ $q$}
\psfrag{slope}[1][Bl]{ $u^\prime $}
\begin{center}
\includegraphics[scale=0.4]{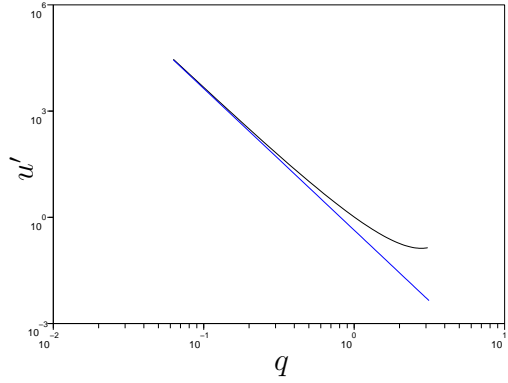}
\end{center}
\caption{\label{pentes}
The curve gives for the different wavenumbers $q$
the constant value of $u^\prime $ in the free flight interval
$\xi \in [-\ell(q) ,\ell(q)]$ (plot in logarithmic scale).
The line corresponds to the approximation
$u^\prime \approx \frac{k_{\alpha}}{\pi}\, q^{\frac{2}{1-\alpha}}$ for $\alpha =3/2$ and $k_{3/2}\approx 1.36$.
}
\end{figure}

\begin{figure}[!h]
\psfrag{q}[0.9]{ $q$}
\psfrag{N}[1][Bl]{ $N(q )$}
\begin{center}
\includegraphics[scale=0.4]{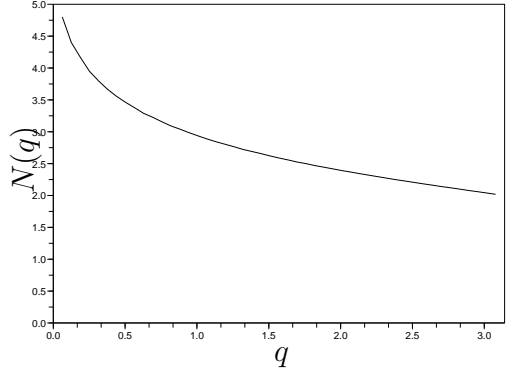}
\end{center}
\caption{\label{number} 
Average number $N(q)$ of consecutive interacting beads when the wavenumber $q$ varies.
}
\end{figure}

\begin{figure}[!h]
\psfrag{q}[0.9]{ $q$}
\psfrag{2l/q}[1][Bl]{ $\frac{2\ell(q)}{q}$}
\begin{center}
\includegraphics[scale=0.4]{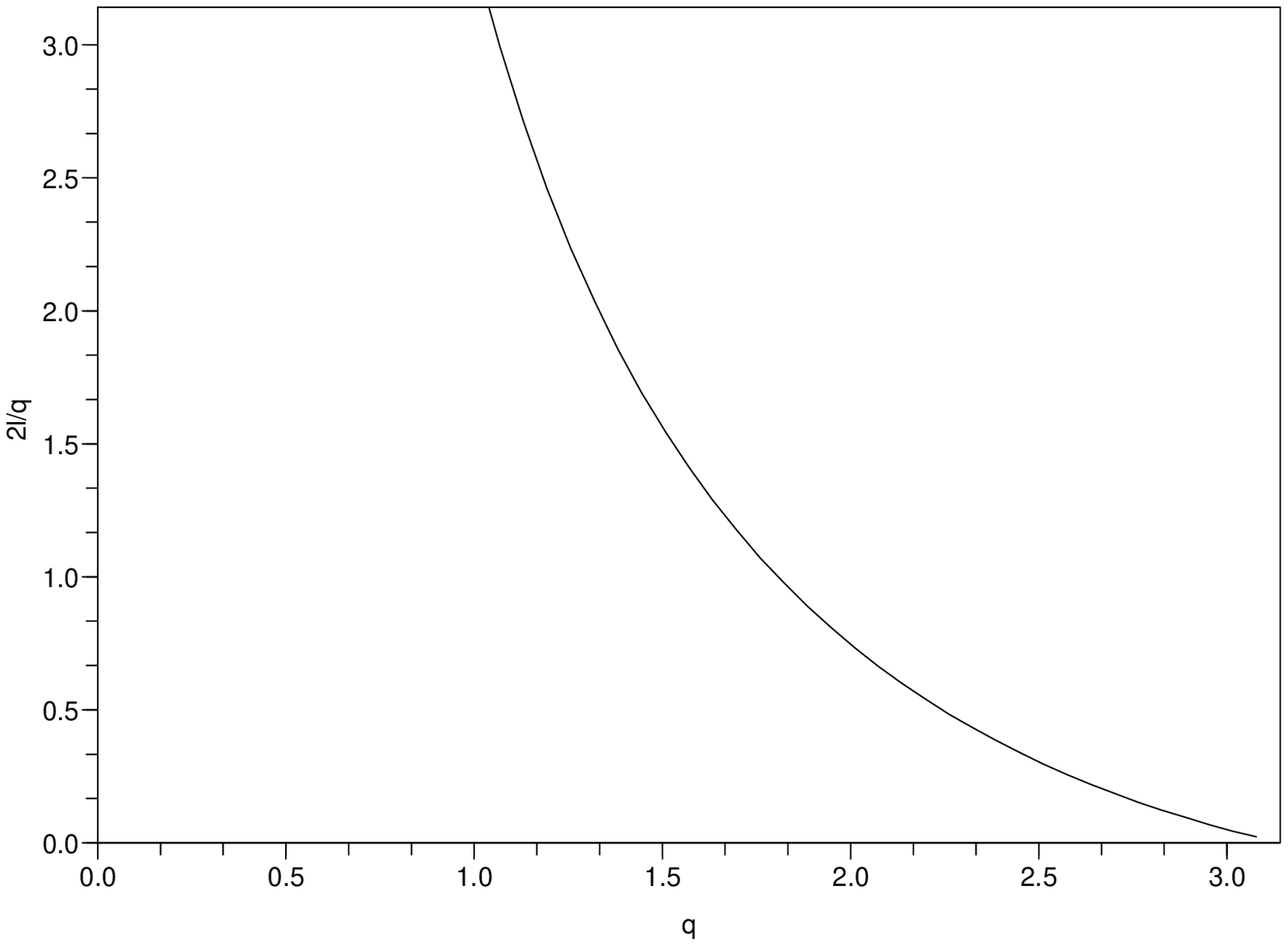}
\end{center}
\caption{\label{lsq} 
Graph of $F(q)=\frac{2\ell(q)}{q}$, the average number of adjacent beads in free flight.}
\end{figure}

\ve

Now let us describe the nonlinear dispersion relation satisfied by the periodic travelling waves. 
Let us denote by $u(\xi;q)$ the family of $2\pi$-periodic wave profiles
parametrized by $q\in (0,\pi ]$, computed by numerical continuation from
$u(.; \pi)=u_0$. We recall that these solutions correspond to travelling wave
solutions of (\ref{nc}) given by
\begin{equation}
\label{twper}
x_n (t)=a\, u(q\, n - a^{\frac{\alpha -1}{2}} t  + \phi ; q ),
\end{equation}
where $a>0$ is a parameter. Instead of $a$ we now choose the wave amplitude
$$
A=\| \{ x_n  \} \|_{L^\infty (\mathbb{Z} \times \mathbb{R})}=a\, {\|  u(.;q)\|_\infty}
$$ 
as a new parameter.
The frequency of solution (\ref{twper}) is consequently given by the dispersion relation
\begin{equation}
\label{dispers}
\omega (q,A)=A^{\frac{\alpha -1}{2}}\, {\|  u(.;q)\|}_\infty^{\frac{1-\alpha}{2}}. 
\end{equation}
From the above numerical results, we deduce
\begin{equation}
\label{disperseq}
\omega (q,A) \sim k_\alpha^{\frac{1-\alpha}{2}}\, A^{\frac{\alpha -1}{2}}\, q\ \ \
\mbox{ as }q\rightarrow 0.
\end{equation}
Consequently, the dispersion relation (\ref{dispers}) behaves linearly in $q$
at fixed $A\neq 0$ when $q\rightarrow 0$,
and one has $\frac{\partial \omega}{\partial q}(0,0)=0$,
which corresponds to a vanishing ``sound velocity" in the limit of small amplitudes.
This property is consistent with the fact that, in the granular chain with 
a precompression $f_0$, the sound velocity
vanishes when $f_0 \rightarrow 0$. 
The graph of $\omega (.,A)$ is shown in figure \ref{dispersplot} for $A=A_0={\|  u_0  \|}_\infty$
(so that $\omega (\pi ,A_0)=1$).

In what follows we formally analyze why the scaling (\ref{normscaling}) of ${\|  u(.;q)\|}_\infty$ occurs 
when $q\rightarrow 0$ and
heuristically compute the constant $k_\alpha$ present in the dispersion relation 
(\ref{disperseq}).

\begin{figure}[!h]
\psfrag{q}[0.9]{ $q$}
\psfrag{om}[1][Bl]{ $\omega$}
\begin{center}
\includegraphics[scale=0.4]{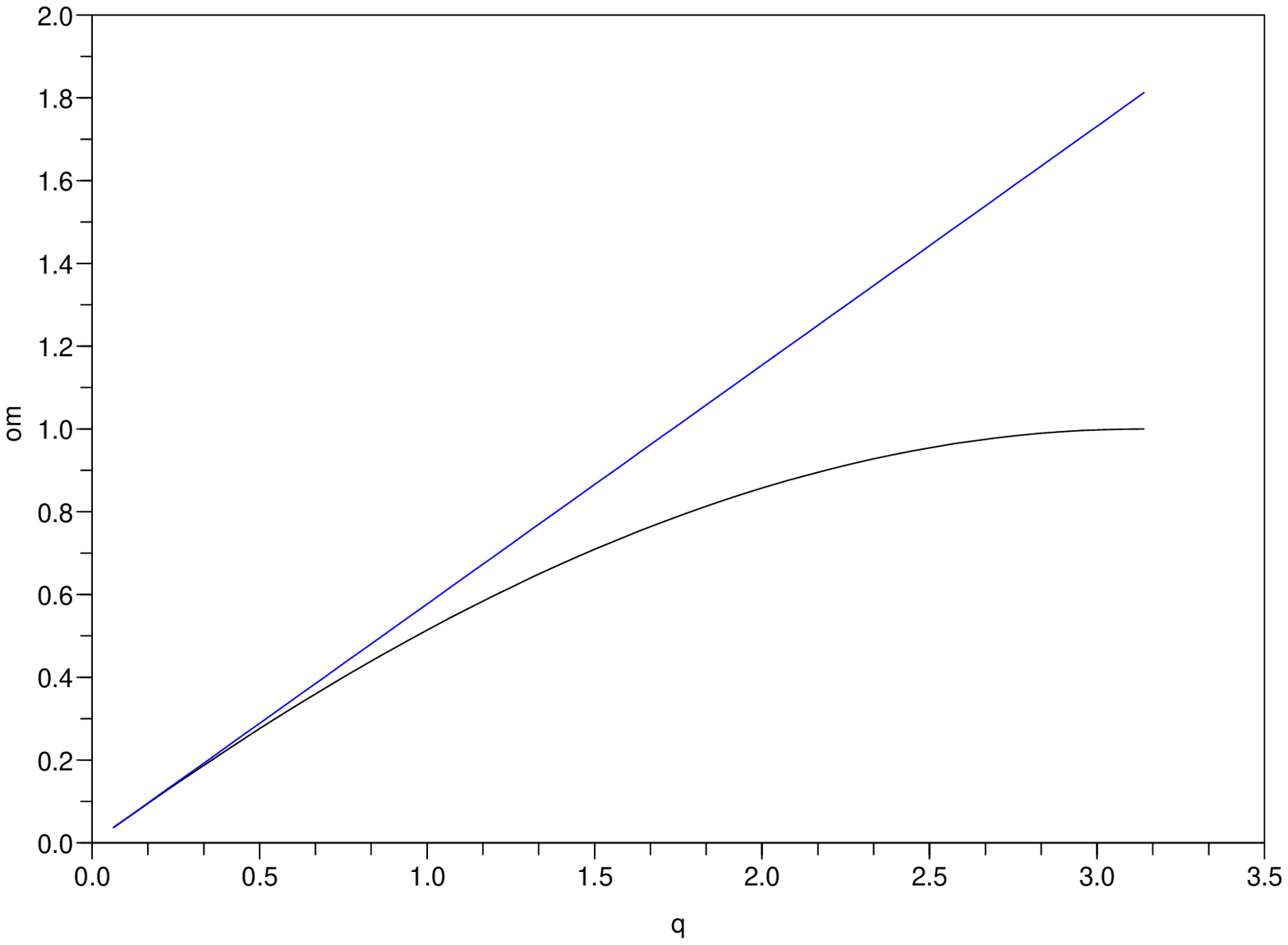}
\end{center}
\caption{\label{dispersplot} 
Plot of the dispersion relation (\ref{dispers}) for $\alpha=3/2$ and $A={\|  u_0 \|}_\infty$, 
and its linear approximation at $q=0$ given by (\ref{disperseq}).}
\end{figure}

\subsection{\label{longw}Long wave limit}

Let us examine the limit $q\rightarrow 0$ more closely.
We first normalize the solution of (\ref{ad})-(\ref{pbc}) computed in section \ref{numc}
by setting $u= q^{\frac{2}{1-\alpha}}\, \tilde{u}$,
which yields
\begin{equation}
\label{adresc0}
q^2\, \tilde{u}^{\prime\prime} (\xi ) =
V^\prime (\tilde{u}(\xi +q)-\tilde{u}(\xi ))-V^\prime (\tilde{u}(\xi )-\tilde{u}(\xi -q)),
\end{equation} 
or equivalently
\begin{equation}
\label{adresc1}
\tilde{u}^{\prime\prime} (\xi ) =q^{\alpha -1}\, \frac{1}{q}\, \big(
V^\prime [\,  \frac{\tilde{u}(\xi +q)-\tilde{u}(\xi )}{q} \, ]-V^\prime [\, \frac{\tilde{u}(\xi )-\tilde{u}(\xi -q)}{q}   \, ]
\big),
\end{equation} 
where we have used the fact that $V^\prime (x)= -|x|^{\alpha}\, H(-x)$.
For $q\approx 0$, we replace the right side of (\ref{adresc1}) by its continuum limit
\begin{equation}
\label{adresc2}
\tilde{u}^{\prime\prime} (\xi ) =q^{\alpha -1}\, 
\frac{d}{d\xi}\, 
V^\prime (\tilde{u}^\prime ).
\end{equation} 
Setting $q=0$ in (\ref{adresc2}) and using the fact that $\tilde{u}(0)=0$, 
we get the outer approximate solution 
\begin{equation}
\label{outer}
\tilde{u}_o(\xi ) = \frac{\lambda}{\pi}\, \xi, \ \ \ \xi \in [0,\pi ),
\end{equation}
where $\lambda$ will be subsequently determined by a matching condition.
When $q$ is small, equation (\ref{adresc2}) with $q=0$ does not 
describe the travelling wave profile in a boundary layer around
$\xi=\pi$ where $\tilde{u}^\prime$ becomes large, and
a rescaling of $\tilde{u}$ is necessary to capture the structure of the inner solution.
Setting $s=q^{-1}\, (\xi - \pi )$ and $\tilde{u}(\xi )= y(s)$, equation
(\ref{adresc0}) becomes
\begin{equation}
\label{adresc}
y^{\prime\prime} (s ) =
V^\prime (y(s +1)-y(s ))-V^\prime (y(s )-y(s -1)), \ \ \
s \in \mathbb{R}.
\end{equation} 
There exists an odd solution $y_{\rm{sol}}$ of (\ref{adresc}) satisfying
\begin{equation}
\label{bcresc}
\lim_{s\rightarrow - \infty}{y_{\rm{sol}}(s)}= k_\alpha, \ \ \
\lim_{s\rightarrow + \infty}{y_{\rm{sol}}(s)}=- k_\alpha,
\end{equation} 
where the convergence towards $k_\alpha $ is super-exponential
and 
\begin{equation}
\label{approxkal}
k_{3/2}\approx 1.3567
\end{equation}
(see appendix \ref{comprec}).
This solution corresponds to a solitary wave solution of (\ref{nc}) \cite{neste2,mackay,english,stef}
with velocity equal to unity. It yields the inner solution
\begin{equation}
\label{inner}
\tilde{u}_i(\xi ) = y_{\rm{sol}}(\frac{\xi - \pi}{q})
\end{equation}
which agrees very well with the numerical solution in the vicinity of $\xi=\pi$
(see figure \ref{approxz}).
Now let us match the outer and inner solutions at some point $\xi = \pi - q\, A(q)$ when $q\rightarrow 0$,
where $A(q)>0$, $\lim_{q\rightarrow 0}{(q\, A(q))}=0$ and $\lim_{q\rightarrow 0}{A(q)}=+\infty$.
The matching condition reads
$$
\lim_{q\rightarrow 0}{\lambda\, (1-\frac{q}{\pi}\, A(q))}=\lim_{q\rightarrow 0}{y_{\rm{sol}}(-A(q))},
$$
which gives consequently
$\lambda = k_\alpha$. This approximation is in reasonable agreement with the numerical results since
the value of $k_{3/2}$ deduced from figure \ref{normes} differs from the value (\ref{approxkal}) by
$8.10^{-3}$ ($q$ should be further decreased to get a better agreement).
From the previous analysis, we deduce the following approximate solution of (\ref{adresc0})
obtained by summing the inner and outer solutions and substracting their common
value at the matching point
\begin{equation}
\label{solapprox}
\tilde{u}_{\rm{app}}(\xi) \approx 
\left\{
\begin{array}{ll}
{\displaystyle \frac{k_\alpha}{\pi}\, \xi + y_{\rm{sol}}(\frac{\xi - \pi}{q}) - k_\alpha,}
&
\xi \in [0,\pi ], \\
{\displaystyle
\frac{k_\alpha}{\pi}\, (\xi -2\pi ) + y_{\rm{sol}}(\frac{\xi - \pi}{q}) + k_\alpha,}
&
\xi \in [\pi, 2\pi ].
\end{array}
\right.
\end{equation}
The numerical and approximate solutions are compared in figure \ref{approx},
which shows a good agreement between both profiles.
The mathematical justification of this formal asymptotic analysis is left
as an interesting open problem.

\ve

With approximation (\ref{solapprox}) at hand, let us now examine the wave structure
in more details. Setting $a=q^{\frac{2}{\alpha -1}}$, $\phi=\pi$ in (\ref{twper}) and using
(\ref{solapprox}) yields approximate travelling wave solutions
$x_n^{\rm{app}}(t)=\tilde{u}_{\rm{app}}(q(n-t)+\pi )$ with velocity
equal to unity and $O(1)$ amplitude when $q\rightarrow 0$.
These solutions take the form
\begin{equation}
\label{twapprox}
x_n^{\rm{app}}(t)= \frac{k_\alpha}{\pi}\, q(n-t)+y_{\rm{sol}}(n-t),\ \ \
\mbox{ for } |n-t | < \frac{\pi}{q} 
\end{equation}
and are $\frac{2\pi}{q} $-periodic with respect to
the moving frame coordinate $s=n-t$.
These waves consist of a succession of
compression solitary waves, separated by large regions of free flight (of size $O(q^{-1})$) 
where particles move at a constant $O(q)$ velocity and equal gaps of size $O(q)$
are present between adjacent beads.

\begin{figure}[!h]
\psfrag{x}[0.9]{ $\xi$}
\psfrag{u}[1][Bl]{ $\tilde{u}(\xi )$}
\begin{center}
\includegraphics[scale=0.4]{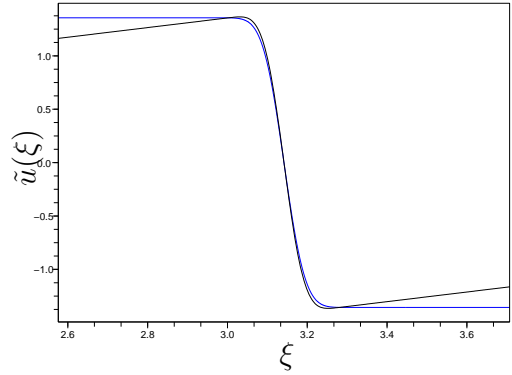}
\end{center}
\caption{\label{approxz} 
Solution of (\ref{adresc0}) with period $2\pi$ computed numerically for
$q=\pi / 50$, with a zoom in the vicinity of $\xi=\pi$ (black curve).
This solution is compared with the inner solution (\ref{inner})
corresponding to the solitary wave studied in the appendix (blue curve).
Both profiles agree very well in
a domain of width $\delta \approx 0.2 $ around $\xi=\pi$, $\delta$ being
approximately four times larger than $1/q$.}
\end{figure}

\begin{figure}[!h]
\psfrag{x}[0.9]{ $\xi$}
\psfrag{u}[1][Bl]{ $\tilde{u}(\xi )$}
\begin{center}
\includegraphics[scale=0.4]{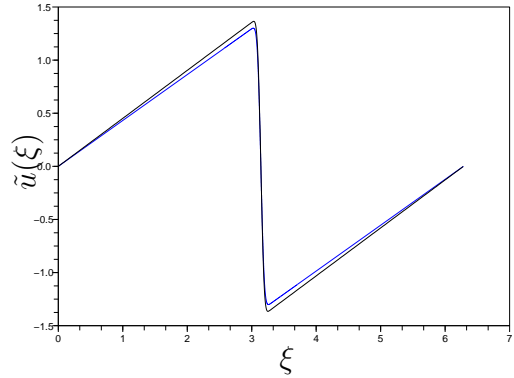}
\end{center}
\caption{\label{approx} 
Solution of (\ref{adresc0}) with period $2\pi$
computed numerically for $q=\pi /50$ (black curve) and its approximation (\ref{solapprox})
(blue curve).}
\end{figure}

\subsection{\label{compact}Compacton solutions}

Several approximations of compression solitary waves
with compact supports have been derived in the literature
\cite{neste1,neste2,ap}, in order to approximate exact
solitary wave solutions of (\ref{nc}) decaying super-exponentially \cite{english,stef}. 
However, the existence of exact solitary wave solutions of (\ref{nc}) with compact
support remained an open question. 
Such solutions are supported by different classes of Hamiltonian PDE with 
fully-nonlinear dispersion, and have been known
as {\em compactons} after the work of Rosenau and Hyman \cite{rh}.
Nonlinear PDE supporting compactons 
can be introduced to analyze 
lattices with fully nonlinear interactions near a continuum limit. 
In this context, a common scenario is the transition from a compacton
to a noncompact (super-exponentially localized) solution when
passing from the continuum model to the discrete lattice \cite{ap}
(see also \cite{ros,jamesc2} for similar results in the context of discrete breather solutions).

\vspace{1ex}

In contrast to this situation, we
show in this section that the numerical results of sections
\ref{numc} imply the existence of compactons for the full lattice (\ref{nc}).
This result is due to the occurence of free flight analyzed prevously, made possible
by the unilateral character of Hertzian interactions, which vanish when beads are not in contact.

\vspace{1ex}

Let us consider the solutions $u(\xi;q)$ of (\ref{ad})-(\ref{pbc}) obtained numerically.
These solutions behave linearly on the intervals
$I_0=[-\ell(q) , \ell(q)]$ and $I_1=[2\pi -\ell(q) , 2\pi + \ell(q)]$. Denoting by
$p(q)=u^\prime (0;q)$ their slope in these intervals, we have $u(\xi;q)=p(q)\, \xi$ on $I_1$ and
$u(\xi;q)=p(q)\, (\xi - 2\pi )$ on $I_2$. 

In what follows we assume $F(q)=2 \ell (q)/q  \geq 1$, which corresponds to fixing 
$q\leq q_k \approx 1.8$ (see figure \ref{lsq}). This condition can be interpreted as follows.
In the case of a strict inequality
$F(q)>1$, in a chain of beads where the wave $u(\xi;q)$ propagates, 
we have seen that two packets of interacting beads defined by (\ref{intb})
are always separated by some beads in free flight, hence
$u(\xi;q)$
can be interpreted as a periodic train of independent pulses of finite width (i.e. compactons). 
Moreover, under the assumption $F(q) \geq 1$, the length of the intervals $I_1$ and $I_2$ is 
larger than the delay $q$ involved in (\ref{ad}). 

If $F(q) \geq 1$,
the solution $u(\xi;q)$ can be linearly extended in order to get a new solution of (\ref{ad}), defined by
\begin{equation}
\label{extens}
U_q(\xi )=
\left\{
\begin{array}{ccc}
p(q)\, \xi & \mbox{ for } & \xi \leq -\ell (q),\\
u(\xi;q) & \mbox{ for } & \xi \in [-\ell (q), 2\pi + \ell(q)], \\
p(q)\, (\xi - 2\pi ) & \mbox{ for } & \xi \geq 2\pi + \ell(q).
\end{array}
\right.
\end{equation}
The profile of $U_q$ is plotted in figure \ref{graphulin} for $q\approx 0.88$.
Let us check that $U_q$ defines a solution of (\ref{ad}) when $F(q) \geq 1$
(for notational simplicity we shall omit the $q$-dependency).
For $\xi \in I^{-}=(-\infty , \ell ]$ we have $U^{\prime\prime}(\xi)=0$,
$U(\xi + q )-U(\xi) \geq 0$ and $U(\xi )-U(\xi -q) \geq 0$, hence 
$U$ is a solution of (\ref{ad}) on $I^{-}$. The same properties hold true 
for $\xi \in I^{+}=[2\pi - \ell , +\infty )$. Moreover, $U$ is a solution of
(\ref{ad}) for $\xi \in I^{c}=[q- \ell , 2\pi + \ell -q]$ because
$U=u$ and $\tau_{\pm q} U=\tau_{\pm q} u$ on $ I^{c}$.
Since the condition $F \geq 1$ is equivalent to having 
$\mathbb{R}=I^{-} \cup I^{c} \cup I^{+}$, $U$ defines a solution of (\ref{ad}) 
on $\mathbb{R}$.

The solutions of (\ref{ad}) defined by (\ref{extens}) correspond to travelling wave
solutions of (\ref{nc}) given by
\begin{equation}
\label{twlin}
X_n (t)=a\, U_q(q\, n - a^{\frac{\alpha -1}{2}} t +\phi ),
\end{equation}
where $a>0$, $q\leq q_k$ and $\phi \in \mathbb{R}$ are 
parameters. Moreover,
by Galilean invariance of (\ref{nc}), we deduce another family of solutions
\begin{equation}
\label{compacton}
\tilde{X}_n (t)= X_n (t) + v\, t
\end{equation}
where we fix
\begin{equation}
\label{defvit}
v= a^{\frac{\alpha +1}{2}}\, p(q).
\end{equation}
These solutions correspond to single compactons. Indeed,
they consist of travelling waves, in the sense that bead velocities
$\dot{X}_n (t)$ and relative displacements $X_n (t) - X_{n-1}(t)$ are
functions of $n-c\, t$ with $c = q^{-1} a^{\frac{\alpha -1}{2}}$. Moreover,
each bead remains stationary
except in a finite time interval where it experiences a compression. For a bead with
index $n$, compression occurs when $q\, n - a^{\frac{\alpha -1}{2}} t + \phi \in (\ell (q), 2\pi - \ell(q))$,
i.e. when $U_q$ does not behave linearly.
Fixing $\phi = - \ell (q)$, the support of the moving compacton at time $t$
corresponds to $n   \in ( c\, t , c \, t +  N(q))$,
where $N(q)$ is defined by (\ref{defnq}). 

Beads outside the support are separated
by equal gaps of size $\Delta = a\, p(q)\, q$ and each bead position is shifted by $2\pi a p(q)$
after the passage of the compacton. The gap $\Delta$ depends 
on the parameters $a$ and $q$, or equivalently on the
compacton width $N(q)$ and velocity $c$. Fixing $c=1$ and letting
$q \rightarrow 0$, the gap becomes $O(q)$ and the compacton profile
converges towards the classical solitary wave in the compression region
(see section \ref{longw}).

\ve

Similarly to (\ref{extens}), one can obtain a
family of $2$-compacton solutions of (\ref{ad}) 
by gluing the solutions $U_q$ and $U_{q,\theta}^{(1)}(\xi )=p(q)\, \theta+ U_q (\xi  - 2\pi - \theta )$.
Let us define
\begin{equation}
\label{extens2}
U^{(2)}_{q,\theta}(\xi )=
\left\{
\begin{array}{ccc}
U_q( \xi ) & \mbox{ for } & \xi \leq 2\pi + \theta +\ell (q),\\
U_{q,\theta}^{(1)}(\xi ) & \mbox{ for } & \xi \geq 2\pi -\ell (q),
\end{array}
\right.
\end{equation}
where $\theta = q\, d$, $d\geq 0$ is a parameter and $q\in (0, q_k ]$, so that $q\leq 2 \ell (q)$.
The graph of $U^{(2)}_{q,\theta}$ is illustrated by figure \ref{graphulin}
(note that $U_q$ and $U_{q,\theta}^{(1)}$ coincide on $[2\pi -\ell (q) , 2\pi + \theta +\ell (q)]$).
Let us check that $U^{(2)}_{q,\theta}$ defines a solution of (\ref{ad})
(for simplicity we shall omit dependency in $q$ and $\theta$ in notations).
Firstly, $U^{(2)}$ is a solution of
(\ref{ad}) for $\xi \in J^{-}=(- \infty , 2\pi + \theta +\ell -q]$ because
$U^{(2)}=U$ and $\tau_{\pm q} U^{(2)}=\tau_{\pm q} U$ on $ J^{-}$.
Similarly, $U^{(2)}$ is a solution of (\ref{ad}) on $J^+ = [2\pi -\ell +q, +\infty )$
because
$U^{(2)}=U^{(1)}$ and $\tau_{\pm q} U^{(2)}=\tau_{\pm q} U^{(1)}$ on $ J^{+}$.
If $q \leq \ell + \theta/2$ we have $J^{-} \cup  J^{+}= \mathbb{R}$ and thus
$U^{(2)}$ defines a solution of (\ref{ad}) on $\mathbb{R}$.
Moreover, if $\ell + \theta/2 <q \leq 2 \ell$, then we have $(U^{(2)})^{\prime\prime}(\xi)=0$,
$U^{(2)}(\xi + q )-U^{(2)}(\xi) > 0$ and $U^{(2)}(\xi )-U^{(2)}(\xi -q) > 0$
for all $\xi \in J^c = (2\pi + \theta +\ell -q , 2\pi -\ell +q )$, hence 
$U^{(2)}$ is a solution of (\ref{ad}) on $J^{c}$, which yields a solution of (\ref{ad}) on
$\mathbb{R} =  J^{-} \cup J^c \cup  J^{+}$.

\ve

Each solution of (\ref{nc})
\begin{equation}
\label{2compacton}
\tilde{X}_n (t)= a\, U^{(2)}_{q,\theta}(q\, n - a^{\frac{\alpha -1}{2}} t +\phi ) + v\, t
\end{equation}
with $v$ defined by (\ref{defvit})
consists of  
two independent compactons separated by the distance $F(q)+d$, with stationary beads
outside the compactons.
Similarly, generalizing the above construction would allow to define an infinity
of solutions of (\ref{nc}), corresponding to an arbitrary number of independent compactons
with variable compacton spacing.

\begin{figure}[h]
\begin{center}
\includegraphics[scale=0.4]{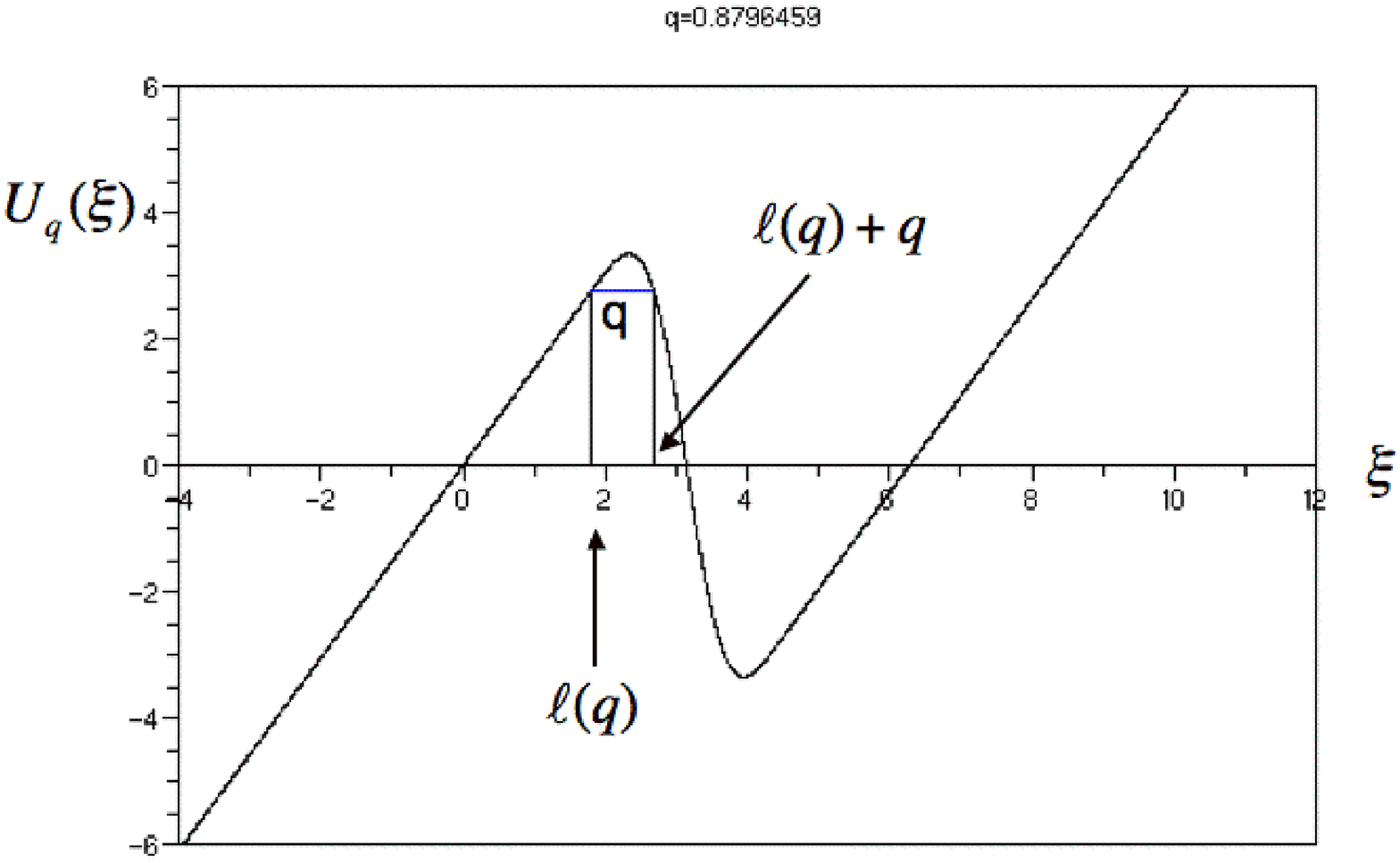}
\includegraphics[scale=0.4]{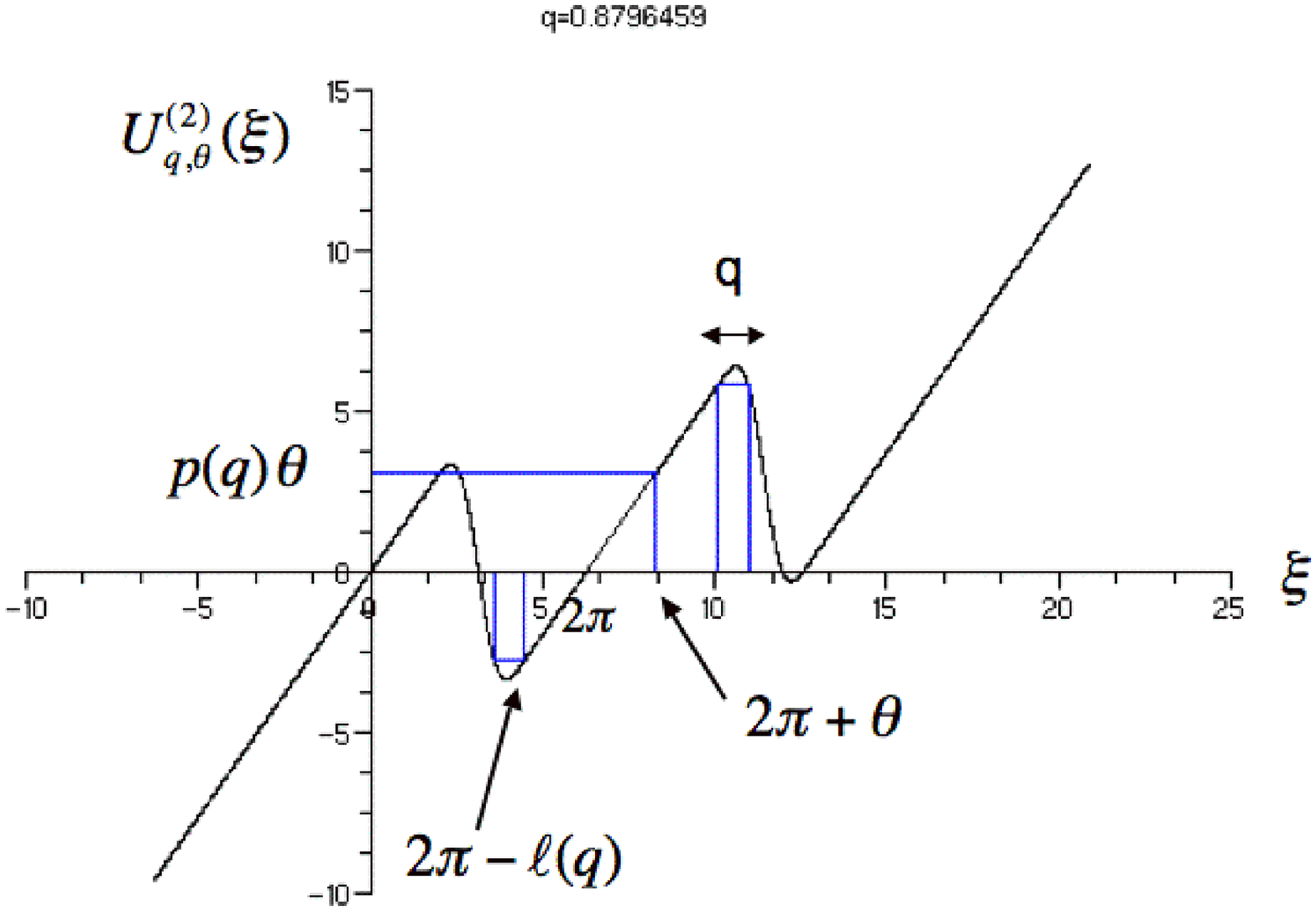}
\end{center}
\caption{\label{graphulin} 
Top panel~:
graph of a solution of (\ref{ad}) defined by (\ref{extens}), corresponding to
a single compacton (we have set $q\approx 0.88$).
Bottom panel~: graph of a solution of (\ref{ad}) defined by (\ref{extens2}),
corresponding to a $2$-compacton
(we have fixed the same value of $q$ and $\theta =2$).
}
\end{figure}

\vspace{1ex}

Note that the remark made here on the existence of compactons 
follows from the numerical results obtained on problem (\ref{ad})-(\ref{pbc})
for wavenumbers $q\leq q_k$. Obtaining an analytical proof of
the existence of compactons remains an open problem. 

\subsection{\label{errs}Error minimization}

In this section we evaluate and improve the numerical precision of 
the computations performed in section \ref{numc}.
Our numerical approach is in the same spirit as former computational methods
developed to analyze the propagation of discrete breathers in nonlinear lattices
\cite{aubryC} (see also \cite{num97}). 

\ve

We consider a chain of $100$ particles with periodic
boundary conditions
\begin{equation}
\label{pbcl}
x_{n+N}(t)=x_n (t), \ \ \ N=100,
\end{equation}
so that the lattice
period is an integer multiple of the wavelength $\lambda = 2\pi / q = 100/m$ for our discrete set of values of $q$.
We fix $a={\| u(.;q)\|}_{\infty}^{-1}$ in (\ref{twper}),
where $u(.;q)$ is the solution of (\ref{ad})-(\ref{pbc}) that was
computed numerically for different values of $q$.
We use the resulting profile as an initial condition in system (\ref{nc}),
which determines a periodic travelling wave solution $X_n = (x_n,\dot{x}_{n})$.
Numerical integrations are performed using the standard ODE solver of the
software package Scilab, 
where we have fixed the relative and absolute error tolerances to $10^{-12}$ and $10^{-14}$ respectively.
For each solution $u(.;q)$, we compute the relative residual error
$$
E(q)=\frac{{\| {\{ X_{n+1}( \mathcal{T}) - X_{n}(0) \}}_n  \|}_{\infty}}{{\| {\{ X_{n}(0) \}}_n  \|}_{\infty}}
$$
with $\mathcal{T}=q\, a^{\frac{1-\alpha }{2}}$, so that
an exact travelling wave solution with velocity $\mathcal{T}^{-1}$ would cancel $E(q)$.
The inverse wave velocity $\mathcal{T}(q)$ is consequently given by
\begin{equation}
\label{invv}
\mathcal{T}(q) = q\, {\| u(.;q)\|}_{\infty}^{\frac{\alpha -1}{2}}.
\end{equation}
Figure \ref{errors} displays the graph of $E(q)$ (dotted line).
With our numerical solution,
this error remains less than $h^2 \approx 10^{-5}$ for $q\geq 6\pi /25$ and grows
when $q$ is further decreased, reaching $3.10^{-4}$ at $q=\pi /50$.
We attribute these larger errors to the sharpness of the travelling wave in the compression
region when $q$ is sufficiently decreased. 
One way to improve the results would consist in discretizing
(\ref{ad})-(\ref{pbc}) with a nonuniform mesh, thinner in the 
compression region. However, in what follows we use a different approach 
which decreases the residual error by several orders of magnitude on the
whole range of wavenumbers $q$.

\ve

We use the Gauss-Newton method \cite{dennis} in order to 
determine a refined initial condition ${\{ X_{n}(0) \}}_{n=1,\ldots ,N}$
minimizing ${\| {\{ X_{n+1}(\mathcal{T}) - X_{n}(0) \}}_n \|}_{2}^2$
(the Newton method is not directly applicable due to the
noninvertibility of the Jacobian matrix at the exact solution).
The improved initial condition obtained in this way for a given wavenumber $q$
will be denoted by $X_{n}(0)=X_{n }^{(q)}$. To compute these solutions,
we initialize the Gauss-Newton iteration 
using the numerical solutions $u(.;q)$
of (\ref{ad})-(\ref{pbc}) obtained previously, since they already 
provide initial conditions quite close to the optima.
To reduce the computational cost, we modify the 
Gauss-Newton method by actualizing the Jacobian matrix
only at some steps of the iteration, when a re-evaluation is
required to decrease the residual error.
The residual error $E(q)$ obtained with this numerical method
drops to $3.10^{-10}$ or less, as shown in figure \ref{errors} (full line).
At the end of the Gauss-Newton iteration ($k$th iteration), the relative difference
between the last two iterates
$r(q)=\frac{{\| {\{ X_{n}^{(k)}(0) - X_{n}^{(k-1)}(0) \}}_n  \|}_{\infty}}{{\| {\{ X_{n}^{(k)}(0) \}}_n  \|}_{\infty}}$
drops below $4.10^{-7}$ (see figure \ref{errors}, dash-dot line).

\begin{figure}[!h]
\psfrag{q}[0.9]{ $q$}
\psfrag{E}[1][Bl]{Errors}
\begin{center}
\includegraphics[scale=0.4]{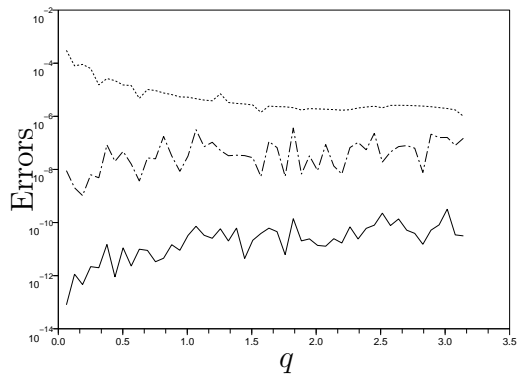}
\end{center}
\caption{\label{errors} 
Graphs of the relative errors $E(q)$ (semi-logarithmic scale),
for the initial condition computed by discretizing (\ref{ad})-(\ref{pbc}) (dotted line) 
and the refined initial condition computed with the Gauss-Newton method
(full line). The dash-dot line gives the relative difference $r(q)$
(in supremum norm) between last two iterates at the end of the 
Gauss-Newton iteration.
}
\end{figure}

\ve

Above a critical value $q_c \approx 0.9$, the initial conditions computed as
indicated above yield travelling waves that 
remain practically unchanged over very long times
when propagating along the lattice (see figure \ref{stabletw} for an example). 
The situation changes drastically below $q_c$ because the travelling waves
become unstable, yielding an amplification at exponential rate
of the initially small errors made on the initial condition. 
This situation is described in figure \ref{unstabletw} for $q=7\pi / 25$.
The travelling wave is destroyed by a instability at $t \approx 360 \approx 303\, \mathcal{T}$,
which rapidly drives the system into a disordered regime.
These instabilities will be described in more detail in the next section.

\begin{figure}[!h]
\psfrag{n}[0.9]{ $n$}
\psfrag{x}[1][Bl]{ $x_n (t)$}
\psfrag{t}[0.9]{ $t$}
\begin{center}
\includegraphics[scale=0.5]{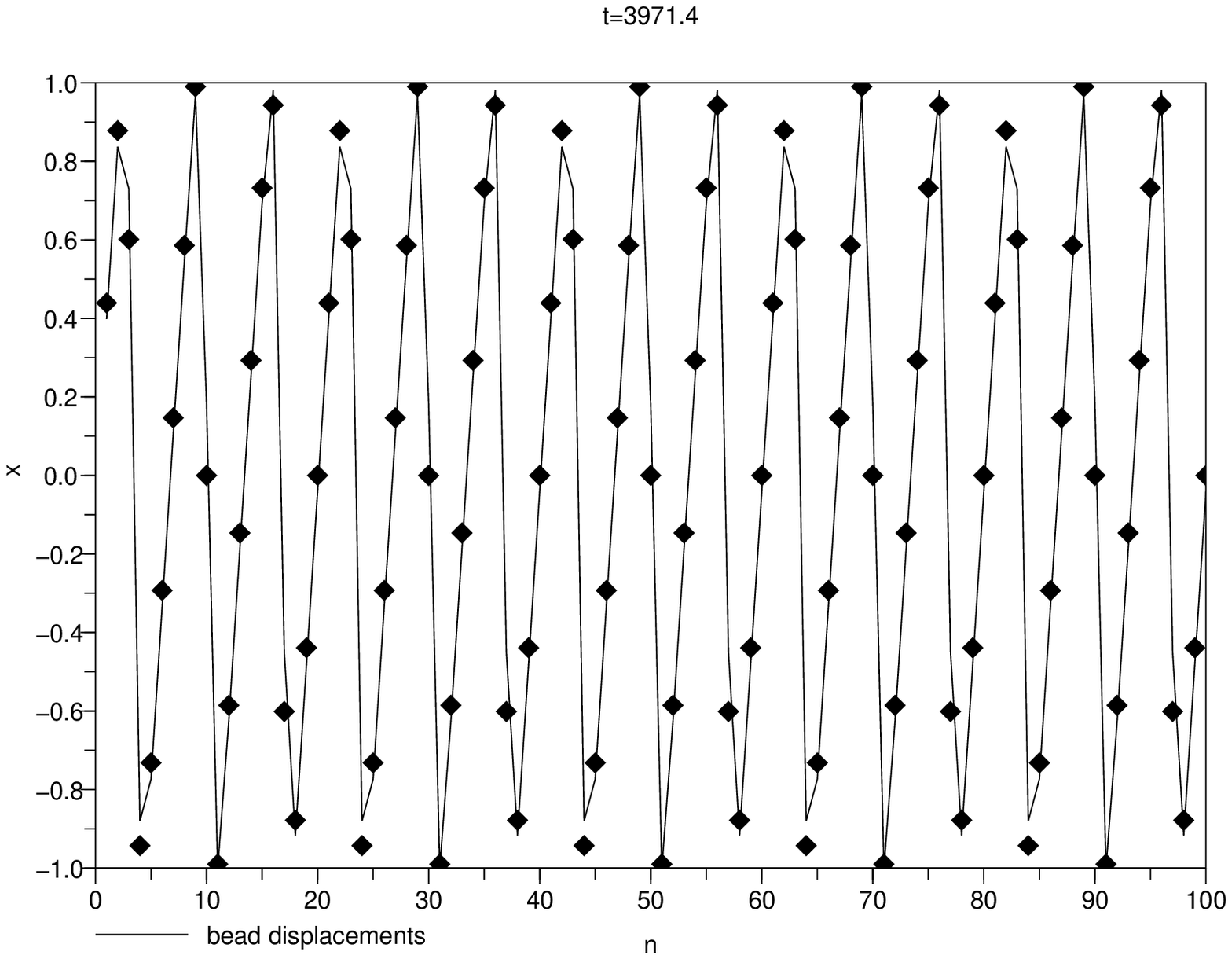}
\includegraphics[scale=0.5]{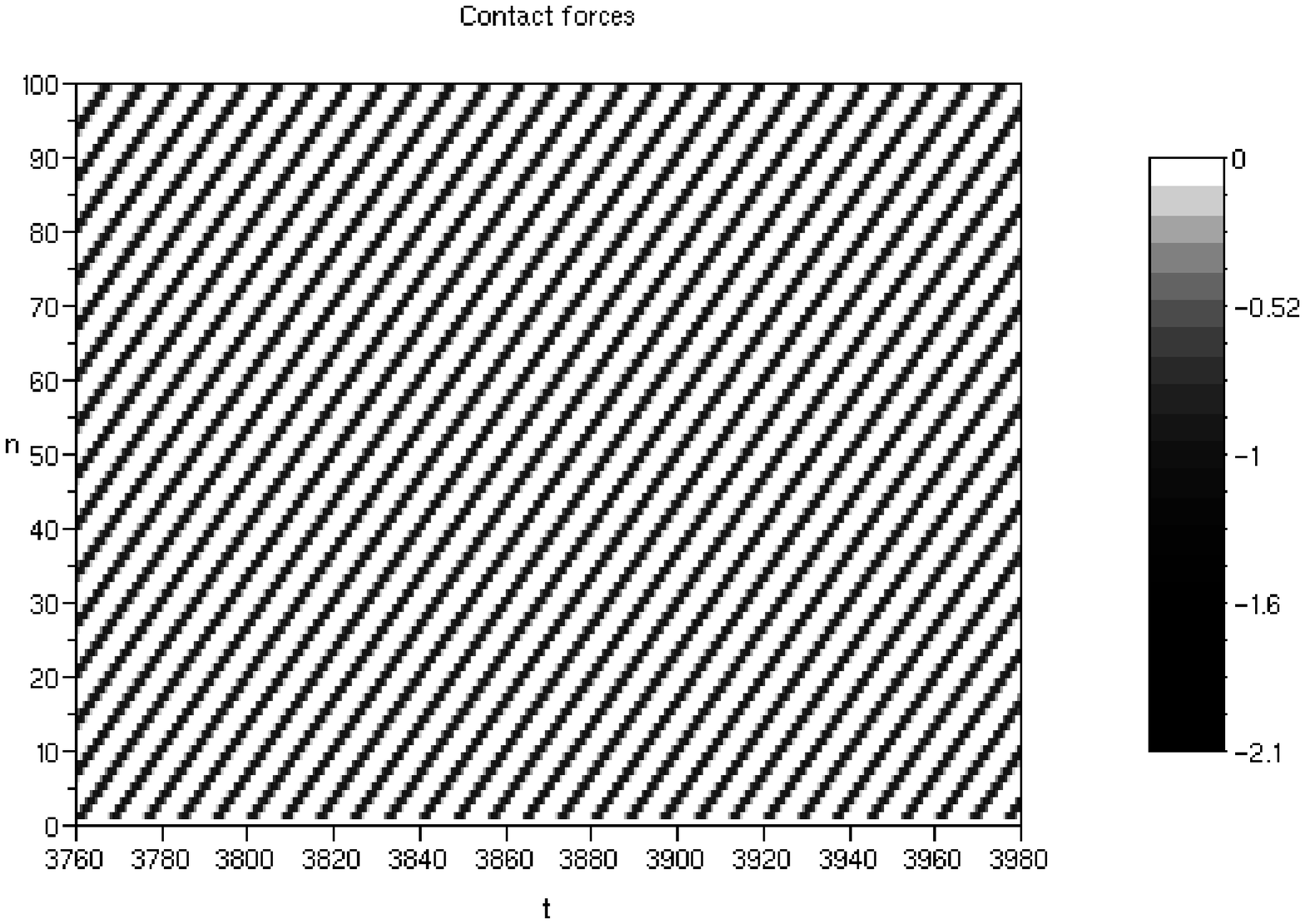}
\end{center}
\caption{\label{stabletw} 
Solution of (\ref{nc}) computed over long times, for
the initial condition determined by Gauss-Newton minimization when
$q=3\pi / 10 \approx 0.94$. The inverse wave velocity is $\mathcal{T} \approx 1.20$.
Upper plot~: bead displacements generated by the travelling wave at $t\approx 3971$ (black curve), compared with
the initial condition at $t=0$ (dots). Lower plot~: spatiotemporal evolution of the interaction forces 
$V^\prime (x_{n+1}(t)-x_n(t))$ in grey levels, for $t\in [3760,3980]$.
}
\end{figure}

\begin{figure}[!h]
\psfrag{n}[0.9]{ $n$}
\psfrag{x}[1][Bl]{ $x_n (0)$}
\psfrag{t}[0.9]{ $t$}
\begin{center}
\includegraphics[scale=0.5]{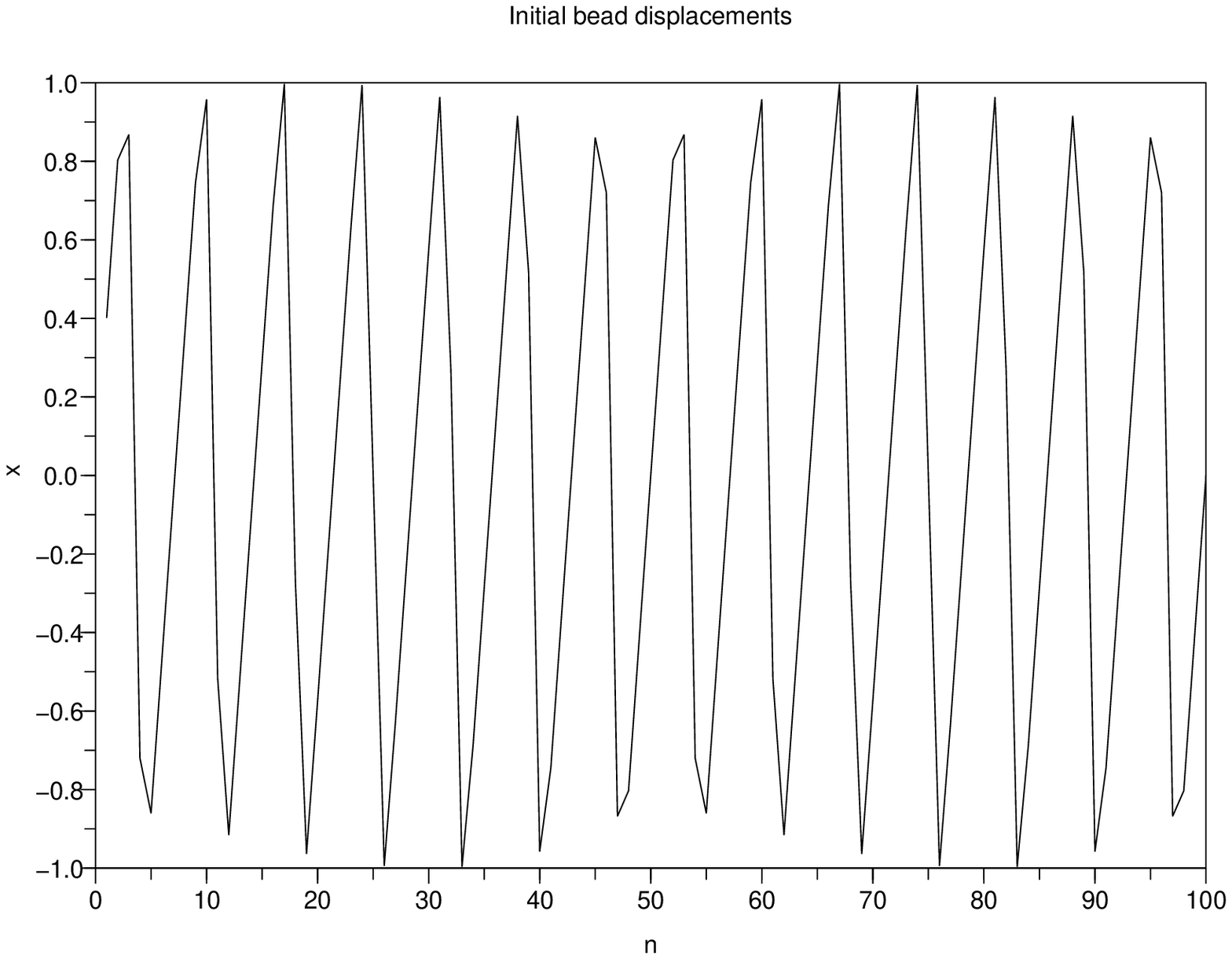}
\includegraphics[scale=0.5]{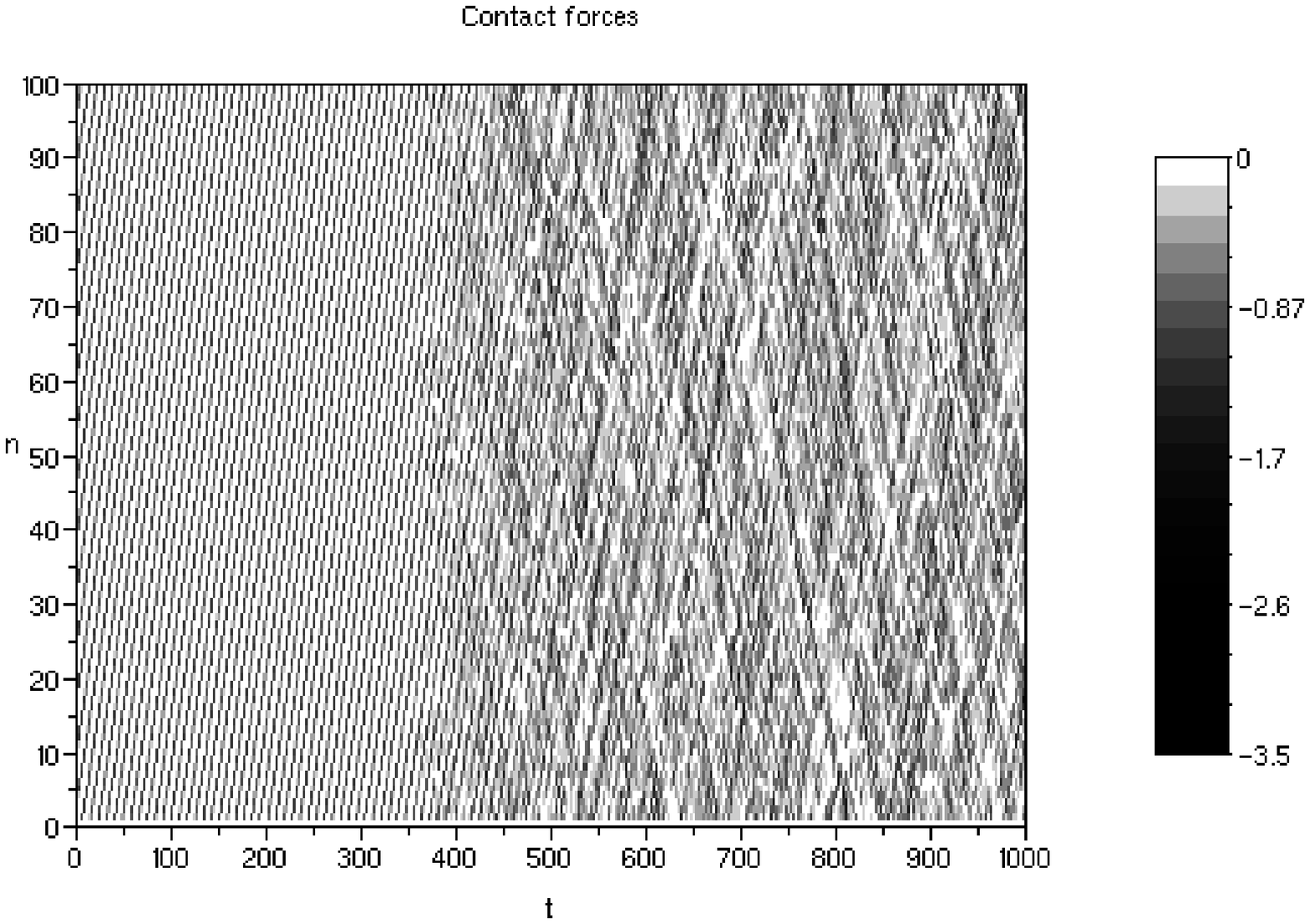}
\end{center}
\caption{\label{unstabletw}
Upper plot : initial bead displacements corresponding to
the numerical travelling wave solution with $q=7\pi / 25 \approx 0.88$. The inverse wave velocity is $\mathcal{T} \approx 1.19$.
Lower plot : spatiotemporal evolution of the interaction forces in grey levels.}
\end{figure}

\subsection{\label{dyn}Wave instabilities}

In this section we consider the dynamical equation (\ref{nc}) of the granular chain with 
$N=100$ particles and periodic boundary conditions, and
numerically study the stability of the travelling wave solutions of section \ref{errs}.
We recall that these travelling waves take the form (\ref{twper}),
where $u(.;q)$ is the solution of (\ref{ad})-(\ref{pbc}) computed numerically for a given value of $q$,
and where we fix $a={\| u(.;q)\|}_{\infty}^{-1}$ to renormalize the solution.
We rewrite equation (\ref{nc}) in the compact form $\dot{X}=f(X)$ with
$X(t)=(X_1 (t), \ldots , X_N (t) )^t  \in \mathbb{R}^{2N}$ and
$X_n = (x_n,\dot{x}_{n})$, and denote by $X^{(q)}(t)$
the solutions corresponding to travelling waves.

\ve

Since the solution $X^{(q)}(t)$ is time-periodic with period $\tau (q)= 2\pi a^{\frac{1-\alpha }{2}}$, its stability
can be analyzed using Floquet theory. Let us 
denote by $\mathcal{R}(t;t_0)$ the resolvent matrix of the linearized equation
$\dot{X}=Df(X^{(q)}(t))\, X$ and $\Phi_q = \mathcal{R}( \tau(q);0)$ the associated monodromy matrix.
The solution $X^{(q)}(t)$ is called spectrally stable if all eigenvalues of
${\Phi}_q$ lie on the unit circle, and $X^{(q)}(t)$ is unstable if there exists an eigenvalue of modulus strictly larger than $1$
(see e.g. \cite{chicone,abram}). In other words, the spectral radius $\rho (\Phi_q  )$ determines if $X^{(q)}(t)$ is unstable
($\rho (\Phi_q  ) >1$) or spectrally stable ($\rho (\Phi_q  ) =1$).

\ve

Since $X^{(q)}(t)$ is a periodic travelling wave and periodic boundary conditions are used, 
one can equivalently determine wave stability using a monodromy matrix modulo shifts. This
approach reduces the length of numerical integration significantly at low wavenumbers, since
the linearized equations are integrated over an interval of length $\mathcal{T}(q)=\frac{q}{2\pi } \tau (q)$. Moreover, it
describes more conveniently the growth
of perturbations in a frame moving with the waves in the linear approximation. 

To be more precise,
let us denote by $\mathcal{S}$ the spatial shift
$\mathcal{S}\, X = (X_2 , \ldots , X_N ,X_1 )^t $ associated to periodic boundary conditions.
The initial condition $X^{(q)}(0)$ corresponding to a travelling wave solution is a
fixed point of the map $\mathcal{N}_q$ defined by
$$\mathcal{N}_q ( X(0) )=\mathcal{S}\, X( \mathcal{T}(q)) ,$$
where the inverse wave velocity $\mathcal{T}(q)$ is given by (\ref{invv}).
Let us introduce the monodromy matrix modulo shift 
$\mathcal{F}_q = D\mathcal{N}_q ( X^{(q)}(0)  )$, which takes the form
$$
\mathcal{F}_q = \mathcal{S}\, \mathcal{R}(\mathcal{T}(q) ; 0).
$$ 
Recalling that $q=2m\pi / N$ ($1\leq m \leq N/2$), we have
$\tau (q)= \frac{N}{m}\, \mathcal{T}(q)$. Thanks to the invariance
$\mathcal{S}\, X^{(q)}( t+\mathcal{T}(q))=X^{(q)}(t)$, one has
$\mathcal{S}\, \mathcal{R}(t+\mathcal{T} ; t_0)=\mathcal{R}(t ; t_0 - \mathcal{T})\, \mathcal{S}$,
which implies
$$
\mathcal{F}_q^N = \Phi_q^m.
$$
Consequently one has $\rho (\Phi_q)={\rho ( \mathcal{F}_q)}^{N/m} $, i.e. 
$X^{(q)}(t)$ is unstable for 
$\rho ( \mathcal{F}_q) >1$ and spectrally stable if $\rho ( \mathcal{F}_q) =1$.

\ve

The spectral radius of $\mathcal{F}_q$ is plotted in figure \ref{floquet} as a function of $q$
(these results are obtained using the software package Scilab, for which matrix eigenvalue computations
are based on the Lapack routine DGEEV). These results show that
the travelling waves are unstable below the critical value $q_c \approx 0.9$.
The growth rate of the instability in the linear approximation
is maximal for $q = q_m \approx 0.56$, and decreases 
rapidly when $q \rightarrow 0$ and $q \rightarrow q_c$.
Figure \ref{eigen} shows the evolution of the eigenvalues of $\mathcal{F}_q$ when $q$ is varied 
in this parameter region. One can see that the number of unstable modes rapidly
increases below $q=q_c$.

According to figure \ref{number} we have $N(q_c)\approx 3$, i.e. instabilities show up 
when the average number $N(q)$ of adjacent interacting beads becomes $\geq 3$.
This case corresponds to the maximal number of adjacent interacting beads becoming $\geq 4$.
This phenomenon can be 
intuitively understood through the results of \cite{sv}, 
where the travelling wave stability was numerically established for a three-ball chain
(with fixed center of mass) using Poincar\'e sections. However, as it follows from our numerical results,
the interactions of the three-ball packets with additional beads generate instabilities
as soon as $q < q_c$.

The weakness of the linear instability for $q \approx 0$ is consistent with
the convergence of the travelling waves towards the classical Hertzian solitary wave 
(section \ref{longw}), given the fact that the Hertzian solitary wave appears
structurally robust in dynamical simulations. 

\ve

Now let us describe how these instabilities act on the wave profiles.
When $q$ is slightly below $q_c$, the early stage of the instability generates two large regions
of average compression separated by two large regions of average extension (see figure \ref{unstabletw2}).
This transitory state disappears rapidly and a disordered regime settles, as shown
previously in figure \ref{unstabletw}. When $q$ reaches $q_m$, the large-scale transitory state
is not visible any more and the system directly evolves from the regular travelling wave
towards the disordered regime. The instability of the travelling waves
is observed up to the long wave limit, where it
manifests differently from the instability at $q\approx q_c$. A slow dispersion of the compression pulse occurs
(see figure \ref{shockbis}, top right plot), and subsequent collisions of dispersive waves with the shock generate
a fast instability. A closer view of the transition from dispersive travelling waves
to spatiotemporal disorder is provided by the lower plot of figure \ref{shockbis}.
In these different examples, we interpret
the abrupt transition to a disordered regime as the result of
the large number of unstable modes.

It is interesting to note that the disordered regime that occurs after these instabilities
coexists with some partial order, because transitory large-scale organized structures
appear intermittently. This phenomenon is illustrated by figure \ref{structures} (lower plots). 
These structures seem to result from the interactions of travelling waves 
that coexist with the disordered dynamics. This is shown in the
upper plot of figure \ref{structures}, which provides the spatiotemporal evolution 
of bead displacements in grey levels. The region corresponding to $t\in [0,140]$ 
(nearly uniform at the scale of the figure) corresponds to the propagation of
an almost unperturbed periodic travelling wave with wavenumber
$q=9\pi / 50 \approx q_m$, and spatiotemporal disorder occurs
for $t\geq 140$ after an instability. In this region, 
the presence of stripes reveals
counter-propagating travelling waves.

\begin{figure}[!h]
\psfrag{q}[0.9]{ $q$}
\psfrag{R}[1][Bl]{ $\rho  ( \mathcal{F}_q)$}
\begin{center}
\includegraphics[scale=0.35]{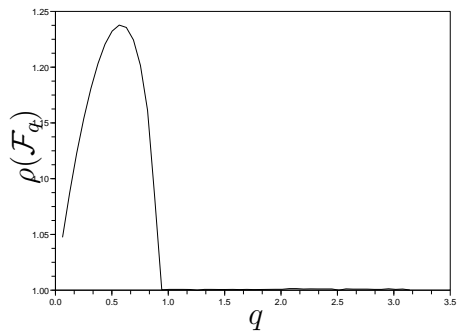}
\end{center}
\caption{\label{floquet} 
Spectral radius of the monodromy matrix modulo shifts $\mathcal{F}_q$.}
\end{figure}

\begin{figure}[!h]
\psfrag{q}[0.9]{ $q$}
\psfrag{R}[1][Bl]{ $\rho  ( \mathcal{F}_q)$}
\begin{center}
\includegraphics[scale=0.35]{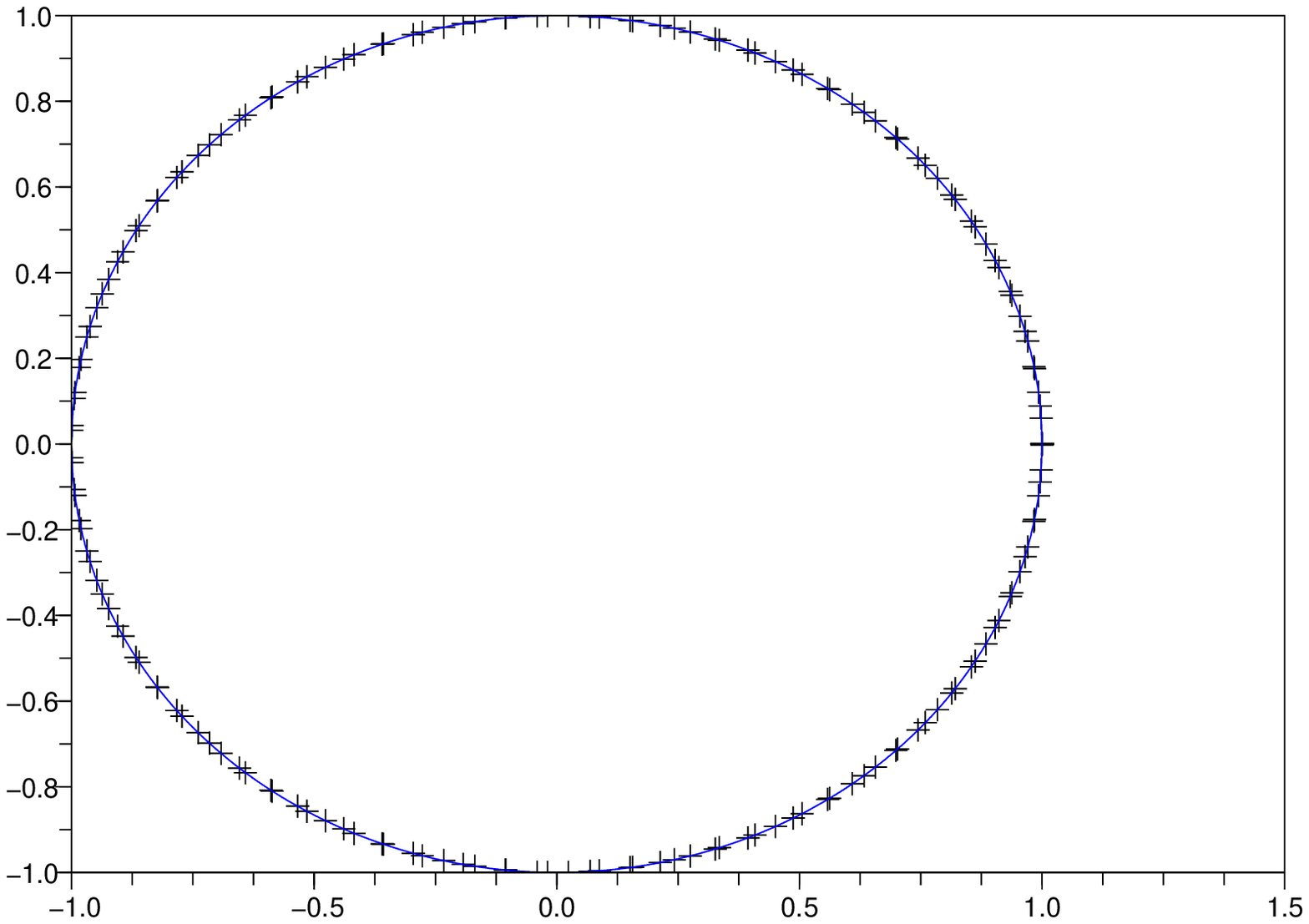}
\includegraphics[scale=0.35]{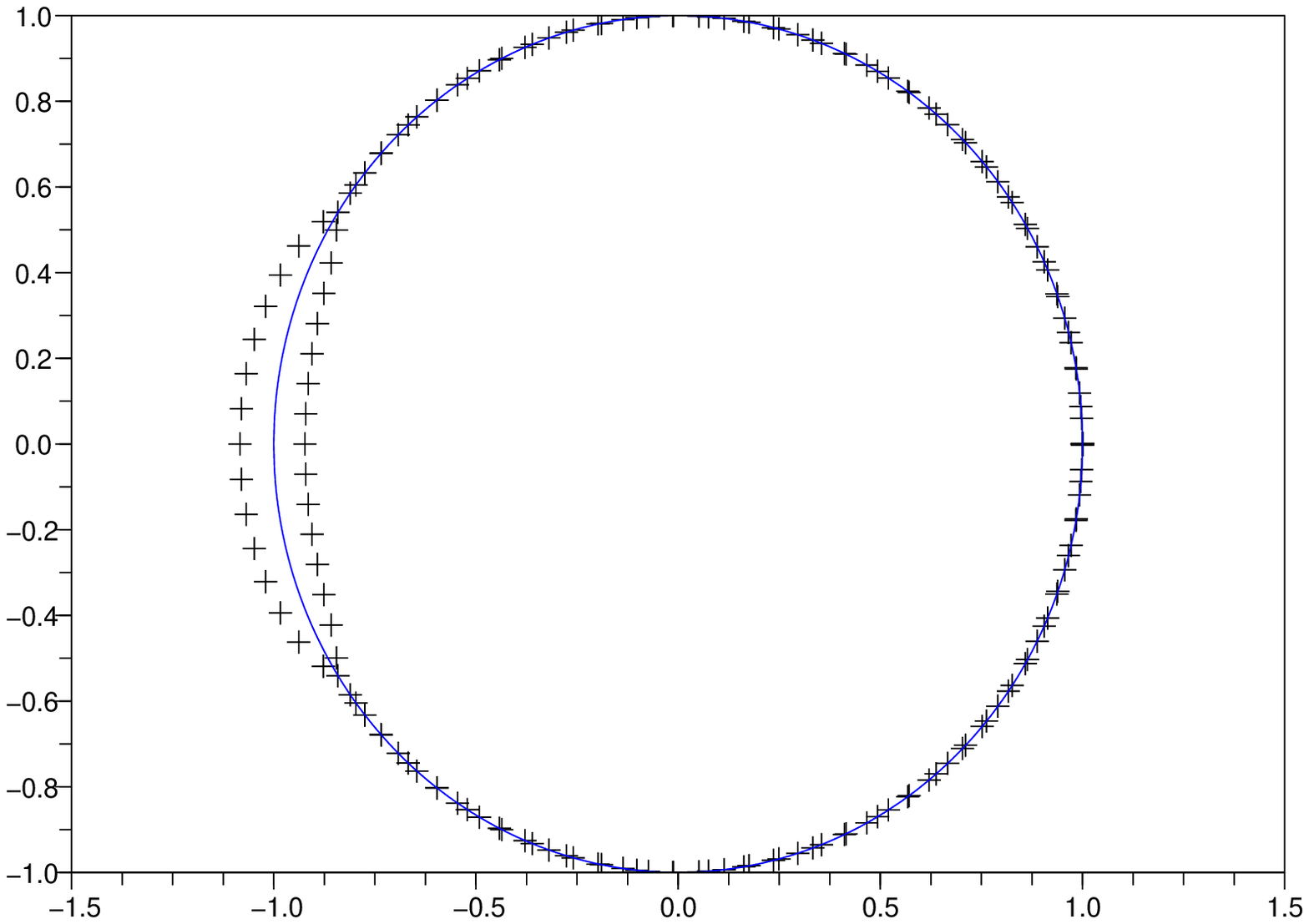}
\includegraphics[scale=0.35]{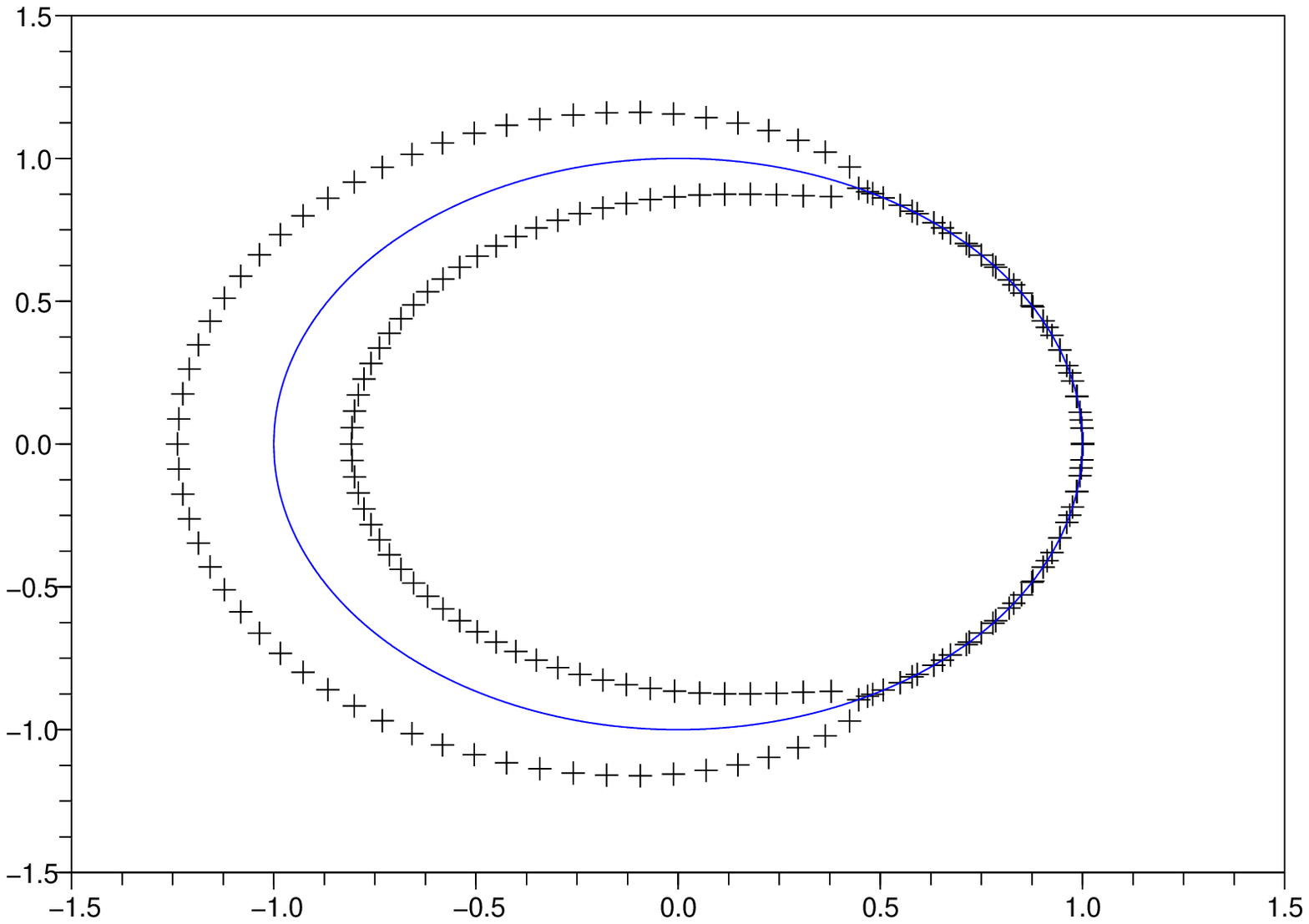}
\includegraphics[scale=0.35]{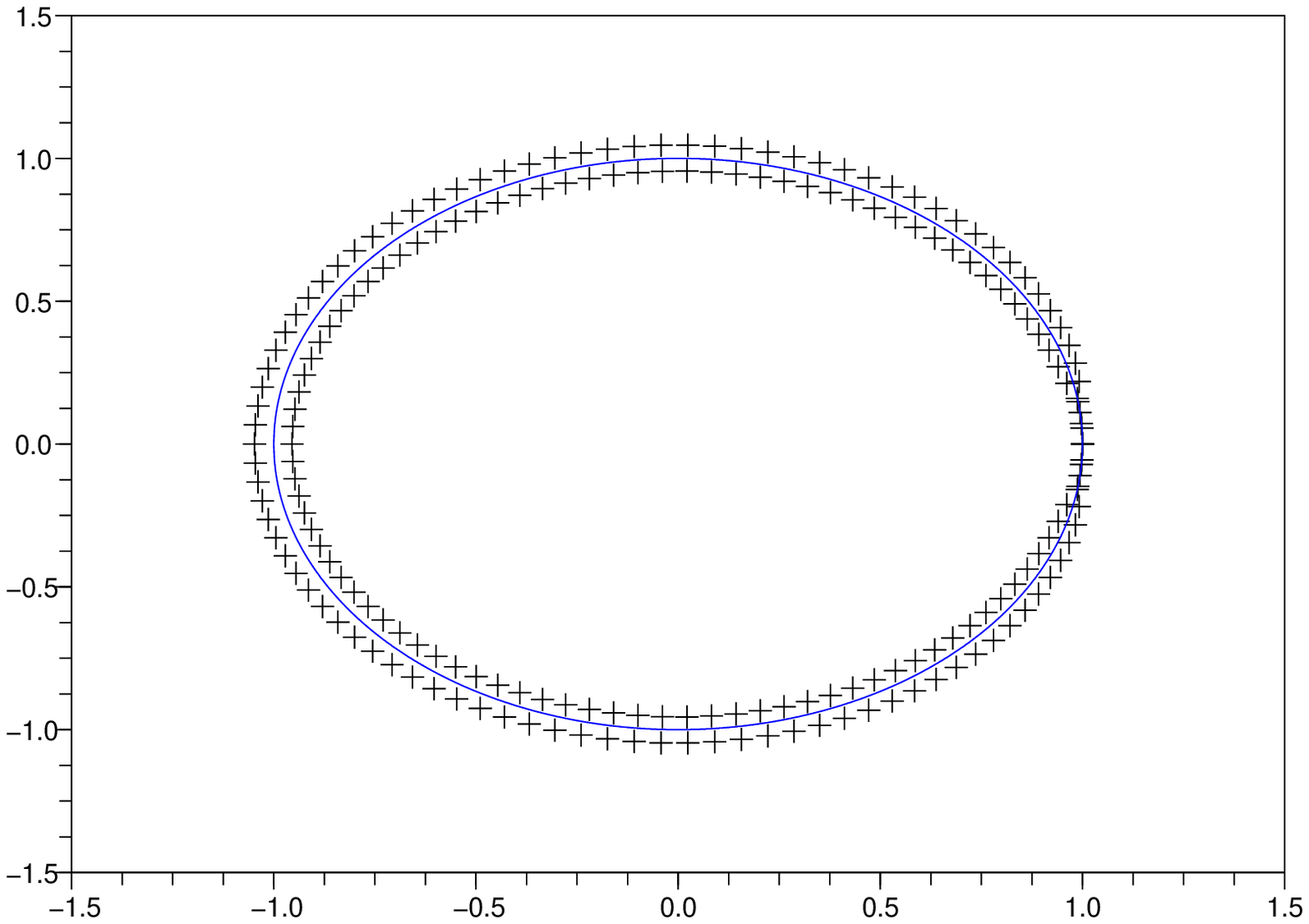}
\end{center}
\caption{\label{eigen} 
Eigenvalues of the monodromy matrix modulo shifts $\mathcal{F}_q$ (marks), plotted in the complex plane for
decreasing values of $q$ :
$q=\frac{3\pi}{10}$ (upper left plot), $q=\frac{7\pi}{25}$ (upper right plot),
$q=\frac{9\pi}{50}$ (lower left plot), $q=\frac{\pi}{50}$ (lower right plot).
The unit circle is also represented.}
\end{figure}

\begin{figure}[!h]
\psfrag{n}[0.9]{ $n$}
\psfrag{x}[1][Bl]{ $x_n (t)$}
\begin{center}
\includegraphics[scale=0.35]{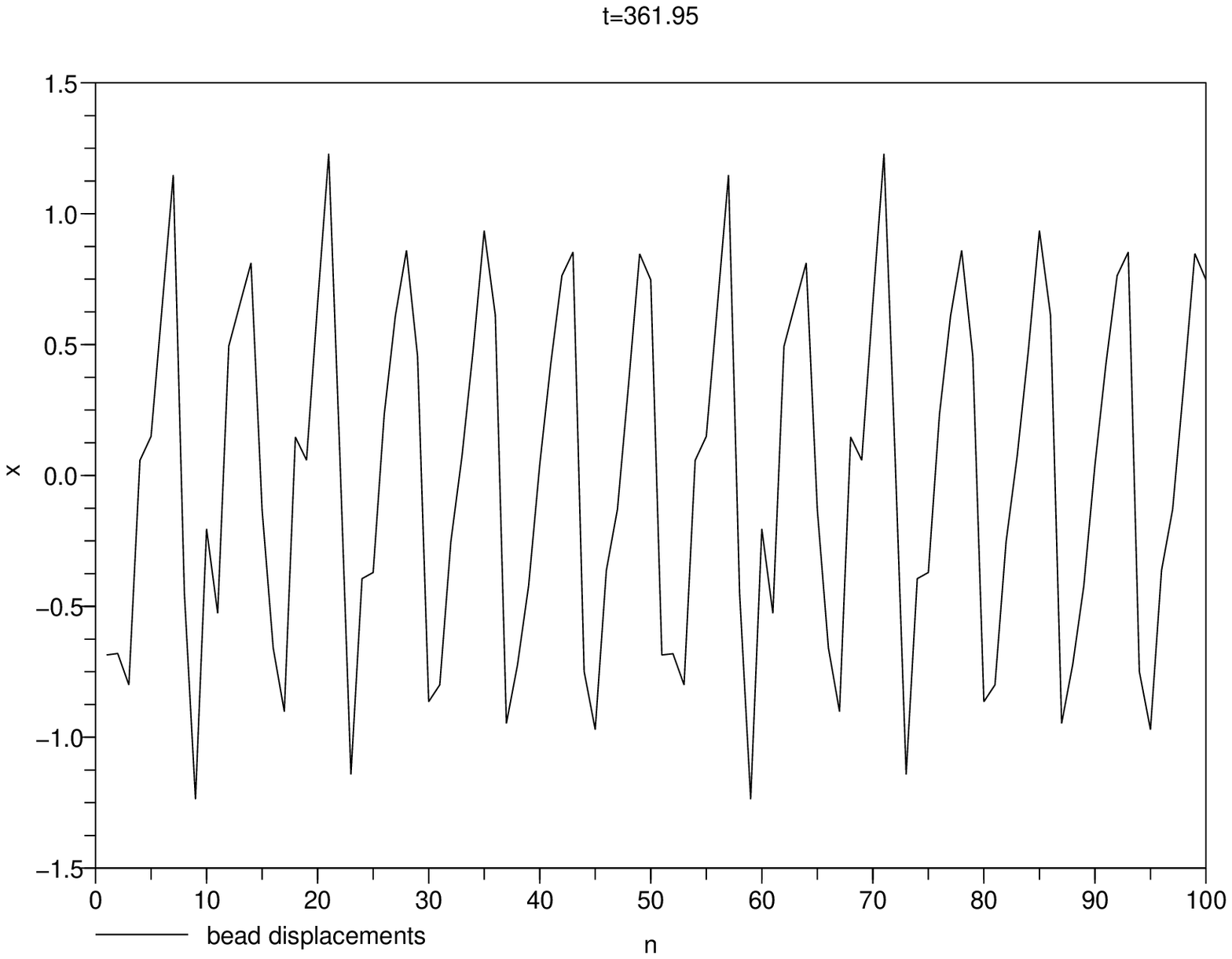}
\includegraphics[scale=0.35]{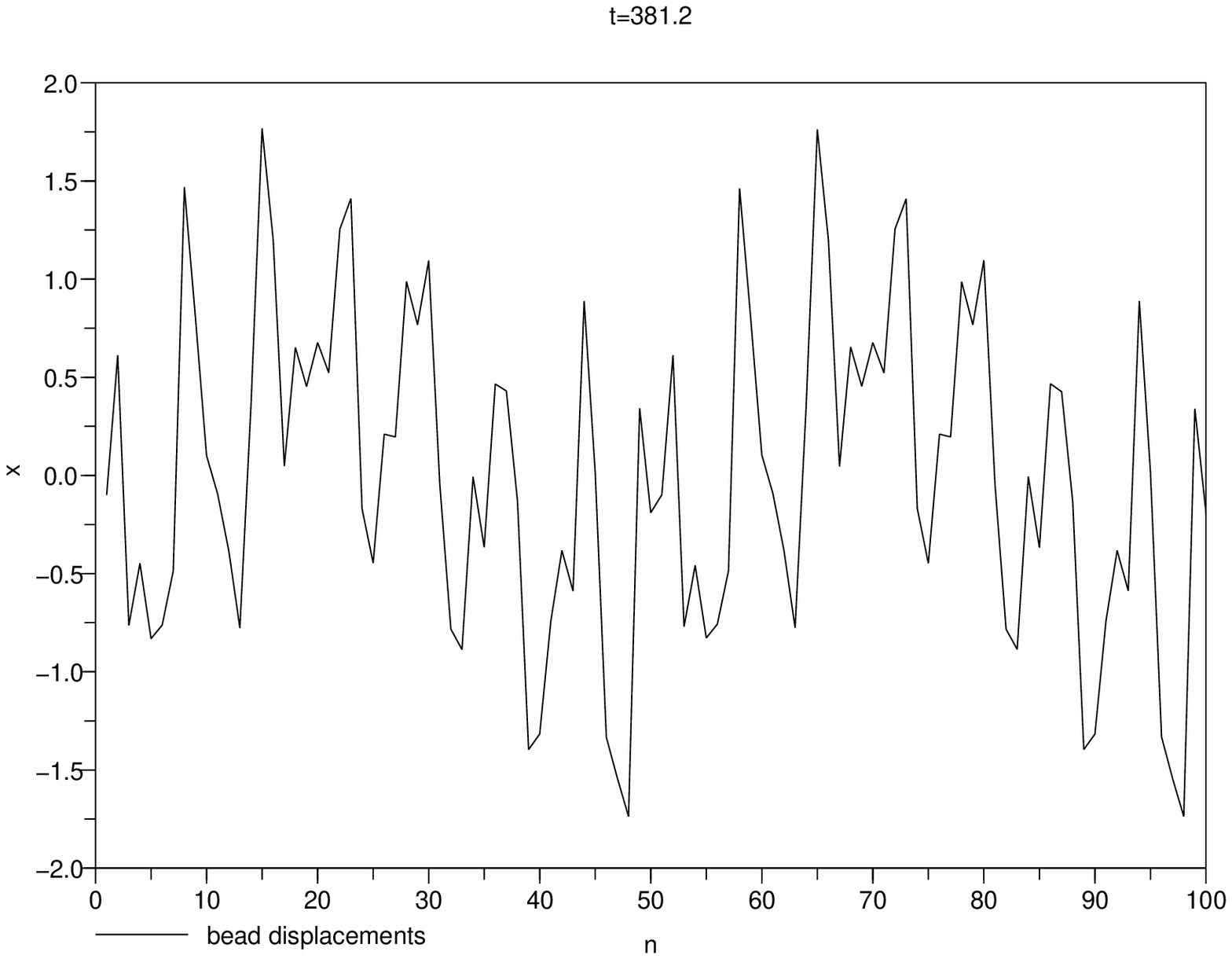}
\includegraphics[scale=0.35]{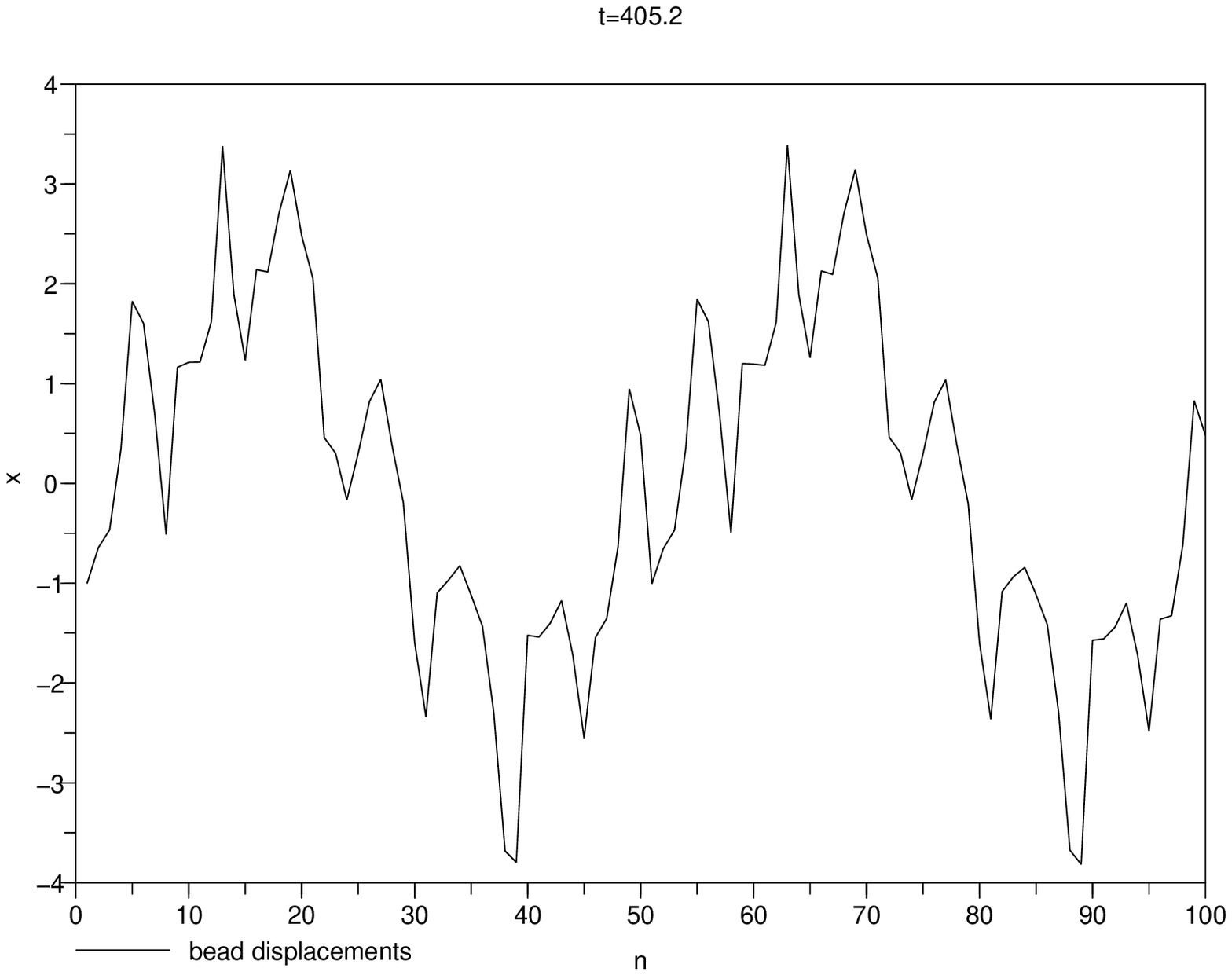}
\end{center}
\caption{\label{unstabletw2}
Bead displacements for the same travelling wave as in figure \ref{unstabletw}
($q=7\pi / 25 \approx 0.88$), showing the growth of an instability
at three different times $t_1 \approx 362$, $t_2 \approx 381$ and $t_3 \approx 405$.}
\end{figure}

\begin{figure}[!h]
\psfrag{n}[0.9]{ $n$}
\psfrag{x}[1][Bl]{ $x_n (t)$}
\psfrag{t}[0.9]{ $t$}
\begin{center}
\includegraphics[scale=0.35]{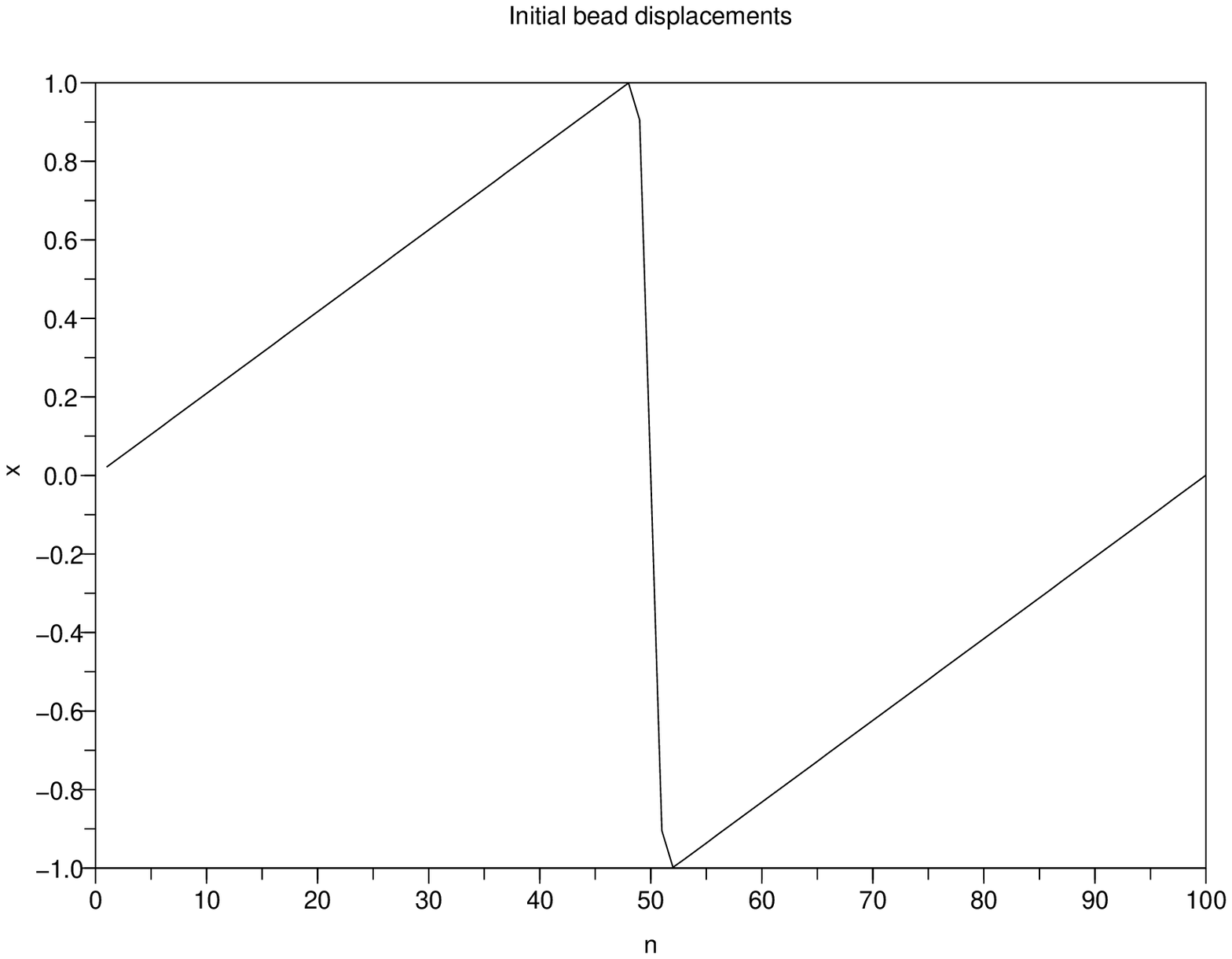}
\includegraphics[scale=0.35]{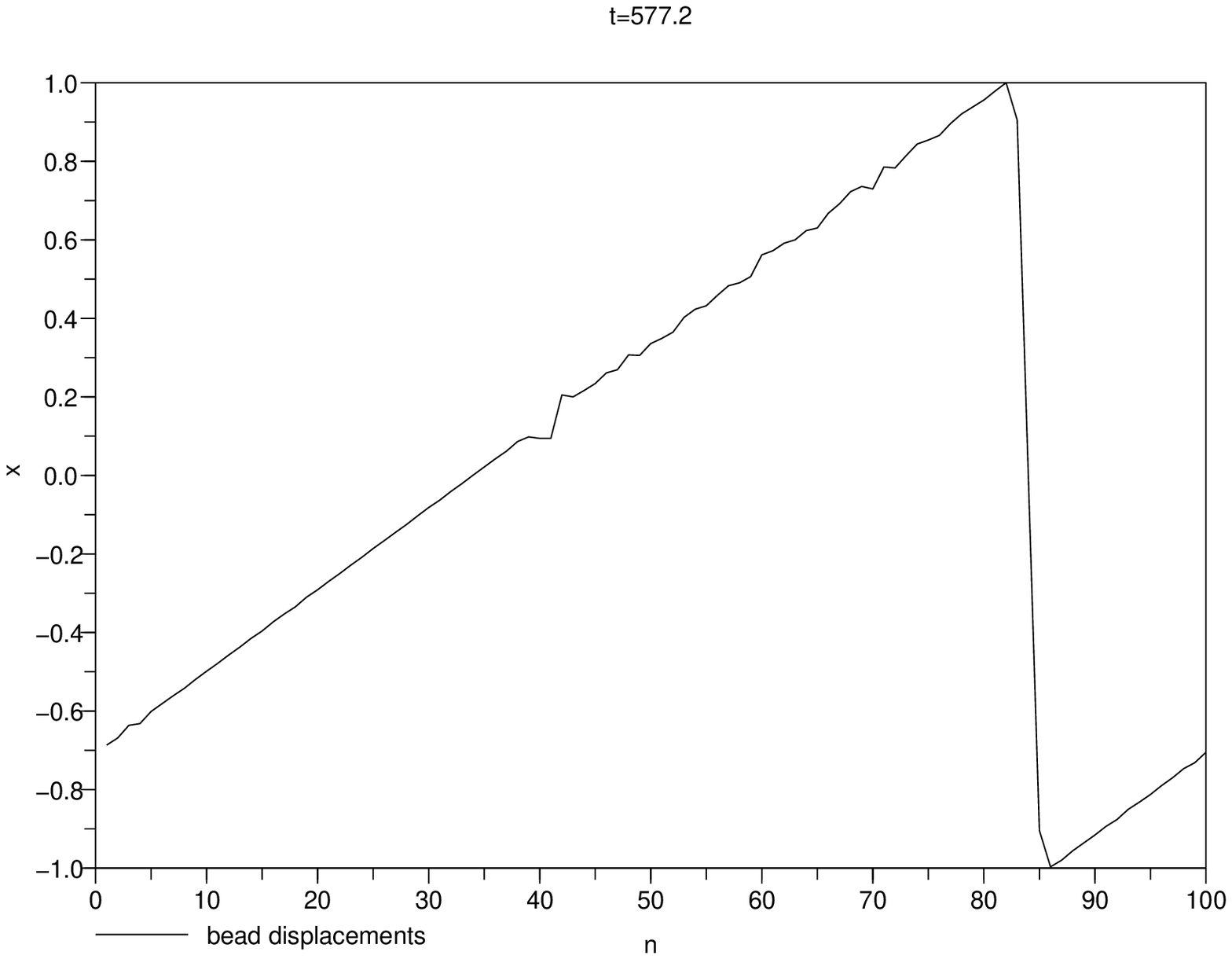}
\includegraphics[scale=0.4]{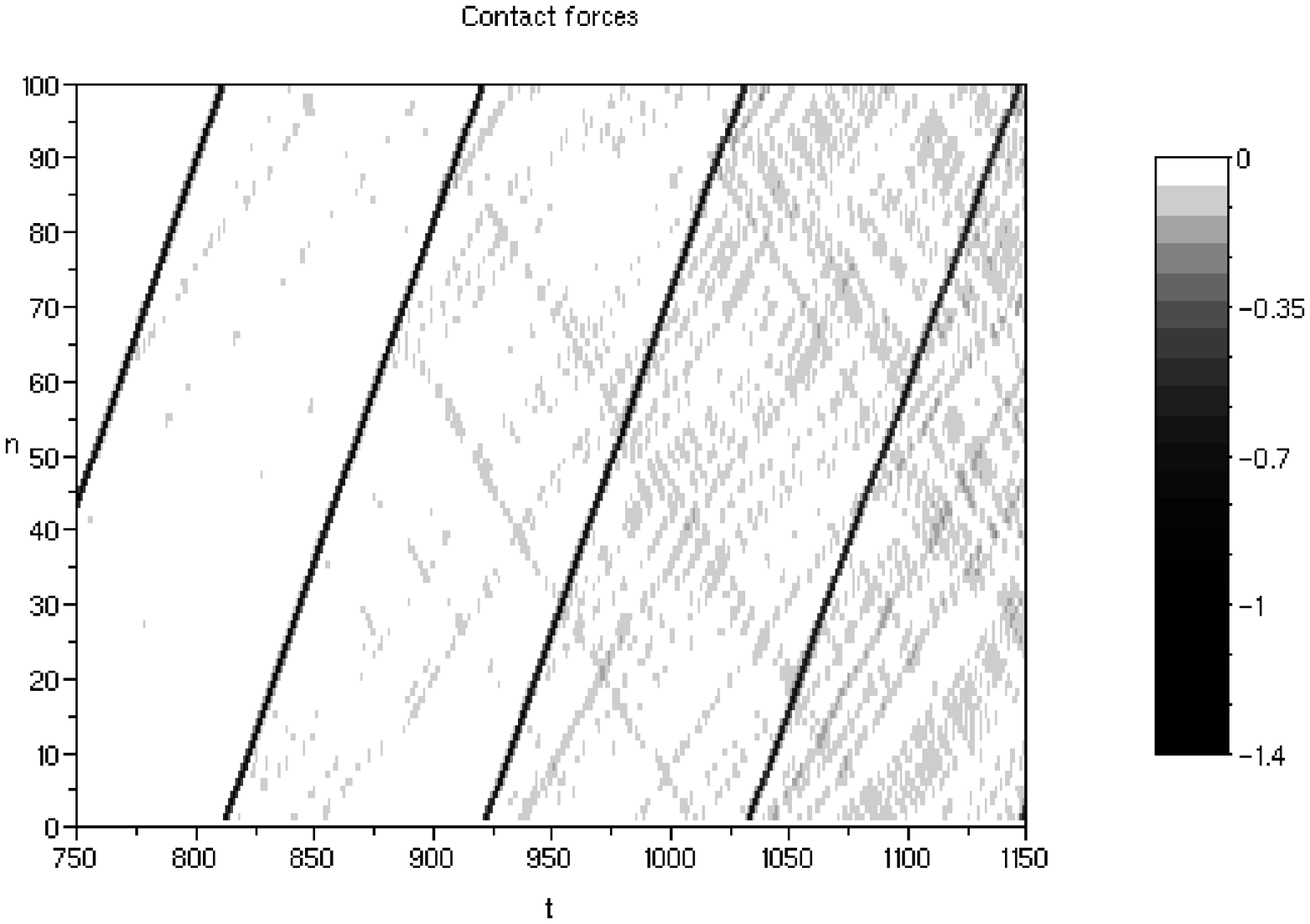}
\end{center}
\caption{\label{shockbis}
Top left plot : initial bead displacements at $t=0$ corresponding to
the travelling wave solution with $q=\pi / 50$. 
Top right plot :  bead displacements at time $t\approx577$
(the inverse wave velocity is $\mathcal{T} \approx 1.08$).
Lower plot : spatiotemporal evolution of the interaction forces in grey levels
for $t\in [750,1150]$, 
near the onset of an instability.
}
\end{figure}

\begin{figure}[!h]
\psfrag{n}[0.9]{ $n$}
\psfrag{x}[1][Bl]{ $x_n (t)$}
\psfrag{t}[0.9]{ $t$}
\begin{center}
\includegraphics[scale=0.4]{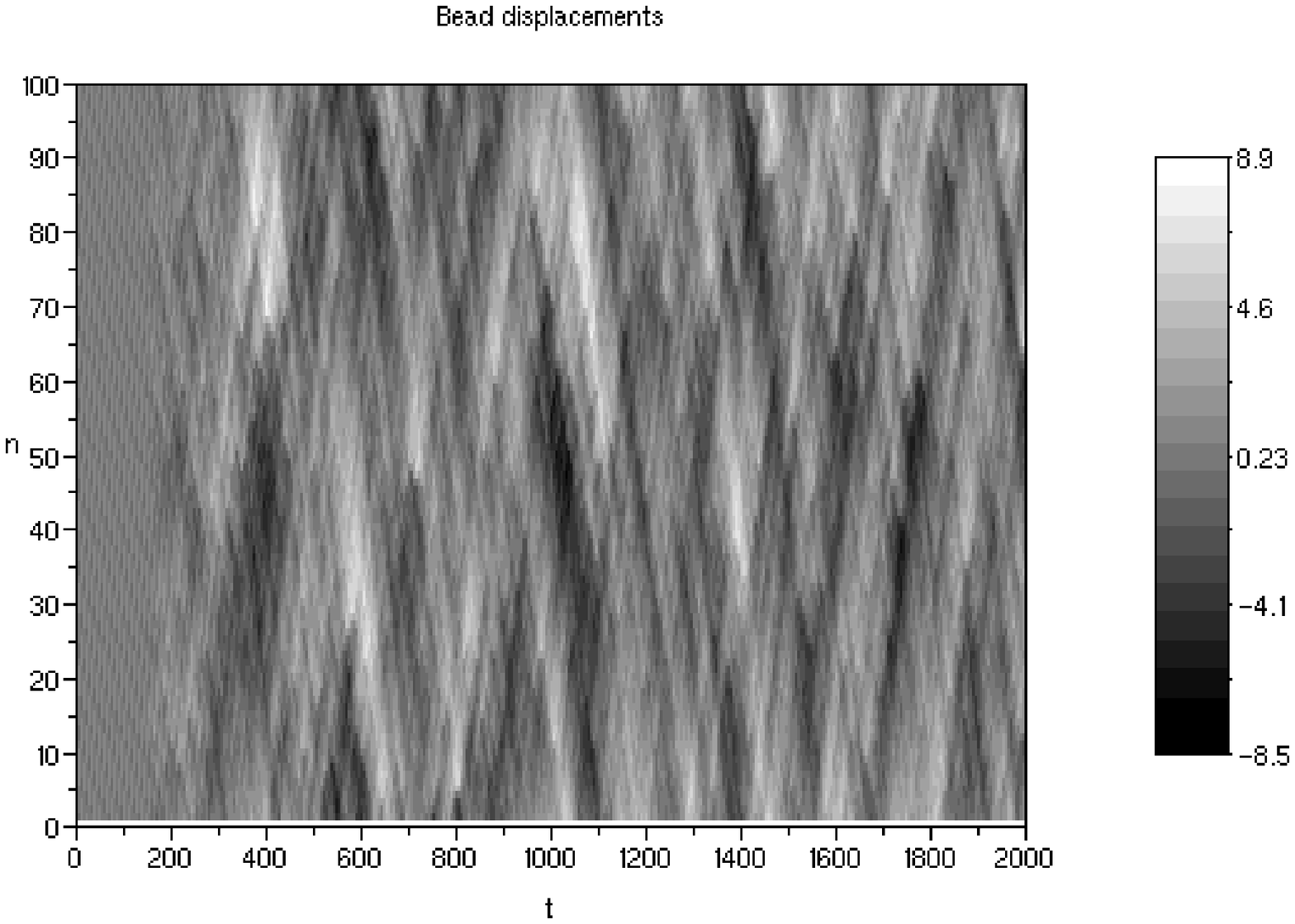}
\includegraphics[scale=0.35]{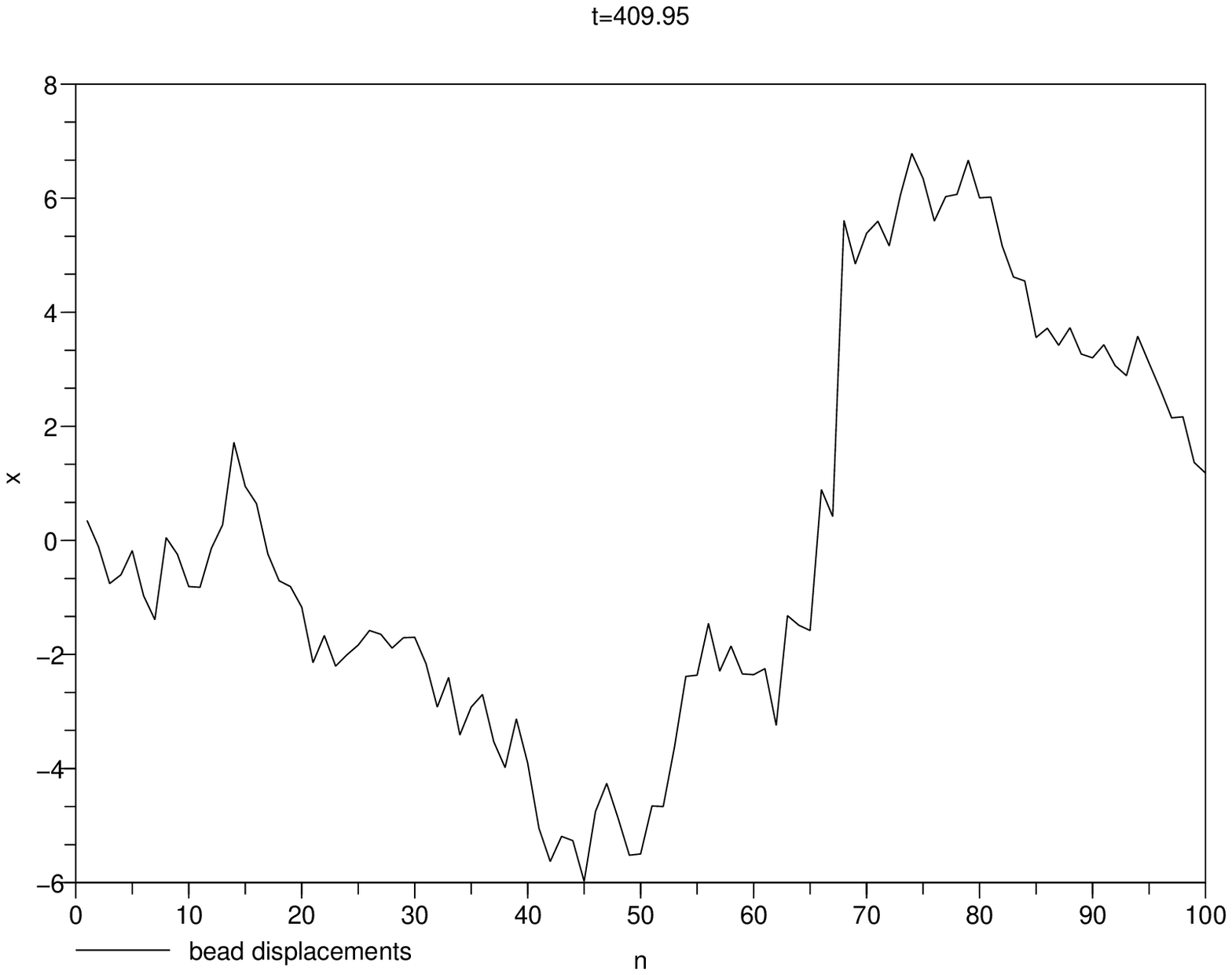}
\includegraphics[scale=0.35]{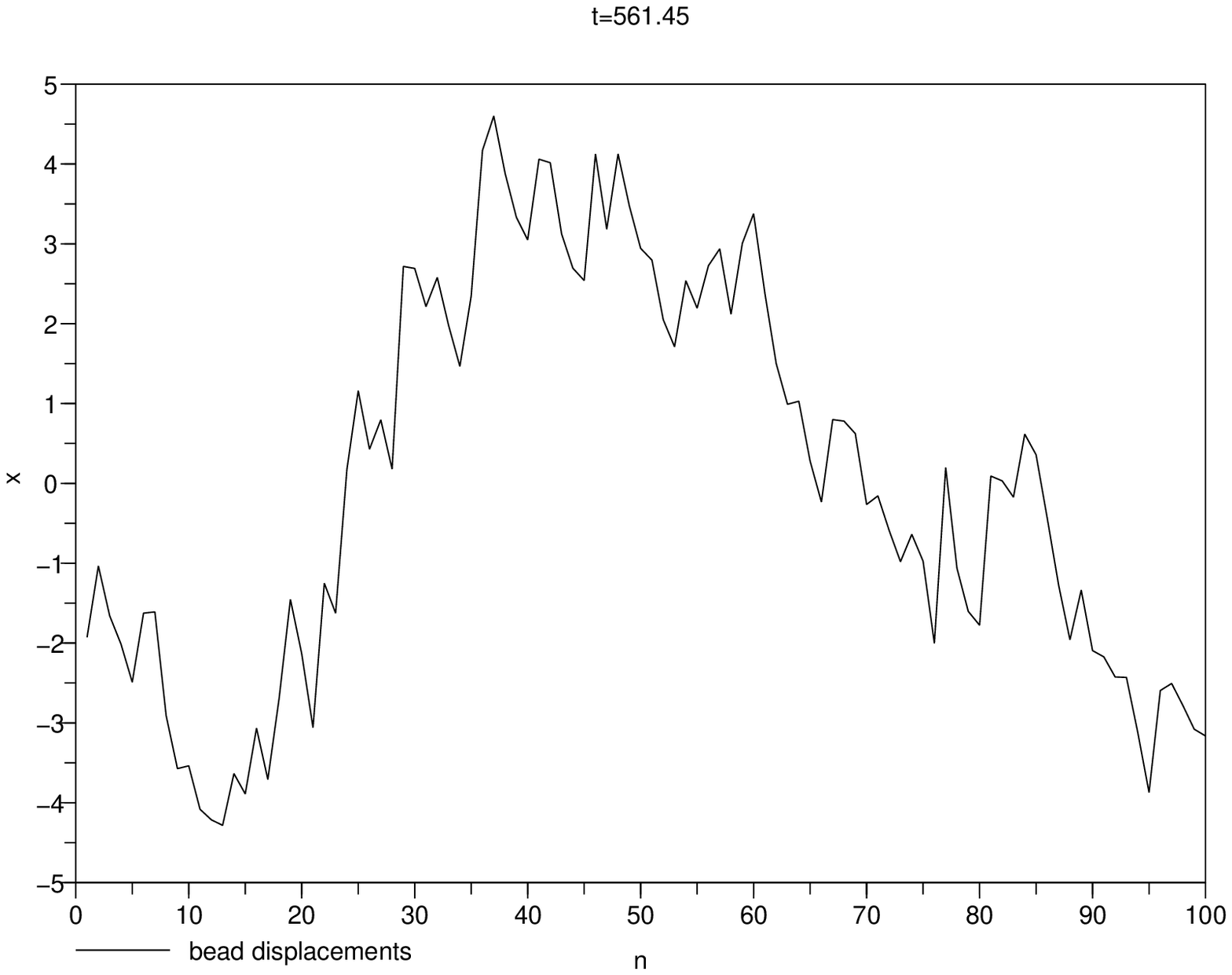}
\end{center}
\caption{\label{structures} 
Case of a highly unstable travelling wave with $q=9\pi / 50 \approx 0.56$ 
(the inverse wave velocity is $\mathcal{T} \approx 1.14$).
Upper plot~:  spatiotemporal evolution of the bead displacements in grey levels.
Lower plots~: bead displacements at $t\approx 410$ (left plot)
and $t\approx 561$ (right plot), revealing an intermittent large-scale organized structure.
}
\end{figure}

Above the critical value $q_c \approx 0.9$,
we have observed very slow modulational instabilities  
by integrating (\ref{nc}), starting from small random perturbations of travelling waves 
with specific wavenumbers. However, the above Floquet analysis is not sufficiently
precise to correctly account for these very small instabilities in the linear regime, 
because they may be overhelmed by numerical errors.

\section{\label{disc}Discussion}

\vspace{1ex}

Even though the Hertzian granular chain is commonly denoted as
a ``sonic vacuum", we have shown that long granular chains
sustain periodic travelling waves on a full range of wavenumbers. 
We have proved the existence of a family of periodic travelling waves 
with wavenumbers $q$ close to $\pi$ and
profiles close to binary oscillations.
Using numerical continuation in $q$, we have been able to follow
this branch of solutions up to the long wave limit $q \approx 0$,
where they become close to a solitary wave inside small compression regions.

\vspace{1ex}

The waves we have obtained display unusual properties, due to the
fully nonlinear and unilateral character of Hertzian interactions. 
Each bead periodically undergoes a compression phase followed by a free flight,
a transition associated with a limited smoothness of the wave profile
(i.e. the corresponding solutions of (\ref{pbnl}) are $C^3$ but not $C^4$
at the onset of free flight). Moreover, 
below a critical value $q=q_k \approx 1.8$, the waves can be considered as
a periodic train of independent compactons separated by beads in free flight.
These numerical findings imply the existence of an isolated compacton
in granular chains, when beads are separated by equal gaps outside the compression wave, the
gaps depending on the velocity and width of the compacton.

An interesting open problem is to prove analytically the existence of
the periodic travelling waves obtained numerically, far from 
the limit of binary oscillations. 
For $q \approx 0$, this problem is equivalent
to the existence of a family of compactons close to the Hertzian solitary wave,
in the limit when the gaps between beads become small. 

\vspace{1ex}

From a dynamical point of view, 
below a critical wavenumber $q_c \approx 0.9$,
we have observed fast instabilities of the periodic travelling waves
leading to a disordered regime. This threshold was attained 
when the average number of adjacent interacting beads becomes $\geq 3$, or
equivalently when the maximal number of adjacent interacting beads becomes $\geq 4$.
The travelling waves appeared far more robust for $q > q_c$, i.e. 
when the number of adjacent interacting beads was always $\leq 3$. This can be 
intuitively understood through the results of \cite{sv}, 
where the travelling wave stability was numerically established for a three-ball chain
with fixed center of mass. However,
as soon as $q < q_c$, we have shown numerically that
the interactions of the three-ball packets with additional beads generate instabilities.
These instabilities persist up to $q \approx 0$, where they display a slower 
growth rate in the linear approximation.

As illustrated by figure \ref{shockbis}, a cumulative instability effect is present due to periodic boundary conditions,
since all perturbations left behind the compression pulse interact subsequently with it.
If the existence of isolated compactons can be mathematically established, 
it would be interesting to analyze their spectral stability and determine if the above
exponential instabilities persist. 
More generally, the spectral stability analysis of periodic travelling waves would deserve
more investigations. Above the critical value $q_c \approx 0.9$,
it would be interesting to determine for which wavenumbers
the waves are stable or slowly unstable. 
This question requires more refined numerical methods to resolve
very slow instabilities present in this parameter regime. 
In addition, unusual perturbations of Floquet eigenvalues may arise from 
the limited smoothness of Hertzian interactions. 

\vspace{1ex}

Another question concerns the generation
of stable periodic travelling waves 
in driven granular chains, in relation with
possible experimental realizations. In principle,
the travelling waves we have analyzed could be
generated in finite systems, provided the motions of
the first and last beads are imposed (following an exact travelling wave profile)
and starting from an exact initial condition. 
However these conditions are extremely restrictive. In addition,
dissipative effects should be also considered for practical applications. 
In this context, it would be interesting to study how to 
generate periodic travelling waves from
simpler initial conditions and driving signals 
in dissipative granular systems. 

\appendix

\section{\label{comprec}Compression solitary waves}

In this appendix we recall some classical properties of solitary waves
in granular chains which are used in section \ref{longw} for the analysis of
the long wave regime.

\vspace{1ex}

Consider travelling wave solutions of (\ref{nc}) taking the form
$x_n (t)=y(s)$, where $s=n-c\, t$. The function $y$ satisfies
\begin{equation}
\label{adresca}
c^2\, y^{\prime\prime} (s ) =
V^\prime (y(s +1)-y(s ))-V^\prime (y(s )-y(s -1)), \ \ \
s \in \mathbb{R},
\end{equation} 
where we recall that $V^\prime (x)= -|x|^{\alpha}\, H(-x)$, $\alpha >1$
and $H$ denotes the Heaviside function.
Up to rescaling $y$, one can fix $c=\pm 1$ in (\ref{adresca}) without loss of generality.
In that case, the renormalized relative displacements
$r(s )= y(s +\frac{1}{2} )-y(s-\frac{1}{2})$ satisfy
\begin{equation}
\label{adr}
r^{\prime\prime} (s ) =
V^\prime (r(s +1))-2V^\prime (r(s ))+V^\prime (r(s -1)), \ \ \
s \in \mathbb{R}.
\end{equation} 
There exists a negative solution of (\ref{adr}) satisfying
$\lim\limits_{s\rightarrow \pm\infty}{r(s)}=0$ and $r^\prime (0)=0$, 
which is known to decay super-exponentially \cite{friesecke,mackay,ji,english,stef}.
This solution corresponds to an exact solitary wave solution of (\ref{nc})
close to Nesterenko's approximate solution, with velocity equal to unity.
Figure \ref{ondesol} shows the solitary wave profile computed numerically, using the same numerical scheme
as for periodic waves (except we change the boundary conditions and
use an explicit approximation of $r$ derived
by Ahnert and Pikovsky \cite{ap} to initialize the Broyden method).
Our results agree with the ones of references \cite{ap,english,stef}, in particular
the solution we obtain is even in $s$. 

To recover $y$ from $r$, we note that the Poisson equation
$$
-y^{\prime\prime}(s)=f(s), \ \ \ s\in \mathbb{R},
$$
admits for all odd functions $f$ decaying exponentially at infinity
a unique odd and bounded solution given by
$y(s)=\int_{\mathbb{R}}{G(s,t)\, f(t)\, dt}$, where
$$G(s,t)=\frac{1}{2}\, H(st)\, (|t+s|-|t-s|).$$
Moreover, one has
\begin{equation}
\label{limites}
\lim_{s\rightarrow \pm\infty}{y(s)}=\pm \int_{0}^{+\infty}{t\, f(t) \, dt},
\end{equation}
the convergence being exponential. 
Equation (\ref{adresca}) with $c=1$ can be rewritten
$$
y^{\prime\prime} (s ) =
V^\prime (r(s+\frac{1}{2}))-
V^\prime (r(s-\frac{1}{2})),
$$
where the right side is odd in $s$ due to the evenness of $r$.
Consequently, we obtain a solution
of (\ref{adresca}) with $c=1$, given by
\begin{equation}
\label{depsol}
y(s)=\int_{\mathbb{R}}{G(s,t)\, [V^\prime (r(t-\frac{1}{2}))-V^\prime (r(t+\frac{1}{2}))]\, dt},
\end{equation}
and satisfying
\begin{equation}
\label{bcresca}
\lim_{s\rightarrow - \infty}{y(s)}= k_\alpha, \ \ \
\lim_{s\rightarrow + \infty}{y(s)}=- k_\alpha ,
\end{equation} 
where we have
\begin{equation}
\label{kal}
k_\alpha = -\int_{0}^{+\infty}{V^\prime (r(s)) \, ds}
\end{equation}
by virtue of (\ref{limites})
(this simplification is obtained using the evenness of $r$ and
elementary changes of variables in the integral).  
Note that a simpler formula can be derived for the
numerical computation of $y$, using the fact that
$$
y(s)=y(s-N-1)+\sum_{k=0}^{N}{r(s-k-\frac{1}{2})}.
$$
Letting $N\rightarrow +\infty$ and using the evenness of $r$ yields 
$$
y(s)=k_\alpha+\sum_{k=0}^{+\infty}{r(k+\frac{1}{2}-s)} .
$$
Then setting $s=0$ gives
\begin{equation}
\label{kal2}
k_\alpha =- \sum_{k=0}^{+\infty}{r(k+\frac{1}{2})},
\end{equation}
and consequently
\begin{equation}
\label{depsol2}
y(s)=\sum_{k=0}^{+\infty}{r(k+\frac{1}{2}-s)-r(k+\frac{1}{2})} .
\end{equation}
Using formula (\ref{kal2}) we numerically obtain $k_{3/2}\approx 1.3567$.

\begin{figure}[!h]
\psfrag{x}[0.9]{ $s$}
\psfrag{r}[1][Bl]{ $r(s )$}
\psfrag{u}[1][Bl]{ $y(s )$}
\begin{center}
\includegraphics[scale=0.4]{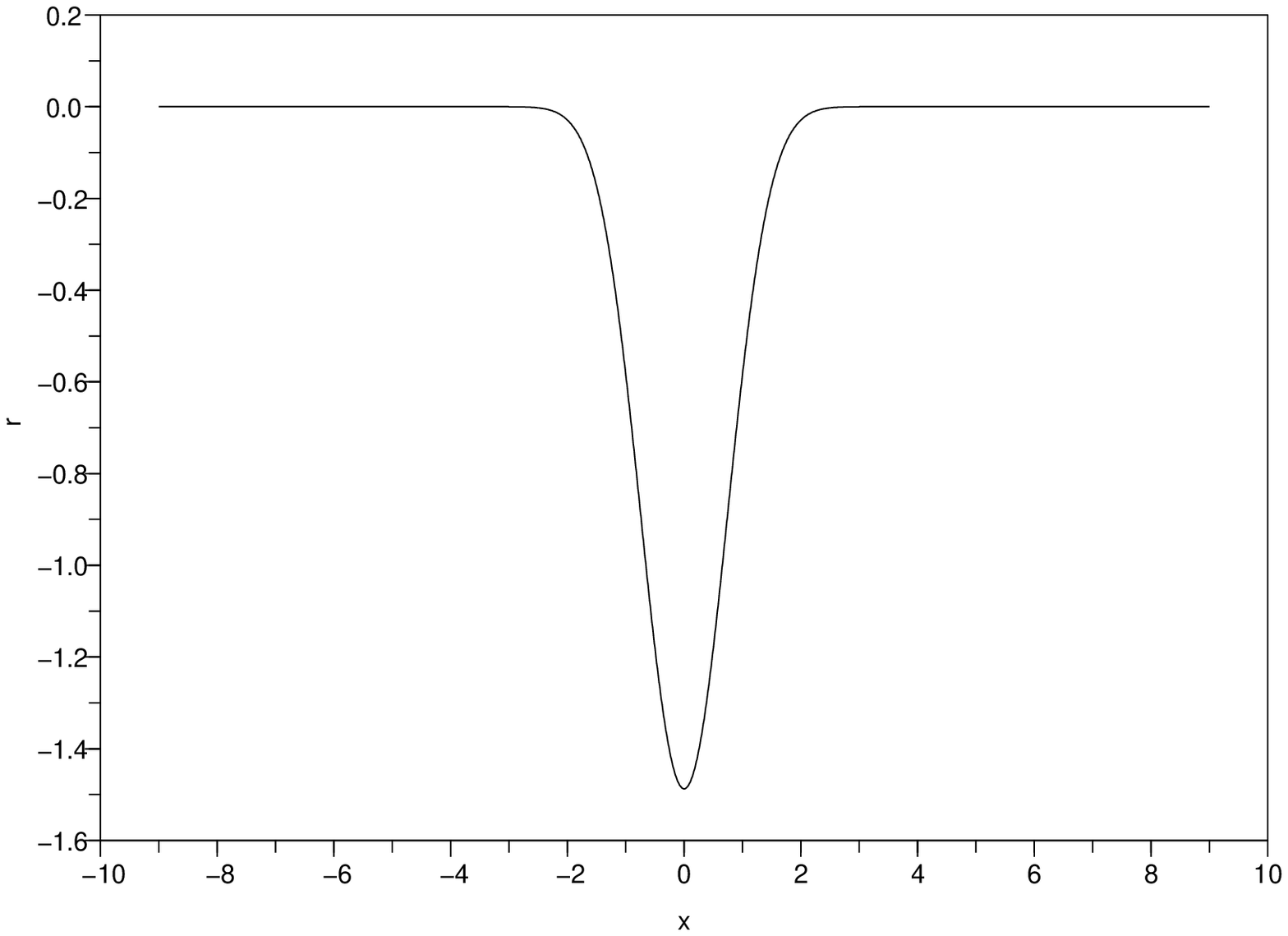}
\includegraphics[scale=0.4]{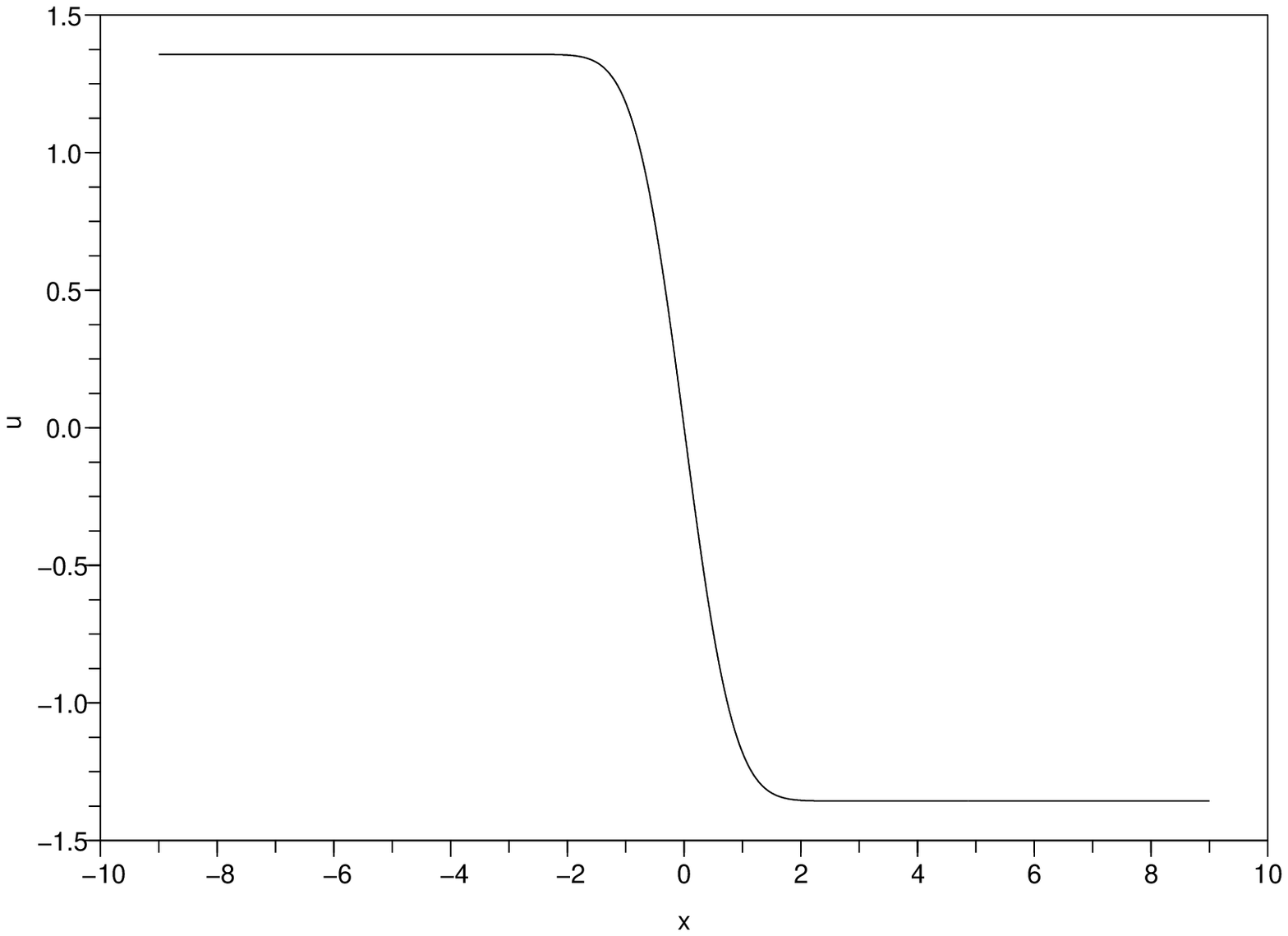}
\end{center}
\caption{\label{ondesol} 
Solitary wave solution of (\ref{adr}) 
computed numerically (upper plot) and corresponding bead displacement $y(s)$
solution of (\ref{adresca}) for $c=1$ (lower plot).
}
\end{figure}

\vspace{1ex}

\noindent
{\it Acknowledgements:} 
The author is grateful to anonymous referees for several suggestions
which essentially improved the paper, in particular for pointing out
the question of the existence of compactons. Helpful discussions with
J. Malick, B.~Brogliato, A.R. Champneys and P.G. Kevrekidis
are also acknowledged.


\begin{thebibliography}{00}
\bibitem{abram}
R. Abraham and J.E. Marsden, {\em Foundations of Mechanics}, Second Edition,
Addison-Wesley Publishing Company (1987). 
\bibitem{abra}
M. Abramowitz and I.A. Stegun, eds. 
{\em Handbook of Mathematical Functions},
National Bureau of Standards, 1964
(10th corrected printing, 1970), www.nr.com.
\bibitem{acary}
V. Acary and B. Brogliato. Concurrent multiple impacts modelling: Case
study of a 3-ball chain, {\em Proc. of the MIT Conference on Computational
Fluid and Solid Mechanics}, 2003 (K.J. Bathe, Ed.), Elsevier Science,
1836-1841.
\bibitem{ap}
K. Ahnert and A. Pikovsky. Compactons and chaos in strongly nonlinear lattices,
{\em Phys. Rev. E 79} (2009), 026209.
\bibitem{aubryC} 
S. Aubry and T. Cretegny. Mobility and reactivity of discrete breathers, {\em Physica D 119} (1998), 34-46.
\bibitem{boe}
N. Boechler, G. Theocharis, S. Job, P.G. Kevrekidis, M.A. Porter and C. Daraio.
Discrete breathers in one-dimensional diatomic granular crystals,
{\em Phys. Rev. Lett. 104} (2010), 244302.
\bibitem{cam}
D.K. Campbell et al, editors.
The Fermi-Pasta-Ulam problem : the first $50$ years,
{\em Chaos 15} (2005). 
\bibitem{ricardo}
R. Carretero-Gonz\'alez, D. Khatri, M.A. Porter, P.G. Kevrekidis and C. Daraio.
Dissipative solitary waves in granular crystals, {\em Phys. Rev. Lett. 102} (2009), 024102.
\bibitem{chat}
A. Chatterjee. Asymptotic solutions for solitary waves in a chain of elastic spheres, {\em Phys. Rev. E 59} (1999), 5912-5918.
\bibitem{chicone}
C. Chicone, {\em Ordinary differential equations with applications}, 
Texts in applied mathematics 34, Springer (1999).
\bibitem{num97}
T. Cretegny and S. Aubry.
{Spatially inhomogeneous time-periodic propagating waves in anharmonic systems},
{\em Phys. Rev. B 55} (1997), R11929-R11932.
\bibitem{dennis}
J.E. Dennis, Jr. and R.B. Schnabel, {\em Numerical methods for unconstrained optimization
and nonlinear equations}, SIAM Classics in Applied Mathematics 16,
SIAM (Society for Industrial and Applied Mathematics), 1996.
\bibitem{dreyer}
W. Dreyer, M. Herrmann and A. Mielke. 
Micro-macro transition in the atomic chain via Whitham's modulation equation, {\em Nonlinearity 19} (2006), 471-500.
\bibitem{dreyer2}
W. Dreyer and M. Herrmann. 
Numerical experiments on the modulation theory for the nonlinear atomic chain, {\em Physica D 237} (2008), 255-282.
\bibitem{english}
J.M. English and R.L. Pego. On the solitary wave pulse in a chain of beads,
{\em Proc. Amer. Math. Soc. 133}, n. 6 (2005), 1763-1768.
\bibitem{falcon}
E. Falcon. {\em Comportements dynamiques associ\'es au contact de Hertz : processus collectifs de collision et propagation d'ondes solitaires dans les milieux granulaires}, PhD thesis, Universit\'e Claude Bernard Lyon 1 (1997).
\bibitem{filip}
A.M. Filip and S. Venakides. Existence and modulation of traveling waves in particle chains, 
{\em Comm. Pure Appl. Math. 52} (1999), 693-735.
\bibitem{frater} 
F. Fraternali, M. A. Porter, and C. Daraio. 
Optimal design of composite granular protectors,
{\em Mech. Adv. Mat. Struct. 17} (2010), 1-19.
\bibitem{pego} 
{G. Friesecke and R.L. Pego}. {Solitary waves on FPU  
lattices : I. Qualitative properties, renormalization and continuum limit},  
{\em Nonlinearity 12} (1999), 1601-1627.
\bibitem{pego4} 
{G. Friesecke and R.L. Pego}. {Solitary waves on FPU  
lattices~: IV. Proof of stability at low energy}, 
{\em Nonlinearity 17} (2004), 229-251.
\bibitem{friesecke} 
G. Friesecke and J.A Wattis. Existence theorem for  
solitary waves on lattices, {\em Commun. Math. Phys. 161} (1994), 391-418.
\bibitem{fu}
G. Fu. An extension of Hertz's theory in contact mechanics,
{\em J. Appl. Mech. 74} (2007), 373-375.
\bibitem{gal}
G. Gallavotti, editor.
{\it The Fermi-Pasta-Ulam Problem. A Status Report},
{Lecture Notes in Physics 728} (2008), Springer.
\bibitem{herrmann}
M. Herrmann. Unimodal wave trains and solitons in convex FPU chains,
{\em Proc. Roy. Soc. Edinburgh A} 140
(2010), 753-785.
\bibitem{hinch}
E. J. Hinch and S. Saint-Jean. The fragmentation of a line of ball by an impact, {\em Proc. R. Soc. London, Ser. A 455} (1999), 3201-3220.
\bibitem{hoffman3}
A. Hoffman and C.E. Wayne. {\em A simple proof of the stability of solitary waves
in the Fermi-Pasta-Ulam model near the KdV limit} (2008), arXiv:0811.2406v1 [nlin.PS].
\bibitem{iooss} 
G. Iooss. Travelling waves in the Fermi-Pasta-Ulam lattice,
{\em Nonlinearity 13} (2000), 849-866.
\bibitem{jamesc}
G. James. Nonlinear waves in Newton's cradle and the discrete $p$-Schr\"odinger equation
(2010), arXiv:1008.1153v1 [nlin.PS]. To appear in {\em Math. Mod. Meth. Appl. Sci.},
DOI No: 10.1142/S0218202511005763. 
\bibitem{jamesc2}
G. James, P.G. Kevrekidis and J. Cuevas. Breathers in oscillator chains with Hertzian interactions
(2011), arXiv:1111.1857v1 [nlin.PS].
\bibitem{ji}
J.-Y. Ji and J. Hong. Existence criterion of solitary waves in a chain of grains,
{\em Phys. Lett. A 260} (1999), 60-61.
\bibitem{johnsonbook}
K.L. Johnson. {\em Contact mechanics}, Cambridge Univ. Press, 1985. 
\bibitem{johnson}
P.A. Johnson and X. Jia. Nonlinear dynamics, granular media and dynamic earthquake triggering, {\em Nature 437} (2005), 871-874.
\bibitem{liu1}
C. Liu, Z. Zhao and B. Brogliato. Frictionless multiple impacts in multibody
systems. I. Theoretical framework. {\em Proc. R. Soc. A-Math. Phys. Eng. Sci., 464} (2008), 3193-3211.
\bibitem{liu2}
C. Liu, Z. Zhao and B. Brogliato. Frictionless multiple impacts in multibody
systems. II. Numerical algorithm and simulation results, {\em Proc. R. Soc. A-Math. Phys. Eng. Sci., 465} (2009), 1-23.
\bibitem{ma}
W. Ma, C. Liu, B. Chen and L. Huang. 
Theoretical model for the pulse dynamics in a long granular chain,
{\em Phys. Rev. E 74} (2006), 046602.
\bibitem{mackay}
R.S. MacKay. Solitary waves in a chain of beads under Hertz contact,
{\em Phys. Lett. A 251} (1999), 191-192.
\bibitem{neste1}
V.F. Nesterenko. Propagation of nonlinear compression pulses in granular media,
{\em J. Appl. Mech. Tech. Phys. 24} (1983), 733-743.
\bibitem{neste2}
V.F. Nesterenko, {\em Dynamics of heterogeneous materials}, Springer Verlag, 2001. 
\bibitem{pankovbook}
A. Pankov. {\em Travelling waves and periodic oscillations in Fermi-Pasta-Ulam lattices}, Imperial
College Press, London, 2005.
\bibitem{porter}
M. Porter, C. Daraio, I. Szelengowicz, E.B. Herbold and P.G. Kevrekidis.
{Highly nonlinear solitary waves in heterogeneous periodic granular media},
{\em Physica D 238} (2009), 666-676.
\bibitem{rh}
P. Rosenau and J.M. Hyman. Compactons: solitons with finite wavelength, {\em Phys. Rev. Lett. 70} (1993), 564.
\bibitem{ros}
P. Rosenau and S. Schochet. Compact and almost compact breathers:
a bridge between an anharmonic lattice and its continuum limit,
{\em Chaos 15} (2005), 015111. 
\bibitem{schm}
J. Schmittbuhl, J.-P. Vilotte and S. Roux. Propagative macrodislocation modes in an earthquake fault model,
{\em Europhys. Lett. 21} (1993), 375-380.
\bibitem{sen}
S. Sen, J. Hong, J. Bang, E. Avalos and R. Doney. Solitary waves in the granular chain,
{\em Physics Reports 462} (2008), 21-66.
\bibitem{senm}
S. Sen, M. Manciu and J.D. Wright. 
Soliton-like pulses in perturbed and driven Hertzian chains and their possible
applications in detecting buried impurities, {\em Phys. Rev. E 57} (1998), 2386-2397.
\bibitem{MacSep}
J.-A. Sepulchre and R.S. MacKay. Localized oscillations in conservative or dissipative networks of weakly coupled autonomous oscillators, {\em Nonlinearity 10} (1997), 679-713 .
\bibitem{sv}
Y. Starosvetsky and A.F. Vakakis. Traveling waves and localized modes in one-dimensional
homogeneous granular chains with no precompression, {\em Phys. Rev. E 82} (2010), 026603. 
\bibitem{stef}
A. Stefanov and P.G. Kevrekidis. On the existence of solitary traveling waves for generalized Hertzian chains,
{\em J. Nonlinear Sci.} (2012), DOI: 10.1007/s00332-011-9119-9.
\end{thebibliography}
\end{document}